\documentclass[11pt]{article}
\usepackage{times}
\usepackage{amssymb}
\usepackage{amsmath}
\usepackage{xspace,multirow,comment}
\usepackage[shortlabels]{enumitem}
\usepackage{calc}
\usepackage{amsthm}
\usepackage[ruled, linesnumbered, vlined, nokwfunc]{algorithm2e}
\usepackage{caption,makecell}
\usepackage[dvipsnames]{xcolor}

\newif\ifdvi
\ifdvi
\else
\usepackage[backref=page]{hyperref}
\hypersetup{
	bookmarks=true,
	bookmarksnumbered=true,
	bookmarksopen=true
}
\fi

\makeatletter
\def\@citex[#1]#2{\leavevmode
  \let\@citea\@empty
  \@cite{\@for\@citeb:=#2\do
    {\@citea\def\@citea{,\penalty\@m\ }%
\edef\magic##1{\let##1\expandafter\noexpand\csname bibalias@\@citeb\endcsname}%
\magic\tmp \ifx\tmp\relax\else \let\@citeb\tmp\fi
     \edef\@citeb{\expandafter\@firstofone\@citeb\@empty}%
     \if@filesw\immediate\write\@auxout{\string\citation{\@citeb}}\fi
     \@ifundefined{b@\@citeb}{\hbox{\reset@font\bfseries ?}%
       \G@refundefinedtrue
       \@latex@warning
         {Citation `\@citeb' on page \thepage \space undefined}}%
       {\@cite@ofmt{\csname b@\@citeb\endcsname}}}}{#1}}
\def\bibalias#1#2{\expandafter\def\csname bibalias@#1\endcsname{#2}}
\makeatother

\newcommand{\np}{{\em NP}\xspace}
\newcommand{\nphard}{\np-hard\xspace}

\DeclareMathOperator{\supp}{supp}
\DeclareMathOperator{\argmin}{argmin}
\DeclareMathOperator{\argmax}{argmax}
\newcommand{\Ex}{\ensuremath{\mathbb{E}}}
\newcommand{\E}[2][]{\ensuremath{\mathbb{E}_{#1}\bigl[#2\bigr]}}

\DeclareMathOperator{\demd}{demd}

\newenvironment{proofof}[1]{\begin{proof}[Proof of {#1}]}{\end{proof}}
\newenvironment{proofsketch}{\begin{proof}[Proof sketch]}{\end{proof}}

\newtheorem{theorem}{Theorem}[section]
\newtheorem{lemma}[theorem]{Lemma}

\newtheorem{claim}[theorem]{Claim}

\newtheorem{corollary}[theorem]{Corollary}

{\theoremstyle{definition} 
\newtheorem{definition}[theorem]{Definition}

\newtheorem{remark}[theorem]{Remark}}

\newcommand{\R}{\ensuremath{\mathbb R}}
\newcommand{\Z}{\ensuremath{\mathbb Z}}

\newcommand{\C}{\ensuremath{\mathcal{C}}}

\newcommand{\F}{\ensuremath{\mathcal F}}
\newcommand{\D}{\ensuremath{\mathcal D}}

\newcommand{\M}{\ensuremath{\mathcal M}}
\newcommand{\Nc}{\ensuremath{\mathcal N}}

\newcommand{\Sc}{\ensuremath{\mathcal S}}
\newcommand{\Pc}{\ensuremath{\mathcal P}}

\newcommand{\opt}{\ensuremath{\mathsf{opt}}}
\newcommand{\optset}{\ensuremath{O^*}}
\newcommand{\boptset}{\ensuremath{\overline O}}
\newcommand{\OPT}{\ensuremath{\mathit{OPT}}}
\newcommand{\cost}{\ensuremath{\mathit{cost}}}

\newcommand{\frall}{\ensuremath{\text{ for all }}}

\newcommand{\es}{\ensuremath{\emptyset}}

\newcommand{\ceil}[1]{\ensuremath{\left\lceil#1\right\rceil}}
\newcommand{\floor}[1]{\ensuremath{\left\lfloor#1\right\rfloor}}

\newcommand{\e}{\ensuremath{\epsilon}}

\newcommand{\gm}{\ensuremath{\gamma}}
\newcommand{\Gm}{\ensuremath{\Gamma}}

\newcommand{\sse}{\subseteq}

\newcommand{\ti}{\ensuremath{\widetilde i}}
\newcommand{\tj}{\ensuremath{\widetilde j}}

\newcommand{\tv}{\ensuremath{\tilde v}}

\newcommand{\hi}{\ensuremath{\widehat i}}
\newcommand{\hS}{\ensuremath{\widehat S}}

\newcommand{\bx}{\ensuremath{\overline x}}

\newcommand{\ld}{\ensuremath{\lambda}}
\newcommand{\kp}{\ensuremath{\kappa}}
\newcommand{\al}{\ensuremath{\alpha}}
\newcommand{\tht}{\ensuremath{\theta}}

\newcommand{\Dt}{\ensuremath{\Delta}}

\newcommand{\Om}{\ensuremath{\Omega}}

\newcommand{\ve}{\ensuremath{\varepsilon}}

\newcommand{\bc}{\ensuremath{\overline c}}

\newcommand{\wt}{\ensuremath{\mathsf{wt}}}

\newcommand{\infeas}{``\ensuremath{\mathsf{infeasible}}''\xspace}

\newcommand{\gset}{\ensuremath{U}}
\newcommand{\targ}{\ensuremath{\mathit{Val}}}

\newcommand{\lpopt}{\ensuremath{\mathsf{LP}^*}}

\newcommand{\mech}{\ensuremath{\M}}
\newcommand{\mechone}{\ensuremath{\mech^{(1)}}}
\newcommand{\mechtwo}{\ensuremath{\mech^{(2)}}}
\newcommand{\flag}{\ensuremath{\mathsf{flag}}}
\newcommand{\util}{\ensuremath{u}}

\newcommand{\optlpname}{BFLP}
\newcommand{\optlp}[1][{}]{\ensuremath{\text{(\optlpname}{#1}\text{)}}\xspace}

\newcommand{\frcoverlp}[1][S]{\ensuremath{\text{(PFC}{(#1)}\text{)}}\xspace}
\newcommand{\optfrcover}[1][S]{\ensuremath{\mathsf{LP^*_{PFC}}({#1})}\xspace}

\newcommand{\OPTAlg}{\ensuremath{\OPT_{\mathsf{Alg}}}\xspace}
\newcommand{\OPTalg}{\OPTAlg}
\newcommand{\optalg}{\OPTalg}
\newcommand{\optbench}{\ensuremath{\OPT_{\mathsf{Bench}}}\xspace}

\newcommand{\optparam}{\ensuremath{\OPT_{\mathsf{Param}}}\xspace}
\newcommand{\bench}{\ensuremath{\mathsf{Bmark}}\xspace}
\newcommand{\vbench}[1][v]{\ensuremath{{#1}_{-1}}}
\newcommand{\vbengen}[2][v]{\ensuremath{{#1}_{-{#2}}}}
\newcommand{\pset}{\ensuremath{G}}
\newcommand{\pl}{\ensuremath{\Nc}}
\newcommand{\vemax}{\ensuremath{v_{\max}}}
\newcommand{\bnew}{\ensuremath{B^{\mathsf{new}}}}
\newcommand{\ecost}[1][{}]{\ensuremath{\E{\cost_{#1}}}}
\newcommand{\priv}{\ensuremath{Z}}
\newcommand{\vxos}[1][v]{\ensuremath{{#1}^{\mathsf{prxos}}}}
\newcommand{\gprxos}{\vxos[g]}

\newcommand{\psmax}{\zeta}
\newcommand{\sduniv}{\Psi}

\newcommand{\est}{\ensuremath{\mathsf{est}}\xspace}

\newcommand{\los}{\ensuremath{\mathsf{LOS}}\xspace}

\newcommand{\xosdemdapx}{\ensuremath{64}}
\newcommand{\xosgenapx}{\ensuremath{80}}
\newcommand{\xossupaddpolyapx}{\ensuremath{100}}
\newcommand{\submodapx}{\ensuremath{48}}
\newcommand{\additiveapx}{\ensuremath{272}}
\newcommand{\xosgenapxobid}{\ensuremath{145}}

\SetAlgoProcName{Algorithm}{Algorithm}
\SetFuncSty{textsc}
\SetAlCapNameSty{textsc}
\SetKwComment{simpc}{// }{}
\SetCommentSty{textnormal}
\DontPrintSemicolon

\let\oldnl\nl
\newcommand{\nonl}{\renewcommand{\nl}{\let\nl\oldnl}}

\let\oldthealgocf\thealgocf

\newcommand{\swamy}[1]{\begingroup {#1} \endgroup}

\title{Multidimensional Budget-Feasible Mechanism Design} 
\author{
    Rian Neogi\thanks{\texttt{\{rneogi,kpashkovich,cswamy\}@uwaterloo.ca}.
    Dept. of Combinatorics and Optimization, Univ. Waterloo, Waterloo, ON N2L 3G1.
    Supported in part by the NSERC Discovery grants of K. Pashkovich and C. Swamy.}
\and
\addtocounter{footnote}{-1}
Kanstantsin Pashkovich\footnotemark
\and
\addtocounter{footnote}{-1}
Chaitanya Swamy\footnotemark
}
\date{}

\begin{document}

\maketitle
\def\thepage{}
\thispagestyle{empty}

\bibalias{HajiaghayiKS18}{Hajiaghayi2018FrugalAD}
\bibalias{ArcherT07}{archer2007frugal}
\bibalias{ElkindGG07}{elkind2007frugality}
\bibalias{KarlinKT05}{karlin2005beyond}
\bibalias{LiptonMMS04}{Lipton2004OnAF}
\bibalias{CaragiannisKMPS16}{Caragiannis2016TheUF}

\begin{abstract}
In {\em budget-feasible mechanism design}, a buyer wishes to procure a set of items of
maximum value 
from self-interested rational players. We have a nonnegative valuation function 
$v:2^\gset\mapsto\R_+$, where $U$ is the set of all items, 
where $v(S)$ specifies the value obtained from set $S$ of items. 
The entirety of current work on budget-feasible mechanisms has focused on the
{\em single-dimensional} setting, wherein each player holds a {\em single} item $e$ and incurs
a {\em private cost} $c_e$ for supplying item $e$ (and each item is held by some player).

We introduce {\em multidimensional budget feasible mechanism design}: 
the universe $U$ is now partitioned into item-sets $\{\pset_i\}$ held by the different
players, and each player $i$ incurs a private cost $c_i(S_i)$ for supplying the set
$S_i\sse\pset_i$ of items. 
A {\em budget-feasible mechanism} is a mechanism (i.e., an algorithm and a payment scheme) 
that is {\em truthful}, i.e., where players are incentivized to report their true
costs, and where the total payment made to the players is at most some given budget $B$. 
The goal (as in the single-dimensional setting) is to devise a budget-feasible mechanism
that procures a set of items of large value. 

{\em We obtain the first approximation guarantees for multidimensional budget feasible 
mechanism design.} 

Our contributions are threefold. First, we prove an impossibility result showing that the
standard benchmark used in single-dimensional budget-feasible mechanism design, namely the
algorithmic optimum $\optalg$ is inadequate in that no budget-feasible mechanism can
achieve good approximation relative to this.
We identify that the chief underlying issue here is that there could be a monopolist, i.e.,
a single player who contributes a large fraction of $\optalg$, which prevents a
budget-feasible mechanism from obtaining good guarantees.
Second, we devise an alternate benchmark, $\optbench$, 
that allows for meaningful approximation guarantees, 
thereby yielding a metric for comparing mechanisms. 
Third, we devise budget-feasible mechanisms that achieve 
{\em constant-factor approximation guarantees} with respect to this benchmark for XOS
valuations. Our most general results pertain to XOS valuations and arbitrary cost
functions, where we obtain a universally budget-feasible mechanism, and a
budget-feasible-in-expectation mechanism that runs in polytime given a demand oracle. 
We also obtain polytime universally budget-feasible mechanisms for: 
(a) additive valuations and additive costs; and
(b) XOS valuations and superadditive cost functions, assuming access to a variant of
a demand oracle. 

Our guarantees for XOS valuations also yield an $O(\log k)$-approximation for subadditive
valuations (with respect to our benchmark), where $k$ is the number of players.
\end{abstract}

\newpage
\pagenumbering{arabic} \normalsize

\section{Introduction} \label{intro}
In {\em budget-feasible mechanism design}, a buyer wishes to procure a set of items of
maximum value 
from self-interested rational players. We have a ground set $\gset$ of items, a
{\em valuation function} $v:2^\gset\mapsto\R_+$ satisfying $v(\es)=0$, where 
$v(S)$ specifies the value obtained from set $S$ of items, and a budget $B$.
The {\em vast majority} of work on budget-feasible mechanisms has focused on the
{\em single-dimensional setting}, wherein each player holds a {\em single, distinct} 
element $e$ 
and incurs a {\em private cost} $c_e$ for supplying item $e$.
In order to incentivize players to reveal their true costs, the buyer needs to make
suitable payments to the players.   
The utility of a player is then equal to the (payment received by it) $-$ (cost incurred
by it). A mechanism, i.e., an algorithm and a payment scheme, is: 
(a) {\em truthful}, if each player maximizes her utility by revealing her true cost and
thereby has no incentive to misreport her cost; and 
(b) {\em individually-rational}, if the utility of every truthful player is nonnegative.
A {\em budget-feasible mechanism} is a truthful, individually rational (IR) mechanism 
where the total payment made to the players is at most the given budget $B$. 
The goal is to devise a budget-feasible mechanism that procures a set of items of large
value, where we measure the quality of the mechanism by comparing its value returned
against the {\em algorithmic optimum} 
$\OPTalg:=\OPTalg(v,B,c):=\max\,\{v(S): \sum_{e\in S}c_e\leq B,\ S\sse\gset\}$, which
denotes the maximum obtainable value if the costs were public information. 

Budget feasible mechanisms were introduced by Singer~\cite{Singer10} and have  been
extensively studied (see, e.g.,~\cite{BeiCGL12,Singer13,LeonardiMSZ17,BalkanskiGGST22} and
the references therein), but (as noted earlier) 
{\em almost exclusively} in the (above) single-dimensional setting. A few works 
consider a somewhat richer, albeit {\em still single-dimensional} setting called the
level-of-service (\los) setting, wherein a player can provide multiple units of
an item (or levels of service), incurring the same cost for each unit
supplied~\cite{ChanC14,AnariGN18,KlumperS22,AmanatidisKMST23}.  
This focus on single-dimensional settings is in stark contrast with
other prominent mechanism-design problem domains, such as social-welfare maximization and
profit-maximization in combinatorial auctions (see~\cite{LaviAGT,HartlineKAGT}), or
cost-sharing mechanism design~\cite{GeorgiouS19,DobzinskiO17}, which considers richer
multidimensional settings involving more-expressive players (notwithstanding the
significant challenges that arise in multidimensional mechanism design).

\vspace*{-1ex}
\subsection{Our contributions and results}
We {\em initiate the study of multidimensional budget feasible mechanism design}: each 
player $i$ now owns a publicly-known set $\pset_i$ of items, where these player-sets
partition $\gset$ 
and incurs cost $c_i(S_i)$ for supplying the set $S_i\sse\pset_i$ of items, 
where $c_i:2^{\pset_i}\mapsto\R_+$ is $i$'s {\em private cost function}.
(Throughout, we use $i$ to index the players, and $e$ to index items in $\gset$.)
We assume only that the $c_i$s are monotone ($c_i(S)\leq c_i(T)$ for all
$S,T\sse\pset_i$), and normalized ($c_i(\es)=0$). 
The goal (as before) is to design a budget-feasible mechanism that maximizes the value of 
the procured set of items.
\footnote{Note that this model is rich enough to  
to capture even a {\em multidimensional \los} setting, wherein $c_i(S_i)$ is a monotone 
(but not necessarily linear) function of $|S_i|$.}
Let $\C_i$ denote the set of all possible player-$i$ private cost functions, and
$\C:=\Pi_i\C_i$. 
Let $n=|\gset|$.
Note that $v$, $B$, and $\{(\pset_i,\C_i)\}_i$ are public knowledge.

{\em We obtain the first approximation guarantees for multidimensional budget feasible 
mechanism design.} 

Multidimensional mechanism design is in general substantially more challenging than
single-dimensional mechanism design, mostly due to the fact that there is no simple and
conveniently-leverageable characterization of truthfulness that is analogous
to Myerson's monotonicity-based characterization of truthfulness in single-dimensional
settings.
With budget-feasible mechanism design, where payments also feature in the constraints,
this challenge also manifests itself in another distinct (but related) way, 
namely, that there is a complicated relationship between truthfulness-inducing payments
and the underlying algorithm 
\footnote{
In the multidimensional setting, there are characterizations of truthfulness based on
cycle monotonicity and weak-monotonicity, and payments can be obtained by computing
shortest paths in a certain graph, but these have been difficult to leverage.}
compared to the single-dimensional setting, where payments
correspond to threshold values.

Furthermore, multidimensional budget-feasible mechanism design poses an (orthogonal)
modeling challenge, namely,   
the benchmark used in the single-dimensional setting, the algorithmic optimum, which
now translates to   
$\OPTalg(v,B,c):=\max\,\{v(S): \sum_i c(S\cap\pset_i)\leq B,\ S\sse\gset\}$, 
turns out to be quite inadequate: 
{\em no budget-feasible mechanism can achieve any non-trivial approximation with respect
to $\OPTalg$}, even for additive valuations and additive cost functions. 
This impossibility result (see Theorem~\ref{intro-thm}) extends to 
{\em budget-feasible-in-expectation} mechanisms, 
which are {\em truthful-in-expectation} mechanisms---i.e., truthful reporting maximizes
the expected utility of each player---
where the expected total payment is at most the budget $B$ (and IR holds with probability
$1$).  
It also extends to the {\em Bayesian} setting, wherein there is a prior distribution
from which players' types are drawn, and we compare the expected value of the mechanism
and $\E{\OPTalg}$.

We say that a budget-feasible mechanism achieves approximation ratio $\al$
with respect to a benchmark $\bench$, where $\al\geq 1$,
if it always obtains value, or expected value in case of a randomized mechanism, at least 
$\bench(v,B,c)/\al$.
We sometimes consider a natural ``no-overbidding'' assumption, 
which states that every item, by itself, constitutes a feasible solution:
formally, for every $i$, every $c\in\C_i$, and every $e\in\pset_i$, we have   
$c_i(\{e\})\leq B$. 
This is without loss of generality in the single-dimensional (single-item and \los)
settings (as dropping a player $i$ with $c_i>B$ does not impact truthfulness or
approximation), but does impose a restriction on player cost functions in the
multidimensional setting.
\footnote{
In particular, with overbidding, one can capture a richer space of private 
inputs, namely the ``unknown'' setting where there is a {\em private} subset
$\priv_i\sse\pset_i$ of items that a player $i$ can provide; with overbidding, observe
that this can be easily encoded by player $i$ reporting (a monotone) $c_i$ with
$c_i(\{e\})>B$ for all $e\in\pset_i-\priv_i$.
}

\begin{theorem}[{\bf Impossibility results:} Informal versions of
    Theorems~\ref{det-overbid}--\ref{detlb}]  
\label{intro-thm} \label{infthm}
(Recall that $n=|U|$.)
\begin{enumerate}[label=(\alph*), topsep=0ex, noitemsep, leftmargin=*]
\item 
No deterministic budget-feasible mechanism can achieve any bounded approximation 
ratio relative to $\OPTalg(v,B,c)$.

\item 
No budget-feasible-in-expectation mechanism 
can achieve approximation ratio better than $n$ relative to $\OPTalg(v,B,c)$.  

\item (\cite{ChanC14}) Under no-overbidding, even for the single-dimensional \los
setting, no deterministic budget-feasible
mechanism can achieve approximation ratio better than $n$, and no
budget-feasible-in-expectation mechanism can achieve approximation ratio better than 
$O(\log n)$, relative to $\OPTalg(v,B,c)$.
\end{enumerate}
These impossibility results hold even when the valuation $v$ and the
$c_i$s are additive functions. 
The lower bounds for budget-feasible-in-expectation mechanisms extend to Bayesian 
budget-feasible mechanisms. 
\end{theorem}

(We include part (c) above mainly for the sake of comparison (when we do not assume
no-overbidding) and completeness. 
\footnote{In~\cite{AmanatidisKMST23}, this lower bound is bypassed in the \los setting by
making a strong ``all-in'' assumption. Under such a strong assumption, one can obtain
$O(1)$-approximation relative to $\optalg$ even in our multidimensional settings; see
Remark~\ref{allin}.}) 
We remark that the lower bounds mentioned above are {\em tight} for XOS valuations (see
Section~\ref{impos-tight}); this 
follows as a by-product of some of our results, and we elaborate upon this later.

\vspace*{-2ex}
\paragraph{Suitable benchmark.}
In light of the above impossibility results, a pertinent question that arises is: 
{\em can one come up with alternative suitable benchmarks that enable one to 
circumvent the above impossibility result and obtain meaningful approximation guarantees}?   
One of our contributions is to define a novel benchmark that answers this question
{\em affirmatively}. 

Before delving into our benchmark, we discuss some considerations to keep in mind when
coming up with an appropriate benchmark.
We first note that the natural idea 
of comparing, for every input $(v,B,c)$, against the maximum value obtainable by a
budget-feasible mechanism for that input, fails.   
This is because for every input $(v,B,c)$, one can always tailor a budget-feasible
mechanism that obtains value $\optalg(v,B,c)$ on this particular input; 
\footnote{Let $S^*\sse\gset$ be such that $v(S^*)=\optalg(v,B,c)$ and
$S^*_i=S^*\cap\pset_i$. Consider the mechanism that on input $d$, returns
$\bigcup_i T_i$, where $T_i=S^*_i$ if $d(S^*_i)\leq c(S^*_i)$ and $\es$ otherwise, and
pays player $i$, $c(S^*_i)$, if $d(S^*_i)\leq c(S^*_i)$ and $0$ otherwise.} 
so this maximum-value benchmark 
coincides with $\optalg$. 
Coupled with the above impossibility results, this also implies
that there is no (pointwise) ``best'' budget-feasible mechanism: for every budget-feasible 
mechanism $\mech$, we can find some input $(v,B,c)$ on which $\mech$ performs rather
poorly (obtaining value at most $\optalg(v,B,c)/n$), but for which some other
budget-feasible mechanism $\mech'$ fares much better (obtaining value $\optalg(v,B,c)$).
A further consequence of this 
to appreciate is that if we are to circumvent the
impossibility results and come up with meaningful approximation guarantees relative to
some benchmark $\bench$, then it must be that $\bench(v,B,c)\ll\optalg(v,B,c)$ for some
input $(v,B,c)$, and therefore $\bench$ cannot be pointwise-close to an upper bound on the
maximum value obtainable from a budget-feasible mechanism; in fact, on some inputs, it must
{\em necessarily} be {significantly smaller} than the maximum value achievable for that
input by budget-feasible mechanisms. Consequently, if one seeks a suitable benchmark, then 
one must necessarily foresake the a priori natural property that the benchmark be 
(close to) an upper bound on the maximum value achievable in the desired space of
solutions. 
Thus, some care and insight is needed to define a suitable benchmark. 

Roughly speaking, the impossibility results stem from the fact that there could be a
single player $i$ responsible for a large fraction of the total value, who can
then act as a monopolist: a budget-feasible mechanism can be forced to spend the entire
budget on player $i$, even if only one item from $\pset_i$ is chosen, which leads to a
poor approximation factor. 
There are two approaches for bypassing this issue: 
(1) restrict attention to inputs where no player is a monopolist (and still compare
against $\OPTalg$); or 
(2) come up with a new benchmark whose definition captures that there is no monopolist.  
We adopt the latter approach, as it has the benefit that it yields guarantees for 
{\em all} inputs. 
Given a vector $c=(c_1,\ldots,c_k)$ of player cost functions,
define $c(S):=\sum_i c_i(S\cap\pset_i)$.
Define 
\[
\swamy{\optbench(v,B,c) \ :=\
\max_{S\sse\gset}\,\Bigl\{\min_{i\in[k]}v(S-\pset_i):\ c(S)\leq B\Bigr\}.}
\]
This benchmark safeguards against a monopolist because we consider the value after
excluding the contribution from any single player; 
so its objective function 
\swamy{$\min_i v(S-\pset_i)$}
degrades in the presence of a monopolist.
\footnote{
\swamy{For subadditive $v$, 
the objective function of $\optbench$ is always larger than the alternate objective
$v(S)-\max_i v(S\cap\pset_i)$.  
In the alternate objective, note that it is important to remove the {\em maximum}
contribution of any player towards $v(S)$; 
if we instead subtract the {\em average}
contribution of the players and consider the objective 
$v(S)-\bigl(\sum_i v(S\cap\pset_i)\bigr)/k'$,
where $k'=|\{i: S\cap\pset_i\neq\es\}|$, then the benchmark again becomes too strong. 
This is because one can pad an instance with a monopolist with dummy players that
contribute little value and incur $0$ cost, thereby increasing $k'$ and driving the
subtracted term to $0$.}}
An appealing aspect of $\optbench$ is that it is parameter-free. 
Observe also that \swamy{for subadditive $v$ (i.e., $v(S\cup T)\leq v(S)+v(T)$ for all
$S,T\sse\gset$)}, on inputs where there is no monopolist, 
$\optbench$ 
is close to $\OPTalg$: if we have some $S^*\sse\gset$ with $v(S^*)=\OPTalg(v,B,c)$
satisfying $v(S^*\cap\pset_i)\leq\ve\cdot v(S^*)$ for all $i\in[k]$, then
$\optbench\geq(1-\ve)\optalg$. 

\swamy{We discuss this benchmark further in Section~\ref{relwork}, relating it also to
other mechanism-design domains where the need for considering suitable banchmarks
arises.} 

\vspace*{-1ex}
\paragraph{Approximation results.}
We devise budget-feasible mechanisms that achieve 
{\em constant-factor approximation guarantees with respect to the above 
$\optbench$ benchmark}.   
Unless otherwise stated, in the sequel, 
when we refer to the approximation factor of a mechanism, it is
always with respect to the $\optbench$ benchmark, and without assuming no-overbidding. (As  
noted earlier, overbidding provides increased modeling power, allowing us to also model
settings where player $i$'s set of items is private.) 

Note that since budget-feasibility imposes a
condition on the payments of the mechanism, 
even the {\em existence} of budget-feasible mechanisms, 
{\em bereft of computational concerns},
that achieve a ``good'' approximation with respect to the above benchmark is not
guaranteed.
\footnote{In fact, in the multidimensional setting, the existence of even just a truthful
(but not necessarily budget-feasible) mechanism 
that achieves a good
approximation relative to $\OPTalg$ is not guaranteed; 
see Theorem~\ref{notruthful}. This in contrast to the single-dimensional
setting, 
where the algorithm that returns an optimal solution is monotone and hence truthfully
implementable.}  
This situation applies also to the single-dimensional setting (where we seek
approximation with respect to $\optalg$), and 
therefore even the development of good budget-feasible mechanisms 
{\em setting aside} computational considerations has been the focus of a significant body
of work. For instance, Dobzinski et al.~\cite{DobzinskiPS11} explicitly raised the
question of whether there exists an $O(1)$-approximation budget-feasible mechanism for
subadditive valuations, which was answered affirmatively by Bei et al.~\cite{BeiCGL12},
albeit in a very non-constructive fashion, and only very recently an explicit, but
non-polytime, mechanism was obtained~\cite{NeogiPS24}; a polytime mechanism for subadditive
valuations remains elusive. 
Similarly, for multidimensional cost-sharing mechanism design, the work
of~\cite{DobzinskiO17} focuses on proving the existence of good cost-sharing mechanisms,
regardless of computational considerations. 

Our main results are for {\em XOS valuations} (and subclasses), and are summarized in
Table~\ref{approxtable}. 
A function $v$ is XOS if it is the maximum of a collection of additive functions.  

\begin{table}[ht!]
\small
\begin{tabular}{l|l|c|l|l} \hline
Valuations & Costs & 
\rule[-6pt]{0pt}{19pt}$O(1)$-approximation wrt. & Mechanism type & Oracle \\ \hline
\multirow{3}{*}[-10pt]{XOS} & \multirow{2}{*}[-2pt]{Arbitrary} & 
\rule[-5pt]{0pt}{18pt}$\optbench$ & BF-in-expectation (Theorem~\ref{xosdemdthm}) & 
Demand \\ \cline{3-5}

& & 
\rule[-6pt]{0pt}{19pt}{$\optbench$} &
Universally BF (Theorem~\ref{xosgenthm-overbid}) & * \\ \cline{2-5}

& Superadditive &
\rule[-5pt]{0pt}{18pt}$\optbench$ &
Universally BF (Theorem~\ref{xossupaddpolythm-overbid}) & 
\makecell*[l]{Constrained \\ demand} \\ \hline

{Additive} & 
Additive &
\rule[-5pt]{0pt}{18pt}$\optbench$ &
Universally BF (Theorem~\ref{additivethm-overbid}) & -- \\ \hline

Submodular & Arbitrary & 
\rule[-6pt]{0pt}{19pt}$\optbench-O\bigl(\max_i v(G_i)\bigr)$ &
Universally BF (Theorem~\ref{submodthm}) & Demand \\ \hline
\end{tabular}
\captionsetup{font=small}
\caption{Summary of our results for XOS valuations, and subclasses. 
BF stands for budget-feasible. 
BF-in-expectation denotes truthful-in-expectation with expected total payment at most $B$,
where IR holds with probability 1; universally BF means that truthfulness, IR, 
and ``total payment $\leq B$'', all hold with probability $1$.
All our results hold without assuming no-overbidding.
The value obtained is at least $\frac{1}{O(1)}$ times the quantity listed under ``$O(1)$
approximation wrt.''
Our mechanisms run in polytime given access to the specified oracle; the entry marked *
requires a more involved oracle. For additive valuations and costs, the
valuation and costs are explicitly given.
}   
\label{approxtable}
\end{table}

Our mechanisms access the valuation function $v$ via a suitably generalized form of demand 
oracle, and leverage the VCG mechanism (see Theorem~\ref{vcgthm}). 
(Recall that 
given $c\in\C:=\Pi_i\C_i$, we define $c(S):=\sum_i c_i(S\cap\pset_i)$ for
$S\sse\gset$.)
In the multidimensional setting, since player-cost functions are not 
necessarily additive, it is natural to extend the notion of a demand oracle to consider
general price functions: 
a demand oracle for class $\C=\Pi_i\C_i$, takes $q\in\C$, $\kp\in\R_+$ as input, 
and returns $\argmax_{S\sse\gset}\,\bigl(v(S)-\kp\cdot q(S)\bigr)$. 
Some of our mechanisms utilize a constrained demand oracle, where we 
are also 
given a cap $\targ$, 
and the oracle returns $\argmax\,\{v(S)-\kp\cdot c(S): S\sse\gset,\ v(S)\leq\targ\}$.
(We remark that we can also work with an oracle that returns
$\argmax_{S\sse\gset}\bigl(\min(v(S),\targ)-\kp\cdot c(S)\bigr)$.) 

Our most general results apply to XOS valuations and 
{\em arbitrary} (monotone, normalized) cost functions. 
We devise two types of mechanisms that both achieve $O(1)$-approximation (with respect to
$\optbench$). 
The first mechanism is {\em budget-feasible-in-expectation} 
and runs in polytime given a demand oracle (Theorem~\ref{xosdemdthm}). 
The second mechanism (Theorem~\ref{xosgenthm-overbid}) satisfies a much stronger
mechanism-design guarantee, {\em universal budget-feasibility}---
wherein truthfulness, IR, and
the budget constraint on total payment, hold with probability $1$---
but requires stronger oracle access. Note that this 
also demonstrates the existence of a good universally budget-feasible mechanism, even in
the most general setting of arbitrary cost functions. 

We next design {\em polytime} universally budget-feasible mechanisms for various special 
cases of interest (that achieve $O(1)$ approximation). 
For XOS valuations and {\em superadditive} cost functions, our mechanisms
run in polytime given a constrained demand oracle (Theorems~\ref{xossupaddpolythm}
and~\ref{xossupaddpolythm-overbid}).  

For {\em additive valuations and additive cost functions}, a constrained demand oracle 
corresponds to solving a knapsack problem, and we show that, 
we can (roughly speaking) solve related (scaled and rounded) knapsack problems optimally  
and thereby obtain polytime mechanisms (Theorems~\ref{additivethm}
and~\ref{additivethm-overbid}). (We do not require any oracle access here since the input
now explicitly specifies the additive valuation and costs.) 

In presenting our universally budget-feasible mechanisms, we 
present them first in the simpler setting where we assume no-overbidding
(Section~\ref{xos-bfuni}), as this helps to illustrate some of the main ideas, and the
arguments become much simpler. We then discuss how to drop this assumption in
Section~\ref{overbid}.     
(Recall that all the results in Table~\ref{approxtable} hold without assuming 
no-overbidding.)  

For the subclass of monotone {\em submodular valuations} and arbitrary cost functions, we
design a polytime universally budget-feasible mechanism using demand oracles, 
with a weaker approximation guarantee 
(Theorem~\ref{submodthm}). Interestingly, this mechanism is not VCG-based: we
do not perform a global VCG computation, but only perform VCG computations locally, for 
individual players (see Algorithm~\ref{submodalg-gencost}).  

Our mechanisms for XOS mechanisms also yield {\em $O(\log k)$-approximation guarantees for
subadditive valuations} (Section~\ref{subadditive}). Recall that $k$ is the number of
players. 
This essentially capitalizes on the idea that an XOS function can
pointwise approximate a subadditive function. However: 
(a) we cannot use this as a black-box since a pointwise approximation of $v$ by $\tv$ does
not imply that the benchmarks $\optbench(v,B,c)$ and $\optbench(\tv,B,c)$ are close to
each other; 
(b) the best pointwise-approximation of $v$ by an XOS function can incur an
$O(\log n)$-factor loss~\cite{BeiCGL12,BhawalkarR11}, as $v$ is defined over a universe of
$n$ items. 
Instead, we obtain our results by working with a pointwise $O(\log k)$-approximation of
$v$ that satisfies a property weaker than XOS, 
and utilizing this internally within our mechanisms in a suitable fashion.

\medskip
We obtain two other types of results as a byproduct of our techniques. 
First, we also obtain tight 
approximation guarantees using demand oracles for the algorithmic problem of approximating  
$\optalg(v,B,c)$ in polynomial time (Section~\ref{algresults}). 
Theorem~\ref{subaddalgthm} shows that when the $c_i$s are superadditive, 
we can obtain a polytime $(2+\ve)$-approximation, for any constant $\ve>0$, even for  
subadditive valuations (i.e., we have $v(S\cup T)\leq v(S)+v(T)$ for all $S,T\sse\gset$),
which is {\em tight}, even for additive cost functions~\cite{BadanidiyuruDO19}.

Second, we can show that the lower bounds mentioned in our impossibility result
(Theorem~\ref{intro-thm}) are {\em tight} for XOS valuations (see
Section~\ref{impos-tight}).  
For instance, we obtain randomized budget-feasible mechanisms that achieve 
$O(\log n)$-approximation relative to $\optalg(v,B,c)$ assuming no-overbidding
(Theorem~\ref{xos-optalg}). 
This is a consequence of the fact (see Section~\ref{overview}) that our mechanisms work
with an estimate $V_1$ of our benchmark, and one can isolate $\optalg$ within a
multiplicative factor of $n$ (in fact, within an $O(\max_i|G_i|)$-factor).  
Notably, this also implies substantially more-general results even for the 
single-dimensional \los setting mentioned at the start of the Introduction, which 
is the very special case where $c_i(S)=c_i|S|$ for $S\sse\pset_i$. For this setup (where
no-overbidding is without loss of generality), prior work yields approximation factors of
$O(\log n)$~\cite{ChanC14} and $O(\max_i|\pset_i|)$~\cite{AmanatidisKMST23}, assuming that 
$v$ is additive across players and concave for a single player. 
We obtain similar guarantees, but 
{\em we do not need additivity} across players; they apply to {\em any XOS valuation},
and to the much-richer 
{\em multidimensional \los setting}, where $c_i(S)$ is a monotone function (as opposed to
a linear function of $|S|$). 

\medskip
Our work opens up the area of multidimensional budget-feasible mechanism design as an
exciting well-motivated research direction. 
Our results show that, notwithstanding the substantial challenges faced in multidimensional 
mechanism design, one can obtain interesting guarantees in this domain. 
While we make substantial progress in this area, a variety of interesting questions
remain, 
and we hope that our work will stimulate further work in this area.

\subsection{Technical overview} \label{overview} \label{techoverview}
We give an overview of our techniques, highlighting the salient ideas underlying our
mechanisms and their analysis.
We use $v(e)$, $c_i(e)$, and $c(e)$ to denote the respective function value for the
singleton set $\{e\}$.
Broadly speaking, our mechanisms for XOS valuations capitalize on two main ideas. 
(A) We can obtain a good estimate of the target value to aim for via random partitioning;  
and
(B) we can leverage such an estimate to select a suitable subset via a suitable
demand-oracle computation, which amounts to a VCG computation.
These ideas have been shown to be useful for single-dimensional
budget-feasible mechanism design~\cite{BeiCGL12,NeogiPS24}, but the multidimensional
setting inevitably leads to various challenges, 
including a rather tricky 
technical difficulty that arises 
when we aim for guarantees without assuming no-overbidding. 

\begin{enumerate}[label={\bf (\Alph*)}, topsep=0.2ex, itemsep=0.1ex, leftmargin=0pt, widest=b, itemindent=*]
\item {\bf Random partitioning.}\ The idea here is to estimate the target value from a
suitable random subset of $\gset$, which is then discarded. 
In the single-dimensional setting, this has been utilized in prior work (see,
e.g.,~\cite{BeiCGL12,AmanatidisKS19,NeogiPS24}), 
where one randomly samples each
element in $\gset$ with probability $0.5$, and one can argue that this well-estimates
$\optalg$. In the single-dimensional setting, sampling players or elements amounts to the
same thing, 
but when players hold multiple items, the two versions differ, 
and one needs to be careful in selecting the right sampling procedure. 
To ensure truthfulness, we do not want any player $i$'s item-set 
$\pset_i$ to be straddled by the sample as then a player $i$ whose items are still in
``play'' after the sample is discarded can influence the computation on the sample.
Therefore, we pick a random subset $\pl_1$ of players by selecting each
player with probability $0.5$, and use the items owned by these players
$\gset_1:=\bigcup_{i\in\pl_1}\pset_i$ to estimate the target value to aim for 
and discard $\gset_1$. However, we
hit an immediate snag, namely that if a single player contributes a large fraction of
$\optalg$ (i.e., is a monopolist), then such a sample will be quite noisy and not
yield a good estimate.  

As shown by our impossibility result, this is a {\em real} issue that precludes good
approximation with respect to $\optalg$ (and not simply an artifact of random sampling). 
Crucially (and conveniently) however, when we move to the $\optbench$ benchmark, 
this no longer presents a difficulty, as this benchmark specifically
safeguards against a monopolist. 

Specifically, we can argue (see Lemma~\ref{rpartition} and
Corollary~\ref{rsample}) that this random sampling: 
(i) yields an $O(1)$-factor estimate $V_1$ of $\optbench$, i.e., $V_1=\Omega(\optbench)$
with constant probability; and 
(ii) the set $\gset_2=\gset-\gset_1$ also contains a good approximation to $\optbench$.
Moreover, we show that for XOS valuations, it also holds that, for any $j=1,2$, the
optimal value of the LP-relaxation of the algorithmic problem on $\gset_j$
yields a good estimate for 
$\optbench$ (see Lemma~\ref{rsample-lp}). 

It is worth noting that the definition of our benchmark $\optbench$ 
to suitably account for a monopolist, is key here to showing that random sampling 
remains an effective tool that can be leveraged in the multidimensional setting.

\item {\bf Leveraging the estimate \boldmath $V_1$ via suitable demand oracles.}\  
Given the estimate $V_1$, we show how to exploit the VCG mechanism---
the prototypical
truthful mechanism  
for social-welfare maximization---in a careful fashion via suitable demand oracles.
One insight to emerge from the work of~\cite{NeogiPS24} on the single-dimensional setting 
is that computing  
$S^*=\argmax_{S\sse\gset_2}\,\bigl\{v(S)-\frac{\ld V_1}{B}\cdot c(S)\bigr\}$
via a demand oracle, yields suitable payments. The argument therein is based on the
monotonicity condition that characterizes truthfulness in single-dimensional settings,
where payments correspond to threshold values, and they show that with an XOS valuation
$v$, if we prune $S^*$ to a set $T$ with $v(T)\leq \ld V_1$, then the payments 
satisfy the budget constraint.  

The above demand-oracle computation is a VCG computation: it amounts to minimizing an
affine function of players' costs. This is promising, since it implies that there are
payments that can be combined with the algorithm to obtain a truthful mechanism (see
Theorem~\ref{vcgthm} and \eqref{vcgpayment}). However, these payments may not satisfy
the budget constraint, even in the single-dimensional setting, which is
why~\cite{NeogiPS24} need the postprocessing pruning step.

The pruning operation is however quite problematic and no longer works in the
multidimensional setting, even if we are careful and select or discard the entire
set $S^*\cap\pset_i$ of a player $i$. The crux of the issue is that, 
since players own multiple items, there is a much-richer space of possibilities in terms 
of how a player's cost function can influence $S^*$ (and hence, the final outcome):
a player who gets discarded by pruning can lie so that a different subset
from $\pset_i$ is included in $S^*$ and she is no longer discarded, and the VCG payment 
then provides her with positive utility. 
\footnote{In contrast, in the one-item-per-player setting,
the demand-set computation is immune to a player $i\in S^*$ in the
sense that $i$ cannot lie and cause a different demand-set to be computed that still
contains $i$.} 
This makes pruning $S^*$ and discarding items and/or players quite problematic. Indeed,
in general, composing procedures (e.g., VCG + postprocessing) while maintaining
truthfulness is quite challenging in multidimensional settings, 
due to the fact that one has much less control on how a player's
reported cost can affect the output of a procedure, and hence the overall output.

We avoid this issue altogether by eliminating pruning. Instead, we use a 
{\em constrained demand oracle} that only considers sets with $v(S)\leq\ld V_1$; that is,
we compute 
$S^*=\argmax_{S\sse\gset_2}\,\bigl\{v(S)-\frac{\ld V_1}{B}\cdot c(S):\ v(S)\leq\ld V_1\bigr\}$.
This still corresponds to a VCG computation, and we can again use VCG to infer payments
that ensure truthfulness. 
Also, when $v$ is XOS, due to the upper bound on $v(S^*)$, budget-feasibility easily follows.

The proof of approximation guarantee requires one additional idea to
show that $v(S^*)=\Omega(\optbench)$. Due to our random sampling, we know
that, with constant probability, $\gset_2$ contains some set $T^*_2$ with 
$c(T^*_2)\leq B$, and $v(T^*_2)=\Omega(\optbench)$; for instance, the maximum-value set in
$\gset_2$ with cost at most $B$ satisfies this. However, we may have 
$v(T^*_2)>\ld V_1$, and 
so it is unclear how to utilize this to lower bound $v(S^*)$.
The key is to argue that we can extract a suitable subset of $T^*_2$ that can act as a 
witness for lower bounding $v(S^*)-\frac{\ld V_1}{B}\cdot c(S^*)$: 
we show that we can always find $T\sse T^*_2$ such that $v(T)$ is roughly $\ld V_1$, and for which
$c(T)$ is bounded away from $B$, say is at most $B/2$ (see Lemmas~\ref{bredn-gencost}
and~\ref{bredn-supaddcost}). This step incurs an additive loss bounded by
$O\bigl(\max_{e\in T^*_2}v(e)\bigr)$ for superadditive cost functions, and 
$O\bigl(\max_i v(T^*_2\cap\pset_i)\bigr)$ for general cost functions. 
With superadditive costs, this loss is not a problem 
{\em when we assume no-overbidding}, since  we can 
``recover'' this by returning the maximum-value element in $\gset$ with some constant
probability. Without no-overbidding however, this poses a rather tricky issue, and one 
needs to come up with novel ideas to help offset this loss (as discussed below). 
\end{enumerate}

Putting everything together, with no-overbidding, our mechanisms for XOS valuations consist
essentially of 
two steps: (1) random partitioning to estimate $\optbench$;  and
(2) using a constrained demand-oracle computation to obtain $S^*$, which is tailored to 
ensure that the corresponding VCG payments yield a budget-feasible mechanism. The chief
component in the analysis, and the differences in the various mechanisms, lie in the proof
of approximation (since, as noted above, VCG payments easily ensure truthfulness and
satisfy the budget constraint). 

The above template yields universally budget-feasible mechanisms using a constrained
demand oracle. We also show that using a (standard) demand oracle, one can obtain a
budget-feasible-in-expectation mechanism (Algorithm~\ref{xosalg-expgencost} in
Section~\ref{xos-bfexp}). Here, we consider an 
{\em LP-relaxation for a constrained-demand oracle} \eqref{cdlp}, which we argue can be
solved efficiently using a demand oracle. Viewing the LP-solution as a distribution, and
again using VCG payments, we obtain a budget-feasible-in-expectation mechanism.

An important takeaway from our constructions and techniques is that the 
{\em VCG mechanism can be suitably adapted and controlled} to obtain payments that satisfy
the budget condition.

\vspace*{-1ex}
\paragraph{Dropping no-overbidding (Section~\ref{overbid}).} 
As discussed above, one incurs an additive loss---
$O\bigl(\max_{e\in T^*_2}v(e)\bigr)$ with superadditive costs, and a worse loss of
$O\bigl(\max_{i\in\pl_2}v(T^*_2\cap\pset_i)\bigr)$ with general costs---
in arguing the existence of a good set $T\sse T^*_2$ with $v(T)\leq\ld V_1$, 
$c(T)\leq B/2$. The question that arises is: {\em how do we offset this loss}?

While this may seem like a benign issue, 
it actually presents a serious obstacle. 
Let $e^*=\argmax\,\{v(e): e\in\gset,\ c(e)\leq B\}$ and $\vemax=v(e^*)$. 
Also, let $\opt_i=\max_{S\sse\pset_i}\,\{v(S): c_i(S)\leq B\}$ for $i\in\pl$, and 
$\opt^*=\max_{i\in\pl}\opt_i$; clearly $\opt^*\geq\vemax$. The additive 
loss incurred is bounded by $\opt^*$, so a natural thought would be to return the set
corresponding to $\opt^*$ with some probability; or, with superadditive costs, to return
$e^*$ with some probability.
However, observe that, even for superadditive costs, and unlike the situation with
no-overbidding, returning $e^*$ does not yield a truthful mechanism: the player $i$ owning
$e^*$ can benefit by {\em overbidding} on $e^*$ and cause the mechanism to choose a
lower-cost element from $\pset_i$. The same issue arises with returning the set
corresponding to $\opt^*$.
More tellingly, 
and this points to the trickiness of this question,
the impossibility result in Theorem~\ref{intro-thm} (a) implies that  
{\em no deterministic truthful mechanism can obtain value $\Omega(\vemax)$}, 
as this would yield an $O(n)$-approximation relative to $\optalg$. 

To circumvent this issue, 
we develop novel (non-VCG-based) budget-feasible mechanisms
(Mechanisms~\ref{optmax-proxy}, \ref{newoptmax-proxy}, and~\ref{neweroptmax-proxy}) 
to serve as a kind of ``proxy'' for the 
``return set-corresponding-to-$\opt^*$ mechanism'' (which is not truthfully implementable), 
which satisfy the following guarantee: they obtain expected value proportional to the
{\em second-largest $\opt_i$ value}, which we denote by $\opt^{(2)}$ 
(Theorems~\ref{optmax-thm}, \ref{newoptmax-thm}, and~\ref{neweroptmax-thm}).
Mechanism~\ref{optmax-proxy} has a particularly clean description: it carefully selects a
player $\hi$ with $\opt_{\hi}=\opt^*$ and returns a minimum-cost set from $\pset_{\hi}$
whose value is at least $\opt^{(2)}$; 
Mechanisms~\ref{newoptmax-proxy} and~\ref{neweroptmax-proxy} build upon this idea to ensure
{\em polytime computation} for additive 
valuations and additive costs, and superadditive costs with a constrained demand oracle,
respectively. 

The $\Omega\bigl(\opt^{(2)}\bigr)$-value guarantee that we obtain lies in a nice sweet
spot: it is weak enough to not be precluded by 
our impossibility result, and yet is strong enough 
that we {\em can} still exploit this 
to recoup the additive loss incurred, provided that we 
can ensure that $\max_{i\in\pl_2}v(T^*_2\cap\pset_i)\leq\opt^{(2)}$. We show that we can
rework the {\em analysis} to achieve this, 
by choosing $T^*_2\sse\gset_2$ appropriately.
The upshot 
is that, if we now run the random-sampling based mechanism
with suitable probability $p$, and Mechanism~\ref{optmax-proxy} with probability $1-p$,
then we obtain a good approximation relative to our benchmark 
(without assuming no-overbidding). 

\vspace*{-1ex}
\paragraph{Submodular and subadditive valuations.}
Our mechanism for monotone, submodular valuations (Section~\ref{submod}) departs from the
above template in that 
after obtaining an estimate $V_1$, it does not use a VCG mechanism. Instead, we use a
greedy algorithm and process the players in a sequence, and for each player $i$ select a
suitable subset from $\gset_2\cap\pset_i$ via a demand oracle. Thus, we perform ``local''
VCG computations but the overall mechanism is not VCG-based.

\smallskip
For subadditive valuations (Section~\ref{subadditive}), we utilize our mechanisms for XOS
valuations, by working with 
a suitable pointwise-approximation of the subadditive valuation $v$ in the demand-set
computation. 
It is well known that any subadditive valuation can be pointwise-approximated by an XOS
function within an $O(\log n)$-factor; 
but utilizing this this would only yield an $O(\log n)$-approximation. Instead, we
observe that our mechanisms for XOS valuations work when the valuation function satisfies
a limited fractional-cover property (see Section~\ref{prelim}), namely, that for any
$S\sse\gset$, every fractional cover $\{x_T\}_{T\sse S}$ of $S$ by 
{\em player-respecting sets of $S$} has value at least $v(S)$, 
where $T\sse S$ is player-respecting if $T\cap\pset_i\in\{\es,S\cap\pset_i\}$ for all $i$.
One can show that a subadditive valuation can be pointwise-approximated within an 
$O(\log k)$-factor by a function satisfying this player-respecting fractional-cover
property, and this yields $O(\log k)$-approximation factors for subadditive valuations.

\subsection{Related work} \label{relwork}
As mentioned earlier, to our knowledge, all of the work on budget-feasible mechanism
design has solely considered single-dimensional settings, and the vast majority of it
has focussed on the setting where each player owns a single item.
Following the work of Singer~\cite{Singer10}, which introduced budget-feasible mechanism
design in the single-item setting, there has been much work on developing
budget-feasible mechanisms 
for different types of
valuation functions, with approximation factors comparing the value returned by the
mechanism to the algorithmic optimum $\optalg$. 
The prominent classes of valuation classes considered are submodular
valuations~\cite{ChenGL11,JalalyT18,AmanatidisBM16,AmanatidisKS19,HuangHCT23,BalkanskiGGST22}, XOS
valuations~\cite{BeiCGL12,BeiCGL17,NeogiPS24}, and subadditive
valuations~\cite{DobzinskiPS11,BeiCGL12,BeiCGL17,NeogiPS24}. For submodular and XOS
valuations, the mechanism-design guarantees 
qualitatively match the guarantees known for the algorithmic problem. 
For subadditive valuations, there is a gap: Bei et al.~\cite{BeiCGL12} showed that an
$O(1)$-approximation mechanism exists, and an explicit exponential-time mechanism was
devised by~\cite{NeogiPS24}, but obtaining a polytime $O(1)$-approximation mechanism using
demand oracles remains an intriguing open question.

Some work~\cite{ChanC14,AnariGN18,KlumperS22,AmanatidisKMST23} has
considered a richer, but still single-dimensional, 
level-of-service (\los) setting, wherein a player
can provide multiple levels of service (or multiple units of an item) and incurs the same
incremental cost for each additional unit provided. 
Typically, one assumes that the valuation function $v$ is additive across players, and 
is a linear or concave function of each player's provided level of service. In this setup,
approximation factors of $O(\log n)$ and $O(\max_i|\pset_i|)$ were obtained
by~\cite{ChanC14,AmanatidisKMST23} respectively; \cite{ChanC14} also considered subadditive
valuations and obtained an $O\bigl(\frac{\log^2 n}{\log\log n}\bigr)$-approximation.
Another line of work considers the divisible-item setting, wherein one can buy a
fraction of an item from a seller, and a related large-market
assumption~\cite{AnariGN14,BalkanskiH16,JalalyT18}, which assumes that a single player has
a negligible ``effect'' on the overall value. The latter assumption is similar in spirit 
to what our benchmark aims to capture, but the goal in these works is quite different.
They still examine the single-dimensional setting, and the goal is to derive improved
approximation guarantees (with respect to $\optalg$) under this assumption.
A natural extension of the large-market assumption to the multidimensional
setting would be to assume that $v(\pset_i)\leq\ve\cdot\optalg(v,B,c)$ for all $i\in\pl$,
for some parameter $\ve<1$. 
This 
precludes a monopolist by design, and so an alternate approach 
would have been to restrict attention to such inputs, but compare against $\optalg$. As
discussed earlier, the benefit of coming up with an alternate benchmark that applies to
all inputs (as we do) is that one can then obtain guarantees for {\em all} inputs. 
Moreover, a guarantee relative to $\optbench$ is stronger, in the sense that 
any such guarantee 
translates to a guarantee relative to $\optalg$ for inputs satisfying the large-market
assumption, as we have $\optbench\geq(1-\ve)\optalg$ for such inputs. 

The need for considering suitable alternate benchmarks invites comparison to the areas of  
{prior-free profit maximization}~\cite{GoldbergHKSW06,GoldbergHW01} 
and  {frugal mechanism design}~\cite{ArcherT07,Talwar03}, where similar considerations arise. 
The issues we encounter in multidimensional budget-feasible mechanism design 
are similar in flavor (but the details differ) to those arising in prior-free profit-maximization. 
There as well, the need for 
alternate benchmarks arises because the natural comparison point 
proves to be too strong due to the presence of a single ``dominant'' player: 
one cannot obtain any non-trivial guarantees with respect to the optimal fixed-price
revenue $\F$, as this could extract all its revenue from a single winner~\cite{GoldbergHKSW06}. 
Moreover, there as well, there is no ``best'' truthful auction. 
One therefore considers the $\F^{(2)}$-benchmark~\cite{GoldbergHKSW06,GoldbergHW01} that
hard-codes that there are at least two winners, which is perhaps closest in spirit as our
benchmark. 
Frugal mechanism design considers a similar setup as budget-feasible mechanism design,
where a buyer seeks to procure a suitable set from self-interested players. 
But any feasible set suffices, and the buyer's goal is to minimize the total payment made
to the players. 
There is no clear benchmark here to compare the payment of a truthful mechanism, and
a variety of benchmarks have been proposed in the literature, such as 
the cost of the second-cheapest set~\cite{ArcherT07,Talwar03}, and quantities that are,
loosely speaking, motivated by considering equilibria of first-price 
auctions~\cite{KarlinKT05,ElkindGG07,HajiaghayiKS18}, and various mechanisms have been
devised that obtain good guarantees relative to these
benchmarks~\cite{KarlinKT05,ElkindGG07,KempeSM10,ChenEGP10,HajiaghayiKS18}. With the sole
exception of~\cite{MinooeiS12} who consider multidimensional vertex cover, all of this work
focusses on single-dimensional problems.
Finally, we remark that removing the maximum contribution from a player 
in the definition of $\optbench$ bears cosmetic similarity to the approximate envy-free 
fairness notions EF1~\cite{LiptonMMS04} and EFX~\cite{CaragiannisKMPS16}.

The setting of cost-sharing mechanisms provides an interesting comparison point
with budget-feasible mechanism design, 
being another example of a prominent domain where:  
(1) prices feature both in truthfulness and the constraints; (2) even the existence of   
mechanisms satisfying the desirable criteria (truthfulness, cost-recovery, and good
approximation of the social-cost objective~\cite{RoughgardenS09}) is not a given; and 
(3) most work has investigated only the single-dimensional setting. 
Unlike budget-feasible mechanism design, there has been some prior work that has 
explored multidimensional cost-sharing mechanism design, 
both in a combinatorial setting~\cite{GeorgiouS19,DobzinskiO17,BirmpasMS22} that is
similar in spirit to the setting we consider, as also in a level-of-service
setting~\cite{MehtaRS09,GeorgiouS19}. While the technical aspects are quite different, it
is worth pointing out that  
Dobzinski and Ovadia~\cite{DobzinskiO17} also leverage VCG in novel ways, developing
mechanisms that significantly advance the realm of problems for which one can obtain good
cost-sharing mechanisms, and they focus on demonstrating the existence of good
mechanisms bereft of computational considerations. 

Multidimensional mechanism design has been most prominently considered in the setting of
combinatorial auctions, both in the context of social-welfare maximization (see,
e.g.,~\cite{LaviAGT,LaviS11,DughmiRY16,AssadiS19,AssadiKS21}), and profit or revenue maximization,
(see,
e.g.~\cite{HartlineKAGT,DuttingKL20,CaiDW12a,CaiDW12b,BabaioffILW14,EdenFFTW17,EdenFFTW21}),
where 
one often considers the goal of obtaining a good approximation via a simple mechanism. 
As noted earlier, one key source of difficulty in multidimensional mechanism design,
compared to the single-dimensional setting, is the lack of a convenient characterization
of truthfulness.
Although characterizations based on cycle monotonicity
and weak-monotonicity are known for truthfulness, and payments can be obtained by computing
shortest paths in a certain graph, these have found quite scant application in the
design of multidimensional truthful mechanisms; see~\cite{LaviS06,BabaioffKS13} for some
exceptions. Instead, VCG-based mechanisms (leading to the MIDR
approach)~\cite{LaviS11,DobzinskiNS05,DobzinskiD13,DughmiRY16} and posted-price
mechanisms~\cite{FeldmanGL14,DuttingKL20}, where one offers take-it-or-leave-it prices,
have been the chief means for obtaining truthful mechanisms.

\section{Preliminaries} \label{prelim}

We use $\R_+$ and $\Z_+$ to denote the set of nonnegative reals, and nonnegative
integers respectively.
For an integer $n\geq 1$, let $[n]$ denote $\{1,\ldots,n\}$.
Let $k$ be the number of players, and $\pl=[k]$ denote the set of all players.
Let $n=|\gset|$. Throughout, we use $i$ to index players, and $e$ to index items.
Recall that:
(a) $\pset_i\sse\gset$ is the publicly-known set of items owned
by player $i$; (b) the $\pset_i$s partition $\gset$;
(c) $\C_i$ is the publicly-known collection of possible player-$i$ cost functions;
and (d) $\C=\Pi_{i=1}^k\C_i$. 
Given a cost-function vector $c\in\C$, we define $c(S):=\sum_{i\in[k]}c_i(S\cap\pset_i)$
for any $S\sse\gset$.
Throughout, $\OPTalg(v,B,c):=\max\,\{v(S): S\sse\gset,\ c(S)\leq B\}$ denotes the algorithmic
optimum. 
It will be convenient to introduce the notation 
\swamy{$\vbench(S):=\min_{i\in[k]}v(S-\pset_i)$}
for $S\sse\gset$. Recall that we utilize the following benchmark
\[
\swamy{\optbench(v,B,c) \ :=\
\max_{S\sse\gset}\,\Bigl\{\min_{i\in[k]}v(S-\pset_i):\ c(S)\leq B\Bigr\}
\ =\ \max_{S\sse\gset}\,\Bigl\{v_{-1}(S):\ c(S)\leq B\Bigr\}.}
\]
As remarked earlier, one nice feature of $\optbench$ is that it is parameter-free. But we
note that one could also consider a parametrized benchmark where (in the spirit of the
$\F^{(2)}$-benchmark used in prior-free profit maximization) we restrict attendion to sets
$S$ where no player contributes more than an $\ve$-fraction of $v(S)$. 
This yields $\optparam(\ve;v,B,c):=
\max_{S\sse\gset}\,\bigl\{v(S):\ c(S)\leq B,\ \ v(S\cap\pset_i)\leq\ve\cdot v(S)\ \ \forall i\in[k]\bigr\}$.
as a benchmark.
Note that $\optbench(v,B,c)\geq\max_{\ve\in[0,1]}(1-\ve)\optparam(\ve;v,B,c)$, so guarantees
relative to $\optbench$ also yield guarantees relative to $\optparam(\ve)$; 
hence, we do not consider the latter benchmark.

To avoid cumbersome notation, we will frequently drop $v,B$ from the arguments, as these
will usually be fixed and clear from the context; we also drop $c$ when this is clear from
the context.
For a singleton set $\{e\}$, we use $v(e)$, $c_i(e)$ and $c(e)$ to denote the respective
function value on that set.

\vspace*{-1ex}
\paragraph{Valuation functions.}
Let $v:2^\gset\mapsto\R_+$ be a valuation function. We will always assume that $v$ is
normalized, i.e., $v(\es)=0$. 
We say that $v$ is {\em monotone}, if $v(S)\leq v(T)$ whenever $S\sse T\sse\gset$. 
We consider various classes of valuation functions. 
We say that $v$ is:  
\begin{enumerate}[label=$\bullet$, topsep=0.25ex, noitemsep, leftmargin=*]
\item {\em additive} (or {\em modular}), if there exists some $a\in\R^\gset$ such that
$v(S)=a(S)$ for all $S\sse\gset$. 
Note that an additive valuation is monotone iff it is nonnegative.
\item {\em submodular}, if $v(S)+v(T)\geq v(S\cap T)+v(S\cup T)$ for all $S,T\sse\gset$.

\item {\em XOS}, if $v$ is the maximum of a finite collection of additive functions, i.e.,
there exist $a^1,\ldots,a^k\in\R^\gset$ such that $v(S)=\max_{i\in[k]}a^i(S)$ for all
$S\sse\gset$. 
Note that we allow the additive functions to be negative, and this allows us to 
capture non-monotone XOS functions.
The above definition is equivalent to saying that for every $S\sse\gset$, there exists
some $w\in\R^\gset$ such that $v(S)=w(S)$ and $v(T)\geq w(T)$ for all $T\sse S$; we say
that {\em $w$ (or the corresponding additive valuation) supports $S$}. 
XOS valuations are also called {\em fractionally subadditive} valuations, and can be
equivalently defined in terms of fractional covers. A {\em fractional cover} of a set
$S\sse\gset$ is a collection $\{\mu_T\}_{T\sse S}$ such that 
$\sum_{T\sse S: e\in T}\mu_T=1$ for all $e\in S$, and its value is defined as
$\sum_{T\sse S}v(T)\mu_T$. We say that $v$ is fractionally subadditive, if for every
$S\sse\gset$, every fractional cover of $S$ has value at least $v(S)$. Using LP duality,
it is not hard to infer that this is equivalent to  stating that every $S\sse\gset$ has a
supporting additive valuation.

When $v$ is monotone and XOS, we can assume that $v$ is the maximum of a
collection of nonnegative additive valuations, and we can relax the fractional-cover
condition to the inequality $\sum_{T\sse\gset: e\in T}\mu_T\geq 1$ for all $e\in S$.
This is because we can always ensure that equality holds here by
dropping elements from sets if needed, and with a monotone valuation, this does not
increase the value of the fractional cover. 

\item {\em subadditive}, if $v(S\cup T)\leq v(S)+v(T)$ for all $S,T\sse\gset$.
\end{enumerate}
It is well known that additive valuations are a strict subclass of submodular valuations,
which in turn form a strict subclass of both XOS and subadditive valuations. 
Also, {\em monotone} XOS valuations form a strict subclass of monotone subadditive
valuations.  

\vspace*{-1ex}
\paragraph{Oracles used for accessing \boldmath $v$.}
Our mechanisms chiefly utilize two types of oracles for accessing the valuation $v$, in
addition to a value oracle: a (suitably generalized) demand oracle and a constrained
demand oracle. 
As noted earlier, since in the multidimensional setting, we work with player-cost
functions that are more general than additive functions (e.g., monotone, normalized
functions, or superadditive functions), we need to extend the notion of a demand to
consider such general price functions. 

Recall that $\C_i\sse\R_+^{2^{\pset_i}}$ denotes the class of possible player-$i$ cost
functions, $\C=\Pi_i\C_i$, and for $c\in\C$ and $S\in\gset$, we define
$c(S):=\sum_{i\in[k]}c_i(S\cap\pset_i)$. 
A {\em demand oracle for the class $\C$} takes $q\in\C$, $\kp\geq 0$, as input, and
returns $\argmax_{S\sse\gset}\,\bigl(v(S)-\kp\cdot q(S)\bigr)$. 
Note that a standard demand oracle amounts to a demand oracle for the class of additive
functions: it takes item prices $q\in\R_+^\gset$ and performs the above computation for
the additive function specified by $q$.
The underlying motivation and rationale behind this definition is the same as that for a
standard demand oracle. A demand oracle answers the economic question of what is the most 
profitable set for the buyer under given prices for the items, which turns out to be
beneficial 
because it serves as a means of aggregating players' costs.  
Since players' costs are now given by general cost functions, it is only apt that we use
the same expressive power when considering item prices, and therefore we now consider a
price function $q\in\C$ when answering this question. 
\footnote{From a computational perspective, it is known that value oracles
are insufficient even for the algorithmic problem of computing a good approximation to
$\optalg$ in polynomial time, even with additive costs.}

A {\em constrained demand oracle for the class $\C$} additionally takes a cap
$\targ\in\R_+$ as input, and returns
$\argmax\,\{v(S)-\kp\cdot q(S): S\sse\gset,\ v(S)\leq\targ\}$. 
It can be seen as answering the more-nuanced economic question of what is the most
profitable value-capped set for the buyer. 

We will often need to 
optimize the objective of a demand- or constrained-demand oracle only over subsets of some 
given set $A\sse\gset$. We assume that a demand-oracle or constrained-demand oracle has
this flexibility.
\footnote{Most often, the set $A$ will be of the form $\bigcup_{i\in I}\pset_i$ for
some $I\sse[k]$. A demand oracle on $A$ can be encoded as a query over the entire set
$\gset$ by taking 
$c_\ell$, for $\ell\notin I$ to be the constant function: $c_\ell(S)=M_\ell$ for all
$S\sse 2^{\pset_\ell}-\{\es\}$, where $M_\ell$ is sufficiently large, say
$2v(\pset_\ell)/\kp$ (assuming $v$ is monotone, subadditive). 
Thus, this added flexibility is a very benign requirement.}

\begin{theorem} \label{lpsolve}
Let $\bigl(v:2^\gset\mapsto\R_+,B,\{\pset_i,\C_i\sse\R_+^{2^{\pset_i}},c_i\in\C_i\}\bigr)$ be a
budget-feasible mechanism-design instance. Suppose we are given a demand oracle for
$\C=\Pi_i\C_i$. For any $A\sse\gset$, $\kp\geq 0$, $\targ\in\R_+$, the following LPs can
be solved in polytime: 
\begin{alignat*}{1}
\text{(a)} & \ \ 
\max\ 
\sum_{S\sse A}v(S)x_S \quad \mathrm{s.t.} \quad 
\sum_{S\sse A}c(S)x_S\leq B, \quad \sum_{S\sse A} x_S\leq 1, \quad x\geq 0. 
\tag{\optlpname($A$)} \\ 
\text{(b)} & \ \ 
\max\ 
\sum_{S\sse A}\bigl(v(S)-\kp\cdot c(S)\bigr)x_S \quad \mathrm{s.t.} \quad 
\sum_{S\sse\gset_2}v(S)x_S\leq\targ, \quad \sum_{S\sse A} x_S\leq 1, \quad x\geq 0.
\tag{CDLP($A$)} \label{cdlp}
\end{alignat*}
\end{theorem}
\begin{proof}
	Both parts follow by considering the dual of the respective LPs and observing that a
	demand oracle yields a separation oracle for the dual. Therefore, the ellipsoid method can 
	be used to solve the dual LPs, and hence the primal LPs.

	For part (a), the dual of $\optlp[(A)]$ has the constraint $\al\cdot c(S)+\beta\geq v(S)$
	for all $S\sse A$, where $\al,\beta\geq 0$ are the dual variables corresponding to the
	primal constraints. This amounts to determining if 
	$\max_{S\sse A}\bigl(v(S)-\al\cdot c(S)\bigr)\leq\beta$, which we can determine by using a 
	demand-oracle query over the set $A$ to find $S^*$.
	For part (b), the dual of \eqref{cdlp} has the constraint 
	$\al\cdot v(S)+\beta\geq v(S)-\kp\cdot c(S)$ for all $S\sse A$, where 
	$\al,\beta\geq 0$. If $\al\geq 1$, then these constraints are trivially satisfied, so
	assume otherwise. Then, feasibility amounts to determining if 
	$\max_{S\sse A}\bigl(v(S)-\frac{\kp}{1-\al}\cdot c(S)\bigr)\leq\frac{\beta}{1-\al}$, which
	can again be answered via a demand-oracle query over $A$.
\end{proof}

\begin{remark} \label{additive-lpsolve}
For an additive valuation function $v$ and additive cost functions, LPs
$\optlp[(A)]$ and \eqref{cdlp} can be solved in polytime, since one has a polytime demand
oracle. They can also be cast as polynomial-size LPs using $(x_e)_{e\in A}$ variables, 
where $x_e$ denotes $\sum_{S\sse A: e\in S}x_S$; the cosntraint $\sum_{S\sse A}x_S\leq 1$
can be dropped in this formulation. 
\end{remark}

\vspace*{-2ex}
\paragraph{Mechanism design.}
In the basic mechanism design setup, we have a set of $k$ players, and a set $A$ of
possible outcomes. Each player $i$ has a {\em private type}
$c_i:A\mapsto\R_+$, where $c_i(a)$ gives the cost of alternative $a\in A$ to player $i$. 
\footnote{We describe the setup in terms of cost incurred and payments, instead of the
more-common choice of value obtained and prices, as this is what we encounter in
budget-feasible mechanism design.} 
Let $\C_i$ be the publicly-known set of all valid types of player $i$ (so $c_i\in\C_i$). 
Let $\C=\C_1\times \cdots \times \C_k$ denote the space of all players' valid types.
(We have deliberately overloaded notation here, as the type $c_i$ and set $\C_i$ above are 
essentially the same as the cost-function $c_i:2^{\pset_i}\mapsto\R_+$ and set $\C_i$ of
such cost functions in the budget-feasible mechanism-design (MD) setup.) 

For example, multidimensional budget-feasible mechanism design can be cast in the above
setup as follows. We have $A=2^\gset$. Each player $i$'s cost function
$c_i:2^{\pset_i}\mapsto\R_+$ yields a corresponding type, that we also denote by $c_i$,
where $c_i(S):=c_i(S\cap\pset_i)$; the set $\C_i$ of player-$i$ cost functions,
correspondingly maps to the set of valid player types for player $i$. 
We use $c$ to denote the tuple $(c_1,\ldots,c_k)$, and $c_{-i}$ to denote the
tuple 
that excludes $i$'s type (or cost function). Similarly $\C_{-i}=\prod_{j\in[k]-\{i\}}\C_{j}$.

A (direct revelation) {\em mechanism} consists of an {\em algorithm} or allocation rule
$f:\C\mapsto A$, and a {\em payment function} $p_i:\C\mapsto\R$ for each player $i$. 
Each player $i$ reports some type $c_i\in\C_i$ (possibly deviating from her true type),
and the mechanism computes the outcome $f(c)$ and pays $p_i(c)$ to each player $i$.
Note that in budget-feasible mechanism design, $f$ and the $p_i$s can depend on the
publicly-known information $(v,B,\{\pset_i,\C_i\})$, which we will treat implicitly as
being fixed. 
The {\em utility} that $i$ obtains when her true type is $\bc_i$, she reports $c_i$, and
other players report $c_{-i}$ is
$\util_i(\bc_i;c_i,c_{-i}):=p_i(c_i,c_{-i})-\bc_i(f(c_i,c_{-i}))$, and each player aims to
maximize her own utility. 
In multidimensional budget-feasible mechanism design, we seek a mechanism 
$\mech=\bigl(f,\{p_i\}_{i\in[k]}\bigr)$ satisfying the following 
properties.  
\begin{enumerate}[label=$\bullet$, topsep=0.25ex, noitemsep, leftmargin=*]
\item $\mech$ is {\em truthful}: each player $i$ maximizes her utility by reporting her
true private type: for every $\bc_i,c_i\in\C_i$ and $c_{-i}\in\C_{-i}$, we have
$\util_e(\bc_e;\bc_e,c_{-e})\geq\util_e(\bc_e;c_e,c_{-e})$.

\item $\mech$ is {\em individually rational} (IR): $\util_i(\bc_i;\bc_i,c_{-i})\geq 0$ for
every $i$, every $\bc_i\in\C_i$ and $c_{-i}\in\C_{-i}$; 
note that this implies that $p_i(c)\geq 0$ for all $c$. 
We say that $\mech$ makes {\em no positive transfers} (NPT) if 
$p_i(c)=0$ whenever $S=f(c)\sse\gset$ is such that $S\cap\pset_i=\es$. 
(More abstractly, we can say that $p_i(c)=0$ whenever $a=f(c)$ is such that $c'_i(a)=0$
for all $c'_i\in\C_i$.) 
In the sequel, we will always implicitly require NPT.
\footnote{We remark that the lower bounds we prove in Section~\ref{impos} hold also for
mechanisms that may not satisfy NPT.}

\item $\mech$ is {\em budget feasible}: we have $\sum_i p_i(c)\leq B$ for every type
$c\in\C$. 
Note that if $\mech$ is individually rational, this implies that 
$\sum_i c_i\bigl(f(c)\bigr)\leq B$. 
\end{enumerate}

A {\em randomized} mechanism can use random bits to determine $f(c)$ and
$\{p_i(c)\}$; so the cost-incurred by, payment made to, and utility of, a player are all
random variables. 
We say that a randomized mechanism is:
\begin{enumerate}[label=(\alph*), topsep=0ex, noitemsep, leftmargin=*]
\item {\em universally budget feasible}, if truthfulness, IR, and budget feasibility hold
with probability $1$, i.e., the mechanism can be viewed as a distribution over deterministic
budget-feasible mechanisms. 
\item {\em budget-feasible in expectation}, if
$\mech$ is {\em truthful in expectation}, the expected total payment is at most $B$, and
IR (and NPT) holds with probability $1$.
\footnote{If IR holds in expectation, i.e., the expected utility of a truthful player is
nonnegative, then one can define random payments that ensure that IR holds with
probability $1$,  without changing the expected payment made by the mechanism; 
see Lemma~\ref{bfinexp}.}
Truthful in expectation means that the expected utility of a player is maximized by
truthful reporting.  
\end{enumerate}
We say that $\mech$ achieves approximation ratio $\al$ with respect to a benchmark
$\bench$, if $v\bigl(f(c)\bigr)\geq\bench(c)/\al$ for all $c\in\C$. If $\mech$ is
randomized, then we have $\Ex\big[v\bigl(f(c)\bigr)\bigr]\geq\bench(c)/\al$ for all
$c\in\C$. 

A central tool that we utilize is the {\em VCG mechanism}~\cite{Vickrey61,Clarke71,Groves73}, 
which is in fact a family of mechanisms, showing that if the algorithm $f$ is such that it
minimizes an affine function of the players' costs, then one can always combine it with
suitable payments to obtain a truthful mechanism. This classical result is one of the
predominant tools often leveraged in multidimensional mechanism
design~\cite{LaviS11,DobzinskiD13,DughmiRY16}, forming the basis of the
maximal-in-distributional-range (MIDR) approach.

\begin{theorem}[VCG mechanism~\cite{Vickrey61,Clarke71,Groves73}] \label{vcgthm}
\label{bfmdvcgthm}
Let $f:\C\mapsto\R_+$ be given by
$f(c)=\argmin_{a\in A}\bigl(\sum_{i\in[k]}\al_ic_i(a)+\beta_a\bigr)$ for all $c\in\C$, where
$\al_i>0$ for every player $i$. Consider the following payments: 
\begin{equation}
p_i(c)=\frac{-1}{\al_i}\bigl(\sum_{\ell\in[k]-\{i\}}\al_\ell c_\ell(a^*)
+\beta_{a^*}\bigr)+h_{i}(c_{-i}), \quad\text{where $a^*=f(c)$} 
\qquad \frall i, c\in\C \tag{VCG} \label{vcgpayment}
\end{equation} 
where $h_i$ is some function that depends only $c_{-i}$.
Then the mechanism $(f,p)$ is truthful.

For budget-feasible mechanism design (where $A=2^\gset$), one choice of $h_i$s that ensures
IR and NPT is to set 
$h_i(c_{-i}):=
\frac{1}{\al_i}\cdot\min_{S\sse\gset-\pset_i}\bigl(\sum_\ell\al_{\ell\in[k]-\{i\}} c_\ell(S)+\beta_S\bigr)$
for all $i$, $c_{-i}\in\C_{-i}$.
\end{theorem}

\begin{proofsketch}
Fix player $i$, $\bc_i,c_i\in\C_i$, $c_{-i}\in\C_{-i}$. Let $a^*=f(\bc_i,c_{-i})$ and
$b=f(c_i,c_{-i})$. Then 
\begin{equation*}
\begin{split}
\util(\bc_i;\bc_i,c_{-i}) &= h_i(c_{-i})-\frac{1}{\al_i}
\bigl(\sum_{\ell\in[k].\ell\neq i}\al_\ell c_\ell(a^*)+\al_i\bc_i(a^*)+\beta_{a^*}\bigr), 
\qquad \text{and} \\
\util(\bc_i;c_i,c_{-i}) &= h_i(c_{-i})-\frac{1}{\al_i}
\bigl(\sum_{\ell\in[k],\ell\neq i}\al_\ell c_\ell(b)+\al_i\bc_i(b)+\beta_{b}\bigr).
\end{split}
\end{equation*}
But by definition, 
$\sum_{\ell\in[k],\ell\neq i}\al_\ell c_\ell(a^*)+\al_i\bc_i(a^*)+\beta_{a^*}
=\min_{a\in A}\bigl(\sum_{\ell\in[k]\ell\neq i}\al_\ell c_\ell(a)+\al_i\bc_i(a)+\beta_{a}\bigr)$.
This also shows that with the given choice of $h_i$s for budget-feasible mechanism design,
we obtain IR, since $h_i$ can be viewed as optimizing the objective function underlying
$f$ over a subset of $A$. We obtain NPT because if $a^*=S^*\sse\gset$ is such that
$S^*\cap\pset_i=\es$, then $\util(\bc_i;\bc_i,c_{-i})=0$ and $\bc_i(a^*)=0$, so $p_i(\bc_i,c_{-i})=0$.
\end{proofsketch}

Our mechanisms utilize the VCG mechanism in a bit more generality. We will
usually have a set $\pl'$ of players, with $\al_i>0$ for all $i\in\pl'$, and $\al_i=0$
for all $i\notin\pl'$. With this,
one can still utilize Theorem~\ref{vcgthm}, taking the alternative set $A$ to be
$2^{\gset'}$, where $\gset'=\bigcup_{i\in\pl'}\pset_i$, 
and setting the payments of players not in $\pl'$ to always be $0$.

\section{Impossibility results and lower bounds} \label{lbounds} \label{impos}
We prove here the impossibility results that were stated informally in
Theorem~\ref{intro-thm}, which rule out any good approximation with respect to 
$\optalg(v,B,c)$. 

\begin{theorem} \label{det-overbid}
Consider any $\al\geq 1$. There is an additive valuation $v$ and budget $B$, such that 
if $\mech$ is a deterministic budget-feasible mechanism, then $\mech$ obtains value at
most $\optalg(v,B,c)/\al$ for some additive cost function $c$.
\end{theorem}

\begin{proof}
Let $\gset=\{e,f\}$ with $v(e)=\al$, $v(f)=1$. Let the budget $B$ be $1$.
There is only one player. Consider the additive cost function $c^{(1)}$ given by
$c^{(1)}_e=B,\ c^{(1)}_f=0$, 
and $c^{(2)}$ 
given by $c^{(2)}_e=B+1,\ c^{(2)}_f=B$. (Note that $c^{(2)}$ is a valid input since we are
not assuming no-overbidding.)

On input $c=c^{(2)}$, 
the mechanism cannot return $e$ due to budget-feasibility, and must return $f$, as
otherwise the statement holds for $c=c^{(2)}$. 
By IR, $\mech$ must pay $B$ to the player under $c^{(2)}$. 
Now, $\mech$ cannot output $e$ on input $c^{(1)}$, and so the statement holds
for $c=c^{(1)}$. This is because, otherwise, on input $c^{(1)}$, the player obtains $0$
utility by reporting $c^{(1)}$ (as she is paid at most $B$ due to budget feasibility), but
obtains utility $B$ by reporting $c^{(2)}$, contradicting truthfulness.
\end{proof}

\begin{theorem} \label{rand-overbid} \label{rand-lb}
Consider any $\e>0$. 
There is an additive valuation $v$ and budget $B$, such that if $\mech$ is a
budget-feasible-in-expectation mechanism, 
then $\mech$ obtains value at most $\optalg(v,B,c)\cdot\frac{1+\e}{n}$ for some
additive cost function $c$. 
\end{theorem}

\begin{proof}
We may assume that $\e\leq 1$. Let $\gm=1+\frac{n}{\epsilon}$.
Again, there is only one player. We identify $\gset$ with $[n]$.
Let $v$ be the additive valuation defined by $v(e)=\gm^{n-e}$ for all $e\in[n]$, and the
budget $B$ be $1$.
For each $\ell\in [n]$, let $c^{(\ell)}$ be the additive cost function defined by 
$c^{(\ell)}_e=M\geq (1+n\gm^n)B$ for all $e\in[\ell-1]$, $c^{(\ell)}_\ell=1$, and
$c^{(\ell)}_e=0$ for all $e\in\{\ell+1,\ldots,n\}$. 
We argue that on some input $c^{(\ell)}$, $\mech$ obtains value at most
$\optalg\bigl(v,B,c^{(\ell)}\bigr)\cdot\frac{1+\e}{n}$. 
(Again, $c^{(\ell)}$ is allowed as input, since we are not assuming no-overbidding.)

For $\ell,r\in[n]$, let $p_\ell$ be the expected payment made by $\mech$ when the player
reports $c^{(\ell)}$, and let $b_{\ell,r}$ be the expected cost incurred by the player
when her true cost is $c^{(\ell)}$ and she reports $c^{(r)}$. 
Note that on input $c^{(\ell)}$, the expected number of items returned by $\mech$ from
$[\ell-1]$ is at most $\frac{1}{1+n\gm^n}$. This is because otherwise, we would have
$b_{\ell,\ell}>B$, and so $p_\ell>B$ by IR, which contradicts budget-feasibility in
expectation. 
By truthfulness, we have $p_\ell-b_{\ell,\ell}\geq p_r-b_{\ell,r}$ for all
$\ell,r\in[n]$. In particular, we have 
\begin{equation*}
p_\ell-p_{\ell+1}\geq b_{\ell,\ell}-b_{\ell,\ell+1} \qquad \frall \ell\in[n].
\end{equation*}
Adding the above for all $\ell\in[n-1]$, along with $p_n\geq b_{n,n}$, which follows due
to IR, we obtain that $p_1\geq b_{1,1}+\sum_{\ell=2}^n(b_{\ell,\ell}-b_{\ell-1,\ell})$. By
budget-feasibility in expectation, we have $p_1\leq 1$. Therefore, we have 
$b_{1,1}\leq 1/n$ or there is some $\ell\in\{2,\ldots,n\}$ with
$b_{\ell,\ell}-b_{\ell-1,\ell}\leq\frac{1}{n}$. Since $b_{\ell,\ell}-b_{\ell-1,\ell}$ is
at least 
$\Pr\bigl[\mech\text{ returns a set containing item $\ell$ on input $c^{(\ell)}$}\bigr]$, 
this in turn implies that there is some $\ell\in[n]$ such that 
$\Pr\bigl[\mech\text{ returns a set containing item $\ell$ on input $c^{(\ell)}$}\bigr]\leq 1/n$.

Finally, we argue that this last inequality implies that on input $c'=c^{(\ell)}$, $\mech$
obtains value at most $\optalg(v,B,c')\cdot\frac{1+\e}{n}$. 
We have $\optalg(v,B,c')=v(\{\ell,\ldots,n\})$.
The expected value $\targ$ obtained by $\mech$ on input $(v,B,c')$ is at most 
\[
\tfrac{v(\ell)}{n}+v(\{\ell+1,\ldots,n\})+
v(1)\cdot\E{\text{no. of items returned by $\mech$ from $[\ell-1]$}} 
\leq v(\ell)\cdot\Bigl(\tfrac{1}{n}+\tfrac{\gm^{n-1}}{1+n\gm^n}\Bigr)+v(\{\ell+1,\ldots,n\}).
\]
So we have 
$\optalg-v(\{\ell+1,\ldots,n\})\geq  
n\bigl(\targ-v(\{\ell+1,\ldots,n\})-\frac{\gm^{n-1}}{1+n\gm^n}\cdot v(\ell)\bigr)$. 
Therefore
\begin{equation*}
\optalg\geq n\cdot\targ-(n-1)v(\{\ell+1,\ldots,n\})-\frac{v(\ell)}{\gm}
\geq n\cdot\targ-\tfrac{n}{\gm-1}\cdot v(\ell)
\geq n\cdot\targ-\e\cdot\optalg. \qedhere
\end{equation*}
\end{proof}

The following impossibility result under no-overbidding was shown by~\cite{ChanC14}, even
for the single-dimensional \los setting.
We include a proof in Appendix~\ref{append-impos} for completeness.

\begin{theorem}[\cite{ChanC14}] \label{detlb} \label{truthexplb}
Let $v$ be the additive valuation with $v(e)=1$ for all $e\in\gset$, and let the budget
be $B=n=|\gset|$. Let $\mech$ be a budget-feasible mechanism. 
\begin{enumerate}[label=(\alph*), topsep=0ex, noitemsep, leftmargin=*]
\item If $\mech$ is deterministic, then for the above $(v,B)$, even assuming
no-overbidding, there is some additive cost
function $c$ for which $\mech$ obtains value at most $\optalg(v,B,c)/n$. 
\item If $\mech$ is budget-feasible in expectation, even assuming no-overbidding, there is
some additive cost function $c$ for which $\mech$ obtains value at most
$\optalg(v,B,c)/O(\log n)$. 
\end{enumerate}
\end{theorem}

The above lower bounds are obtained on instances involving a single player. One
can therefore argue using Yao's minimax principle that the lower bounds for randomized
mechanisms in Theorems~\ref{rand-overbid} and~\ref{truthexplb} (b) extend to Bayesian
budget-feasible mechanisms. We define the Bayesian setting, and prove the following
corollary, in Appendix~\ref{append-impos}.

\begin{corollary} \label{bayesianlb}
No Bayesian budget-feasible mechanism can achieve approximation ratio better than $n$, and  
better than $O(\log n)$ assuming no-overbidding, relative to $\optalg(v,B,c)$.
\end{corollary}

We also prove a lower bound on the approximation guarantee achievable with respect to
$\optalg$ by any truthful (even non budget-feasible) mechanism. 
This is in stark contrast  
to the single-dimensional setting, where lower bounds only exist for budget-feasible
mechanisms; 
the algorithm that returns an optimal solution $\argmax_{S\sse U}\,\{v(S): c(S)\leq B\}$
(with consistent tie breaking) satisfies Myerson's monotonicity
condition, and hence is truthfully implementable.

\begin{theorem} \label{notruthful} 
No deterministic truthful mechanism can achieve approximation ratio strictly larger
than $\phi=\frac{1+\sqrt{5}}{2}$ with respect to $\optalg$, even with additive valuations
and additive cost functions.
\end{theorem}

\begin{proof}
Consider again the setting with a single player. There are $2$ items, with values
$v(e_1)=1$, $v(e_2)=\phi$, and $v$ is additive, and the budget is $B=2$.  
Consider the additive cost functions $c_1, c'_1$ that assign costs 
$c_1(e_1)=1, c_1(e_2)=B$ and $c'_1(e_1)=1.5, c'_1(e_2)=0.5$. 
We have $\optalg(c_1)=\phi$, so if the mechanism attains approximation ratio larger than
$\phi$, it must return $e_2$ on input $c_1$. 
But then on input $c'_1$, due to truthfulness, the mechanism must still return only
$e_2$. Suppose the mechanism returns $S\sse\{e_1,e_2\}$. Then, by weak-monotonicity, we
must have 
$c_1(e_2)+c'_1(S)\leq c_1(S)+c'_1(e_2)$, i.e., $c'_1(S)\leq c_1(S)-(B-0.5)$, and only
$S=\{e_2\}$ satisfies this inequality. 
But then the mechanism's approximation ratio on input $c'_1$ is $\frac{\phi+1}{\phi}=\phi$.
\end{proof}

We remark that while the constructions above utilize only a single player, we can 
always pad the instance by introducing ``dummy'' players that incur $0$ cost, and
contribute very little value to $\optalg$. 

\paragraph{Lower bounds on the approximation factor relative to \boldmath $\optbench$.} 
Lower bounds under a large-market assumption introduced by \cite{AnariGN14} in
the single-dimensional setting, carry over to lower bounds on the approximation factor
achievable relative to \optbench. This is simply because under the large-market assumption
$\vemax\ll\optalg$ in the single-dimensional setting, 
$\optbench$, which is at least $\optalg -\vemax$, essentially coincides with $\optalg$.  
Therefore, an $\frac{e}{e-1}$ 
approximation-factor lower bound (relative to \optalg) shown by \cite{AnariGN14} in the
large market setting, translates to the same lower bound against \optbench in our setting.  
The lower bound is proved in~\cite{AnariGN14} under the assumption $\max_ec(e)\ll B$,
but they mention that it carries over to the setting $\vemax\ll\optalg$.
We include a proof of the following theorem in Appendix~\ref{append-impos} for completeness.

\begin{theorem}[Follows from~\cite{AnariGN14}] \label{optbenchlb}
No budget-feasible mechanism can achieve value better than
$\bigl(1-\frac{1}{e}\bigr)\cdot\optbench(v,B,c)$ on every instance, 
even with additive $v$ and additive cost functions. 
\end{theorem}

\section{XOS Valuations} \label{sec:XOS} \label{xos}
We design and analyze mechanisms for XOS valuations that achieve $O(1)$-approximation
ratio with respect to our benchmark $\optbench$. 
Section~\ref{xos-bfexp} describes a budget-feasible-in-expectation mechanism for 
general cost functions, 
and Section~\ref{xos-bfuni} focuses on universally budget-feasible mechanisms under the
no-overbidding assumption. We show how to drop the no-overbidding assumption in
Section~\ref{overbid}. 
As noted earlier, our mechanisms perform a VCG computation, which then defines the
payments made to the players as in the VCG mechanism (see Theorem~\ref{vcgthm} and
\eqref{vcgpayment}), so we describe the underlying algorithm, and discuss payments given
by \eqref{vcgpayment} in the analysis. 

Sections~\ref{xos-bfexp} and~\ref{xos-bfuni} can be read independently of each other. We
begin by collecting various properties of XOS and subadditive functions that we 
frequently utilize in our analyses.

\subsection{Properties of XOS and subadditive valuations} \label{sec-xosprops}

\begin{claim} \label{xosnprop}
Let $v:2^\gset\mapsto\R_+$ be an XOS valuation. Then, for any $S\sse\gset$, and any
partition $A_1,\ldots,A_r$ of $\gset$, we have 
we have $\sum_{i\in[r]}\bigl(v(S)-v(S-A_i)\bigr)\leq v(S)$.
\end{claim}
\begin{proof}
	Let $I\sse[r]$ be the indices $i$ for which $A_i\cap S\neq\es$.
	The stated inequality is equivalent to showing that 
	$\sum_{i\in I}\frac{v(S-A_i)}{|I|-1}\geq v(S)$. This inequality follows
	because taking $\mu_T=\frac{1}{|I|-1}$ for all $T\in\{S-A_i: i\in I\}$ yields a
	fractional cover of $S$.
\end{proof}

Our randomized mechanisms all use a random-sampling step to compute a good estimate
of $\optbench$. 
Let $\pl_1,\pl_2$ be a random partition of $\pl$ obtained by placing each player
independently with probability $\frac{1}{2}$ in $\pl_1$ or $\pl_2$. 
Let $\gset_j:=\bigcup_{i\in\pl_j}\pset_i$ for $j=1,2$ be the corresponding partition of
$\gset$ induced by $\pl_1,\pl_2$.
The idea is to use $\gset_1$ to obtain a good estimate, and work with this estimate for
$\gset_2$. 
Random partitioning has also been used in prior
work~\cite{BeiCGL12,AmanatidisKS19,NeogiPS24} for single-dimensional budget-feasible
mechanism design, where $\gset$ is directly 
partitioned by assigning each element to a part with probability $\frac{1}{2}$. 
While in the single-dimensional setting, players and items are synonymous, 
as noted earlier, in the multidimensional setting, to ensure
truthfulness, it is important to partition players 
and consider the partition 
of $\gset$ induced by this partition of
players, as defined above. 
We prove the following result.

\begin{lemma}[{\bf Random-partitioning lemma}] \label{lem:random_partition} 
\label{rpartition}
Let $g:2^\gset\mapsto\R_+$ be subadditive. 
\swamy{Define $\vbench[g](S):=\min_{i\in[k]}g(S-\pset_i)$ for $S\sse\gset$.} 
Consider any $S\sse\gset$, and let $S_1=S\cap\gset_1$, $S_2=S\cap\gset_2$. 
Let $\Omega$ be the event $\bigl\{g(S_2),g(S_1)\geq\frac{\vbench[g](S)}{4}\bigr\}$.
Then, (a) $\Pr[\Omega]\geq\frac{1}{2}$,
and (b) $\Pr\bigl[\{g(S_2)\geq\frac{g(S)}{2},\ g(S_2)\geq g(S_1)\}\cap\Omega\bigr]\geq\frac{1}{4}$. 
\end{lemma}

\begin{proof}
Let $I$ be a minimal prefix of $[k]$ such that, letting 
$A_1=\bigcup_{i\in I}(S\cap\pset_i)$, we have
$g(A_1)\geq\frac{\vbench[g](S)}{2}$. Let $A_2=S-A_1$. 
\swamy{Letting $\ell$ be the last index in $I$,  
we have $g(A_1-\pset_\ell)<\frac{\vbench[g](S)}{2}$. Also,  
$g(S-\pset_{\ell})\geq\vbench[g](S)$.
So since $g$ is subadditive, we have 
$g(A_2)\geq g(S-\pset_{\ell})-g(A_1-\pset_{\ell})>\frac{\vbench[g](S)}{2}$.}

Now fix a partition $A_1^H, A_1^L$ of $A_1$, and a partition $A_2^H, A_2^L$ of $A_2$,
where $g(A_1^H)\geq g(A_1)/2$ and $g(A_2^H)\geq g(A_2)/2$. We call $A_1^H$ and $A_2^H$,
the ``big sets''. 

The random partition $\gset_1,\gset_2$ induces random partitions of $A_1$ and $A_2$. 
Consider the event $\Gm$ that for both $\ell=1,2$, the random partition of $A_\ell$ induced
by $\gset_1,\gset_2$ is the same as the partition $A_\ell^H, A_\ell^L$, up to permutations
of the parts. 
That is, $\Gm$ is the event that 
$\gset_j\cap A_\ell\in\{A_\ell^H,A_\ell^L\}$ for $j=1,2$ and $\ell=1,2$. 
For any $j=1,2$, we have 
$\Pr\bigl[\gset_j\cap A_1=A_1^L,\ \gset_j\cap A_2=A_2^L\,|\,\Gm\bigr]=\frac{1}{4}$. So
conditioned on $\Gm$, with probability at least $\frac{1}{2}$, we have that both
$\gset_1,\gset_2$ contain some big set. Removing the conditioning yields part (a).

For part (b), we observe that $g(S_1)$, $g(S_2)$ are identically distributed, and this
remains true even when we condition on the event $\Omega$. 
It follows that $\Pr[g(S_2)\geq g(S_1)|\Omega]\geq\frac{1}{2}$. Also 
$g(S_1)+g(S_2)\geq g(S)$, so $g(S_2)\geq g(S_1)$ also implies that
$g(S_2)\geq\frac{g(S)}{2}$. 
\end{proof}

Throughout, for $j=1,2$, we use 
$V^*_j:=\max\,\{v(S):\ S\sse\gset_j,\ c(S)\leq B\}$ to denote the
optimal value that can be achieved using only players in $\pl_j$ and elements in
$\gset_j$. 

\begin{corollary} \label{rsample}
Let $v:2^\gset\mapsto\R_+$ be subadditive. 
We have 
$\Pr\bigl[V^*_2\geq\frac{\optalg(v,B,c)}{2},\ V^*_2\geq V^*_1\geq\frac{\optbench(v,B,c)}{4}\bigr]\geq\frac{1}{4}$.
\end{corollary}

\begin{proof}
Let $\optset\sse\gset$ be such that $\vbench(\optset)=\optbench(v,B,c)$.
Since $V^*_j\geq v(\optset\cap\gset_j)$ for $j=1,2$, applying Lemma~\ref{rpartition} (a)
to the set $S=\optset$, yields $\Pr\bigl[V^*_2,V^*_1\geq\frac{\optbench(v,B,c)}{4}\bigr]\geq\frac{1}{2}$.
$V^*_1$, $V^*_2$ are identically distributed, even conditioned on the above event, 
so $V^*_2\geq V^*_1$ with probability $1/2$, conditioned on this event.
Also $V^*_1+V^*_2\geq\optalg(v,B,c)$, so $V^*_2\geq V^*_1$ implies that $V^*_2\geq \frac{\optalg(v,B,c)}{2}$.
\end{proof}

We have an analogous result for the optimal value of the
LP-relaxation for the algorithmic problem of computing $\OPTalg$. 
For $A\sse\gset$, let $\lpopt(A)$ be the optimal value of $\optlp[(A)]$, which recall is
the following LP 
that can be solved in polytime using a demand oracle 
(Theorem~\ref{lpsolve}). 
\begin{equation}
\max \ \ \sum_{S\sse A}v(S)x_S \qquad \text{s.t.} \qquad
\sum_{S\sse A}c(S)x_S\leq B, \qquad \sum_{S\sse A}x_S\leq 1, \qquad x\geq 0. 
\tag*{$\optlp[(A)]$} 
\end{equation}
Throughout, for $j=1,2$, let $\lpopt_j$ denote the optimal value of $\optlp[(\gset_j)]$,
and let $\lpopt$ be the optimal value of $\optlp[(\gset)]$.

\begin{lemma} \label{rsample-lp}
Let $v:2^\gset\mapsto\R_+$ be subadditive.
With probability at least $1/4$, 
we have \swamy{$\lpopt_2\geq\lpopt_1\geq\frac{\optbench(v,B,c)}{4}$} 
and $\lpopt_2\geq\frac{\lpopt}{2}$. 
\end{lemma}

\begin{proof}
\swamy{Again, let $\optset\sse\gset$ be such that $\vbench(\optset)=\optbench(v,B,c)$.
Clearly, we have $\lpopt_j\geq V^*_j\geq v(\optset\cap\gset_j)$ for $j=1,2$. 
So by Lemma~\ref{rpartition} (a), we have 
$\Pr\bigl[\lpopt_2,\lpopt_1\geq\frac{\optbench(v,B,c)}{4}\bigr]\geq\frac{1}{2}$.
Conditioned on this event, with probability $\frac{1}{2}$, we have that
$\lpopt_2\geq\lpopt_1$, which also implies that $\lpopt_2\geq\lpopt/2$ since
$\lpopt_1+\lpopt_2\geq\lpopt$.}  
\end{proof}

In the analysis of our mechanisms, we will often need to demonstrate the existence of a
good-value set whose cost is bounded away from the budget, say, is at most $B/2$.
We obtain this by arguing that a good-value set satisfying the budget constraint can be
suitably pruned, as shown by Lemmas~\ref{bredn-gencost} and~\ref{bredn-supaddcost} 
for general cost functions and superadditive cost functions respectively.
Recall that $c(S):=\sum_{i\in[k]}c_i(S\cap\pset_i)$ for $S\sse\gset$. Note that the
cost-function $c$ inherits the properties of the $c_i$s: it is monotone and normalized, 
and if all $c_i$s are superadditive, then so is $c$.

\begin{lemma} \label{bredn-gencost}
Let $g:2^\gset\mapsto\R_+$ be subadditive. 
Let $\targ\in\R_+$, and 
$S\sse\gset$ be such that $c(S)\leq B$. 
We can
find $T\sse S$ such that $c(T)\leq B/2$ and 
$\min\{\vbench[g](S)-\targ,\targ-\max_{e\in S}g(e)\}<g(T)\leq\targ$.
\end{lemma}
\begin{proof}
	Let $\tht=\max_{e\in S}v(e)$. 
	Let $I$ be a minimal prefix of $[k]$ such that, letting 
	$S_1=\bigcup_{i\in I}(S\cap\pset_i)$, we have
	$c(S_1)>B/2$ or $g(S_1)\geq\targ$. Let $S_2=S-S_1$, and let $S'_1=S_1-\pset_\ell$,
	where $\ell$ is the last index in $I$.

	If $c(S_1)\leq B/2$, then we must have $g(S_1)>\targ$. Considering elements of $S_1$ in
	some fixed order, we take $T$ to be a maximal prefix of $S_1$ such that
	$g(T)\leq\targ$. Then we have $c(T)\leq B/2$ (since $T\sse S_1$), and the maximality of
	$T$ shows that $g(T)>\targ-\tht$.

	If $c(S_1)>B/2$, then $c(S_2)<B/2$.
	By the minimality of $I$, we have $g(S'_1)<\targ$. 
	We have \swamy{$g(S-\pset_\ell)\geq\vbench[g](S)$, 
	from the definition of $\vbench[g](S)$.}
	So \swamy{since $g$ is subadditive,} 
        $g(S_2)\geq g(S-\pset_\ell)-g(S'_1)>\vbench[g](S)-\targ$. So again taking $T$ to be a
	maximal prefix of $S_2$ with $g(T)\leq\targ$, we obtain that 
	$T=S_2$ and $g(T)>\vbench[g](S)-\targ$, or $g(T)>\targ-\tht$.
\end{proof}

\begin{lemma} \label{bredn-supaddcost}
Let $g:2^\gset\mapsto\R_+$ be subadditive, 
and $c:2^\gset\mapsto\R_+$ be superadditive.
Let $\targ\in\R_+$, and
$S\sse\gset$ be such that $c(S)\leq B$. 
We can find $T\sse S$ such that $c(T)\leq B/2$ and 
$\min\{g(S)-\targ,\targ\}-\max_{e\in S}g(e)<g(T)\leq\targ$.
\end{lemma}
\begin{proof}
	Let $\tht=\max_{e\in S}g(e)$.
	We proceed as in the proof of Lemma~\ref{bredn-gencost}, except that we can exploit
	superadditivity and do not need to consider whole player-sets when forming $S_1$ and
	$S_2$. Considering elements in some fixed order, we now take $S_1$ to be a minimal prefix
	of $S$ such that 
	$c(S_1)>B/2$ or $g(S_1)\geq\targ$. If $c(S_1)\leq B/2$, then we again take $T$ to be a
	maximal prefix of value at most $\targ$. Otherwise, we have
	$g(S_1)<\targ+\tht$. Letting $S_2=S-S_1$, we than have $g(S_2)>g(S)-\targ-\tht$.
	We again let $T$ be a maximal prefix of $S_2$ of value at most $\targ$, so that
	$T=S_2$ or $g(T)>\targ-\tht$. Also, $c(T)\leq c(S_2)$, and since $c$ is superadditive, we
	obtain that $c(S_2)\leq c(S)-c(S_1)<B/2$.
\end{proof}

\subsection{Budget-feasible-in-expectation mechanism for general cost
  functions} \label{xos-bfexp} 

Our budget-feasible-in-expectation mechanism uses random sampling to compute an estimate
of $\optbench$ from one part $(\pl_1,\gset_1)$, and utilizes this to solve an LP-relaxation of a
constrained demand-oracle query on $(\pl_2,\gset_2)$. By viewing the LP solution as a
distribution, and using VCG payments, we obtain the desired mechanism.
Recall that $\lpopt_j$ is the optimal value of $\optlp[(\gset_j)]$, for $j=1,2$.

\begin{procedure}[ht!] 
\caption{XOS-BFInExp() \textnormal{\qquad // budget-feasible-in-expectation mechanism for
    general costs} \label{xosalg-expgencost}}
\KwIn{Budget-feasible MD instance
$\bigl(v:2^\gset\mapsto\R_+,B,\{\pset_i,\C_i\sse\R_+^{2^{\pset_i}},c_i\in\C_i\}\bigr)$; 
$\ld\in[0,1]$}
\KwOut{subset of $\gset$; \quad payments are VCG payments, as specified in Lemma~\ref{bfinexp}}
\SetKwComment{simpc}{// }{}
\SetCommentSty{textnormal}
\DontPrintSemicolon

Partition $\pl$ into two sets $\pl_1$, $\pl_2$ by placing each player independently with
probability $\frac{1}{2}$ in $\pl_1$ or $\pl_2$.
Let $\gset_j:=\bigcup_{i\in\pl_j}\pset_i$ be the induced partition of $\gset$. 
Compute $V_1=\lpopt_1$ using a demand oracle.
\label{bfexp-optestim}

Obtain $x^*$ by solving the 
following constrained-demand-oracle LP, which is \eqref{cdlp} with 
$A=\gset_2$, $\kp=\frac{\ld V_1}{B}$, $\targ=\ld V_1$. 
\begin{equation}
\max\ 
\sum_{S\sse\gset_2}\bigl(v(S)-\tfrac{\ld V_1}{B}\cdot c(S)\bigr)x_S \quad \text{s.t.} \quad 
\sum_{S\sse\gset_2}v(S)x_S\leq\ld V_1, \quad \sum_{S\sse\gset_2} x_S\leq 1, \quad x\geq 0.
\label{cdemd} 
\end{equation} 
\; \label{bfexp-demandset} 

\vspace*{-4ex} Sample a set $T$ from the distribution $(x^*_S)_{S\sse\gset_2}$, 
and \Return $T$.
\label{bfexp-setoutput}
\end{procedure}

\begin{theorem} \label{xosdemdthm} \label{bfexp-thm}
Taking $\ld=0.5$ in Algorithm~\ref{xosalg-expgencost}, along with suitable payments,  
we obtain a budget-feasible-in-expectation mechanism that obtains expected value
at least $\frac{\optbench(v,B,c)}{\xosdemdapx}$.
The mechanism runs in polytime given a demand oracle. 
\end{theorem}

\begin{proof}
The polynomial running time follows from Theorem~\ref{lpsolve}, which shows that we can
solve LPs of the form $\optlp[(A)]$ and \eqref{cdlp} in polytime given a demand
oracle. 
Lemma~\ref{bfinexp} specifies the payments, shows that they can be computed efficiently,
and the resulting mechanism is budget-feasible in expectation.
Lemma~\ref{rsample-lp} shows that 
$\lpopt_2\geq\lpopt_1\geq\frac{\optbench(v,B,c)}{4}$ holds with probability at least
$\frac{1}{4}$. Assuming that this event happens, Lemma~\ref{bfexp-val} shows that we
obtain value at least $\frac{V_1}{4}\geq\frac{\optbench(v,B,c)}{16}$, therefore the
expected value returned at least $\frac{\optbench(v,B,c)}{64}$.
\end{proof}

\begin{lemma} \label{bfexp-val}
Suppose that $\lpopt_2\geq V_1$.
Then the optimal value of \eqref{cdemd} is at least 
$\ld V_1\bigl(1-\frac{\ld V_1}{\lpopt_2}\bigr)\geq\ld(1-\ld)V_1$. 
\end{lemma}

\begin{proof}
Let $\bx$ be an optimal solution to $\optlp[(\gset_2)]$. We use $S$ below to index over
subsets of $\gset_2$. 
So $\sum_S v(S)\bx_S=\lpopt_2$ and
$\sum_S c(S)\bx_S\leq B$. Consider $x'=\frac{\bx}{\lpopt_2/\ld V_1}$. 
We have $\sum_Sv(S)x'_S=\ld V_1$, so $x'$ is feasible to \eqref{cdemd}. 
We also have $\sum_S c(S)x'_S\leq\frac{\ld V_1}{\lpopt_2}\cdot B$. 
So the optimal value of \eqref{cdemd} is at least the objective value of $x'$, which is
\begin{equation*}
\sum_S v(S)x'_S-\frac{\ld V_1}{B}\cdot \sum_S c(S)x'_S
\geq \ld V_1-\frac{\ld V_1}{B}\cdot \frac{\ld V_1}{\lpopt_2}\cdot B
= \ld V_1\Bigl(1-\tfrac{\ld V_1}{\lpopt_2}\Bigr). \qedhere
\end{equation*}
\end{proof}

\begin{lemma} \label{bfinexp}
Given a demand oracle, one can efficiently compute payments that when combined with
Algorithm~\ref{xosalg-expgencost} yield a budget-feasible-in-expectation mechanism (that
is individually rational with probability $1$).
\end{lemma}

\begin{proof}
Fix a random partition $(\pl_1,\gset_1)$, $(\pl_2,\gset_2)$.
We first obtain expected payments for this random partition that yield truthfulness in
expectation and budget-feasibility in expectation. Given these, there is a standard way of 
obtaining actual payments for each random outcome of the mechanism that satisfy IR and NPT
with probability $1$. 

Observe that the computation in step~\ref{bfexp-demandset} amounts to a VCG computation over
the domain of feasible {\em fractional} solutions to \eqref{cdemd}. Thus, we can still use
expression \eqref{vcgpayment} to obtain expected payments for the players. 
Let $\kp=\frac{\ld V_1}{B}$.
The expected payment to player $i\in\pl_2$ when step~\ref{bfexp-demandset} returns the
fractional solution $x$ is given by 
$\sum_Sx_S\bigl(\frac{1}{\kp}\cdot v(S)-c(S-\pset_i)\bigr)-h_i(c_{-i})$, 
where $h_{-i}(c_i)$ is $\frac{1}{\kp}$ times the optimal value of \eqref{cdemd} when
player $i$ is excluded, that is, we are only allowed to use sets $S\sse\gset_2-\pset_i$.  
Note that $h_{-i}(c_{-i})$ can be calculated efficiently given a demand oracle, so these
expected payments can be computed efficiently.
Since the expected cost to player $i$ is $\sum_S c_i(S\cap\pset_i)x_S$, the expected
utility is $\frac{1}{\kp}\times(\text{objective value of $x$})-h_{-i}(c_{-i})$, which
is maximized by $x=x^*$, the outcome when player $i$ reports truthfully. This shows
truthfulness in expectation, and also shows that the expected utility under truthful
reporting is nonnegative. 
One feasible solution to \eqref{cdemd} when player $i$ is excluded is given by 
setting $x'_T=\sum_{S\sse\gset: S-\pset_i=T}x^*_S$ for all $T\sse\gset_2-\pset_i$, whose
objective value is 
$\sum_{S\sse\gset}v(S-\pset_i)x^*_S-\kp\cdot\sum_{\ell\in\pl_2-\{i\}}c_\ell(S\cap\pset_\ell)x^*_S$.
This yields a lower bound on $h_{-i}(c_{-i})$ and an upper bound of
$\frac{1}{\kp}\cdot\sum_S (v(S)-v(S-\pset_i))x^*_S$ on the payment made to player $i$.
Therefore, the total expected payment is at most 
\begin{equation}
\frac{B}{\ld V_1}\cdot\sum_S x^*_S\sum_{\ell\in\pl_1}(v(S)-v(S-\pset_\ell))
\leq \frac{B}{\ld V_1}\cdot\sum_S x^*_Sv(S) \leq B \label{bfexp-paymt}
\end{equation}
where the first inequality follows from Claim~\ref{xosnprop}.

We obtain payments for each random outcome $T$ of the mechanism as follows. 
Consider a player $i\in\pl_2$. Let $\mu_i$ be the expected payment to $i$, as computed
above. Let $\ecost[i]=\sum_S c_i(S\cap\pset_i)x^*_S$ be the expected cost incurred by $i$. 
We set the payment of $i$ under outcome $T$ to be 
$\frac{\mu_i}{\ecost[i]}\cdot c_i(T\cap\pset_i)$. 
(Here $c_i$ is $i$'s reported cost, which we may assume is her true cost, due to
truthfulness in expectation.)
Since $\mu_i\geq\ecost[i]$, this ensures that the payment is always the cost incurred by
$i$, and is $0$ if this cost is $0$.
\end{proof}

\begin{remark} \label{additive-bfexpremk}
With additive valuations and additive cost functions, observe that
Algorithm~\ref{xosalg-expgencost} runs in polynomial time, since the LPs
$\optlp[(\gset_1)]$ and \eqref{cdlp} can be solved in polytime.
\end{remark}

\subsection{Universally budget-feasible mechanisms assuming no-overbidding} 
\label{xos-bfuni}

We next describe how to obtain the stronger mechanism-design guarantee of universal
budget-feasibility.
We consider here the simpler setting where we assume no-overbidding; we show how to drop
this assumption in Section~\ref{overbid}. 
While 
the results here are subsumed by those
obtained in Section~\ref{overbid} (modulo $O(1)$ approximation factors), 
we discuss things first in the simpler setting of
no-overbidding as this will introduce many of the key underlying ideas, and the arguments
are simpler.  

Recall that no-overbidding imposes that for every player $i$, cost-function
$c_i\in\C_i$, element $e\in\pset_i$, we have $c_i(e)\leq B$.
Let $e^*=\argmax_{e\in\gset}v(e)$, and $\vemax$ denotes $v(e^*)$. 
Our mechanisms will exploit the fact that under
no-overbidding, the mechanism that returns $e^*$ and pays $B$ to the player who owns $e^*$
is a budget-feasible mechanism. (As noted earlier, without no-overbidding, the adaptation
where we return $\argmax_{e\in\gset}\,\{v(e): c(e)\leq B\}$ is not truthfully
implementable.) 

Algorithm~\ref{xosalg-gencost} describes the underlying algorithm for general cost
functions. It yields an $O(1)$-approximation with respect to a weaker benchmark than
$\optbench$, which we describe below. 
This serves chiefly as a warm-up for the more sophisticated constructions in
Section~\ref{overbid}, where we drop the no-overbidding assumption by suitably mdifying
portions of the algorithm (and obtain $O(1)$-approximation relative to $\optbench$).  
We then consider superadditive costs (Section~\ref{xos-supaddcost}), and show that with
some minor changes, 
we can achieve a stronger guarantee, namely, 
$O(1)$-approximation with respect to $\optbench$, and we can do so in polytime
given a constrained demand oracle. 

\begin{procedure}[ht!] 
\caption{XOS-UniBF() \textnormal{\qquad // universally budget-feasible mechanism for
    general costs} \label{xosalg-gencost}}
\KwIn{Budget-feasible MD instance
$\bigl(v:2^\gset\mapsto\R_+,B,\{\pset_i,\C_i\sse\R_+^{2^{\pset_i}}\},\{c_i\in\C_i\}\bigr)$; 
parameters $\ld\in[0,0.5]$, $p\in[0,1]$}
\KwOut{subset of $\gset$; \quad payments are VCG payments}
\SetKwComment{simpc}{// }{}
\SetCommentSty{textnormal}
\DontPrintSemicolon

Partition $\pl$ into two sets $\pl_1$, $\pl_2$ by placing each player independently with
probability $\frac{1}{2}$ in $\pl_1$ or $\pl_2$.
For $j=1,2$, let $\gset_j:=\bigcup_{i\in\pl_j}\pset_i$ give the induced partition of
$\gset$. \; 
\label{partition} 

Compute
$V_1=\max_{S\sse\gset_1}\,\{\vbench(S):\ c(S)\leq B\}$, 
the benchmark $\optbench$ associated with $\pl_1$ and item-set $\gset_1$. 
\;
\label{optestim}

Compute 
$S^* \gets \argmax_{S\sse\gset_2}\,\bigl\{v(S)-\frac{\ld V_1}{B}\cdot c(S):\ v(S)\leq\ld V_1\bigr\}$ 
using a constrained demand oracle.
\; \label{demandset}

\Return $S^*$ with probability $p$ and $e^*$ with probability $1-p$. \label{setoutput}
\end{procedure}

Algorithm~\ref{xosalg-gencost} is based on the template used in~\cite{NeogiPS24}
for budget-feasible mechanism design in the single-dimensional setting, 
and this works out because of the no-overbidding assumption. 
Without this, as discussed earlier, returning $e^*$ in step~\ref{setoutput} is no
longer viable and one needs a much-more sophisticated approach to find a workaround; 
so 
the mechanism and its analysis become more involved, especially when we seek
polytime guarantees using a constrained demand oracle (see Section~\ref{poly-obidsupadd}).   

\swamy{
To state the performance guarantee obtained by Algorithm~\ref{xosalg-gencost}, 
we introduce the following notation. 
For an integer $\ell\geq 1$, and function $g:2^\gset\mapsto\R_+$
define $\vbengen[g]{\ell}(S):=\min_{I\sse[k]:|I|\leq\ell}g\bigl(S-\bigcup_{i\in I}\pset_i\bigr)$
for $S\sse\gset$. Define  
$\optbench(\ell)=\optbench(\ell;v,B,c):=\max_{S\sse\gset}\,\{\vbengen{\ell}(S): c(S)\leq B\}$.
Thus, $\vbengen{\ell}(S)$ captures the value obtained from $S$ after we exlcude the
contribution from any collection of $\ell$ players, and
$\optbench(\ell)$ is the maximum value achievable within the budget upon excluding any
collection of $\ell$ players. 
}

\swamy{In Appendix~\ref{append-rpartition}, we prove
the following generalization of Lemma~\ref{rpartition} (a).

\begin{lemma} \label{gen-rpartition}
Let $g:2^\gset\mapsto\R_+$ be subadditive, and $\ell\geq 0$ be an integer.
Consider any $S\sse\gset$, and let $S_1=S\cap\gset_1$, $S_2=S\cap\gset_2$.
Then 
$\Pr\bigl[\vbengen[g]{\ell}(S_1),\vbengen[g]{\ell}(S_2)\geq\frac{\vbengen[g]{(4\ell+1)}(S)}{4}\bigr]\geq\frac{1}{2}$.
\end{lemma}}

\begin{theorem} \label{xosgenthm}
Taking $\ld=0.5$ and $p=0.8$ in Algorithm~\ref{xosalg-gencost}, together with suitable 
payments, we obtain a  
universally budget-feasible mechanism that obtains expected value at least
$\frac{1}{\xosgenapx}\cdot\optbench(5;v,B,c)$.
\end{theorem}

\begin{proof} 
Lemma~\ref{xosgen-budgetfeas} 
shows that we obtain a universally budget-feasible mechanism with suitable payments. 
We focus on proving the approximation guarantee. 

\swamy{
Let $\optset=\argmax_{S\sse\gset}\,\{\vbengen{5}(S):\ c(S)\leq B\}$.
Let $\optset_j=\optset\cap\gset_j$ for $j=1,2$.
Let $T^*_2:=\argmax_{S\sse\gset_2}\,\{\vbench(S):\ c(S)\leq B\}$,
and $V_2=\vbench(T^*_2)$
be the $\optbench$-benchmark associated with player-set $\pl_2$ and item-set $\gset_2$. 
Let $\Gm$ be the event that $\vbench(\optset_1),\vbench(\optset_2)\geq\vbengen{5}(\optset)/4$. 
Applying Lemma~\ref{gen-rpartition} to $\optset$, we obtain that $\Pr[\Gm]\geq 0.5$.
Note that $V_1$ and $V_2$ are identically distributed, and this remains true even when
we condition on $\Gm$.  
Therefore, we have that $\Pr[\{V_2\geq V_1\}\cap\Gm]\geq\frac{1}{4}$. 
Assume that this event happens. 

Applying Lemma~\ref{bredn-gencost} on $T^*_2$, we can obtain $T\sse T^*_2$ such 
that $c(T)\leq B/2$ and $\ld V_1-\vemax<v(T)\leq\ld V_1$.
It follows that
\begin{equation*}
v(S^*)\geq v(S^*)-\frac{\ld V_1}{B}\cdot c(S^*)
\geq v(T)-\frac{\ld V_1}{B}\cdot c(T)
\geq \ld V_1-\vemax-\ld V_1/2\geq \frac{\ld V_1}{2}-\vemax.
\end{equation*}
We also have 
$V_1\geq\vbench(\optset_1)\geq\frac{\vbengen{5}(\optset)}{4}$.

Putting everything together, and taking $\ld=\frac{1}{2}$, we obtain that 
the expected value returned is at least
\begin{equation*}
\frac{p}{4}\cdot\bigl(\tfrac{1}{16}\cdot\optbench(5)-\vemax\bigr)+(1-p)\vemax
\geq \frac{1}{80}\cdot\optbench(5)\geq
\frac{1}{80}\cdot\max_{\ve\in[0,1]}(1-5\ve)\optparam(\ve).
\qedhere
\end{equation*}}
\end{proof} 

\begin{lemma} \label{xosgen-budgetfeas}
There exist payments that when combined with Algorithm~\ref{xosalg-gencost} yield a
universally budget-feasible mechanism.
\end{lemma}

\begin{proof}
We consider each possible outcome of the random choices in Algorithm~\ref{xosalg-gencost}
and supply payments for which the resulting mechanism is budget feasible. If the outcome
is to return $e^*$ and $i$ is such that $e^*\in\pset_i$, then we pay $B$ to player $i$ and
$0$ to the other players. The utility of player $i$ is then $B-c_i(e^*)\geq 0$,
regardless of her reported cost function. So we trivially obtain truthfulness, IR, and
budget feasibility.

So suppose otherwise, and let $(\pl_1,\gset_1)$, $(\pl_2,\gset_2)$ be the partition
obtained in step~\ref{partition}. Let $\kp=\frac{\ld V_1}{B}$.
The computation in step~\ref{demandset} is a VCG
computation. So using \eqref{vcgpayment}, paying each player $i\in\pl_2$ the amount
\begin{equation*}
\frac{v(S^*)}{\kp}-\sum_{\ell\in\pl_2-\{i\}}c_\ell(S^*\cap\pset_\ell)-h_i(c_{-i}), 
\qquad \text{where}\ \ 
h_i(c_{-i})=\frac{1}{\kp}\cdot\max_{S\sse\gset_2-\pset_i}\,
\Bigl\{v(S)-\kp\cdot c(S):\ v(S)\leq\ld V_1\Bigr\}
\end{equation*}
and the other players $0$, yields truthfulness.
Under this payment, the utility of every player $i\in\pl_2$ is 
$\frac{1}{\kp}\cdot\bigl(v(S^*)-\kp\cdot c(S^*)\bigr)-h_{-i}(c_{-i})$, which is
nonnegative 
since $v(S^*)-\kp\cdot c(S^*)\geq v(S)-\kp\cdot c(S)$ for all $S\sse\gset_2$ with
$v(S)\leq\ld V_1$. Moreover, if $S^*\cap\pset_i=0$, then this utility, and hence payment
to $i$ are $0$. This shows IR and NPT.

Finally, since 
$h_i(c_{-i})\geq\frac{1}{\kp}\cdot\bigl(v(S^*-\pset_i)-\kp\cdot c(S^*-\pset_i)\bigr)$, 
the payment to $i\in\pl_2$ is at most
$\frac{1}{\kp}\cdot\bigl(v(S^*)-v(S^*-\pset_i)\bigr)$. 
We have $\sum_{i\in\pl_2}\bigl(v(S^*)-v(S^*-\pset_i)\bigr)\leq v(S^*)$ by
Claim~\ref{xosnprop}, and so the total payment is at most $B$.
\end{proof}

\subsubsection{Superadditive cost functions: polytime \boldmath $O(1)$-approximation with 
respect to $\optbench$} \label{xos-supaddcost}

With superadditive costs, we show that 
with minor tweaks to Algorithm~\ref{xosalg-gencost}, 
we can obtain a mechanism that runs in polytime given a constrained demand oracle, 
and achieves an $O(1)$-approximation with respect to $\optbench$ 
(Theorem~\ref{xossupaddpolythm}).
In step~\ref{optestim} of Algorithm~\ref{xosalg-gencost}, we now compute $V_1=\lpopt_1$,
and in step~\ref{demandset}, we compute
$S^* \gets \argmax_{S\sse\gset_2}\,
\bigl\{v(S)-\frac{\ld V_1+\vemax}{B}\cdot c(S):\ v(S)\leq\ld V_1+\vemax\bigr\}$.

As in Lemma~\ref{xosgen-budgetfeas}, we can obtain suitable payments that when combined
with this modified algorithm yield a
universally budget-feasible mechanism. 
The computations of $V_1$, $S^*$ and the payments, can all be done in polytime,
since we are given a constrained demand oracle. The only portion of the analysis that
changes more significantly is the proof of the approximation guarantee. 

\begin{theorem} \label{xossupaddpolythm}
Taking $\ld=\frac{2}{5}$ and $p=\frac{14}{15}$ in the above modified version of
Algorithm~\ref{xosalg-gencost}, together with suitable payments, we obtain a  
universally budget-feasible mechanism for superadditive costs with approximation ratio
$\xossupaddpolyapx$ that runs in polytime given a constrained demand oracle. 
\end{theorem}

\begin{proof}
As discussed above, we focus on proving the approximation guarantee.
Fix an input $(v,B,c)$, where the $c_i$s are superadditive.
We abbreviate $\optbench(v,B,c)$ to $\optbench$.
Lemma~\ref{rsample-lp} shows that 
$\lpopt_2\geq\lpopt_1\geq\frac{\optbench}{4}$, holds with probability at least
$\frac{1}{4}$. Assume that this event happens.

Let $\kp=\ld V_1+\vemax$. Let $\bx$ be an optimal solution to $\optlp[(\gset_2)]$. We use
$S$ below to index over subsets of $\gset_2$.
We transform $\bx$ into a fractional solution $x'$ such that $x'_S>0$ only if
$v(S)\leq\kp$, and then use $x'$ to obtain a lower bound on $v(S^*)$.  
Initialize $x'\gets 0$.
Consider a set $S$ with $\bx_S>0$. 
By repeatedly finding a maximal subset of $S$ having value at most $\kp$, and deleting
this, we obtain that as long as $v(S)>\kp$, the deleted subset 
has value strictly larger than $\ld V_1$. 
So we can partition $S$ into 
at most $\ceil{\frac{v(S)}{\ld V_1}}$ sets in this fashion, each having value at most
$\kp$; 
we increase the $x'$-value of each of these sets by $\bx_S$. 
Observe that 
$\sum_S x'_S\leq\sum_S\bx_S\ceil{\tfrac{v(S)}{\ld V_1}}
\leq\tfrac{\lpopt_2}{\ld V_1}+1$. 
Also, $\sum_S v(S)x'_S\geq\sum_S v(S)\bx_S=\lpopt_2$ since $v$ is subadditive, and
$\sum_S c(S)x'_S\leq\sum_S c(S)\bx_S\leq B$ since the costs are superadditive.
Since $x'_S>0$ only if $v(S)\leq\kp$, any such set is a candidate set for the constrained
demand oracle. So we have 
\begin{equation*}
\begin{split}
v(S^*)-\frac{\kp}{B}\cdot c(S^*)
& \geq\frac{1}{\sum_S x'_S}\cdot\sum_S x'_S\Bigl(v(S)-\tfrac{\kp}{B}\cdot c(S)\Bigr) 
\geq\frac{\lpopt_2-\kp}{\lpopt_2/\ld V_1+1} \\
& =\frac{\ld V_1}{1+\ld V_1/\lpopt_2}-\frac{\kp}{1+\lpopt_2/\ld V_1}
\geq\frac{\ld V_1}{1+\ld}-\frac{\kp}{1+1/\ld}
=V_1\cdot\frac{\ld(1-\ld)}{1+\ld}-\frac{\vemax}{1+1/\ld}
\end{split}
\end{equation*}
where the final inequality above follows since $\lpopt_2\geq\lpopt_1=V_1$.
We have $\frac{p}{4(1+1/\ld)}=1-p$, and
the expected value returned is at least 
$\frac{p}{4}\cdot v(S^*)+(1-p)\vemax\geq\optbench/100$.
\end{proof}

\begin{remark}
We can also tweak Algorithm~\ref{xosalg-gencost} differently, 
by taking $V_1=V^*_1$ in step~\ref{optestim} (and no other changes).
Recall that $V^*_j:=\max\,\{v(S):\ S\sse\gset_j,\ c(S)\leq B\}$ for $j=1,2$. 
This yields a slightly better $\xosgenapx$-approximation (the same factor as in
Theorem~\ref{xosgenthm}) relative to $\optbench$ for superadditive costs, but not in
polytime, since computing $V^*_1$ is \nphard even for additive valuations and costs. 
To see this, let $\ld=0.5$, $p=0.8$ (as in Theorem~\ref{xosgenthm}). 
As before, we can obtain suitable payments that when
combined with this modified algorithm yield a universally-budget-feasible mechanism. 
For the approximation guarantee, suppose that the event 
$V^*_2\geq V^*_1\geq\frac{\optbench(v,B,c)}{4}$ occurs, which happens with probability
$\frac{1}{4}$ (Corollary~\ref{rsample}). 
Let $T^*_2$ be an optimal solution to $\max\,\{v(S):\ S\sse\gset_2,\ c(S)\leq B\}$. 
So $c(T^*_2)\leq B$ and $v(T^*_2)=V^*_2\geq 2\ld V_1$.
Applying Lemma~\ref{bredn-supaddcost}, we can obtain $T\sse T^*_2$ such that $c(T)\leq B/2$
and $\ld V_1-\vemax<v(T)\leq\ld V_1$. We can now proceed as in the proof of
Theorem~\ref{xosgenthm} to obtain expected value at least $\optbench(v,B,c)/\xosgenapx$.
\end{remark}

\subsubsection{Polytime mechanism for additive valuations and additive
  costs} \label{additive} 
The above polytime mechanism for XOS valuations and superadditive
costs 
can be adapted to run in polynomial time (i.e., without any oracle) for additive
valuations and additive costs. 
With additive valuations and additive cost functions, a constrained demand oracle involves
solving a knapsack problem. We argue that we can instead work with a related knapsack
problem that can be solved efficiently using dynamic programming (DP), and thereby obtain
a polytime $O(1)$-approximation mechanism. 

This result is subsumed (modulo $O(1)$ approximation factors) by the result in
Section~\ref{poly-obidadd}, where we do not assume no-overbidding, 
but the mechanism and its analysis become much simpler assuming no-overbidding, so we
include this below. The reader
interested in the setting without assuming no-overbidding (for additive valuations and
additive costs) can directly skip to Section~\ref{poly-obidadd}. 

\begin{theorem} \label{additivethm}
There is a polytime mechanism for additive valuations and additive costs that achieves
approximation ratio $\additiveapx$ with respect to $\optbench$.
\end{theorem}

\begin{proof}
In the oracle polytime algorithm for superadditive costs described at the start of
Section~\ref{xos-supaddcost}, recall that 
we modify Algorithm~\ref{xosalg-gencost} by computing $V_1=\lpopt_1$ in
step~\ref{optestim}, and 
$S^*=\argmax_{S\sse\gset_2}\,\bigl\{v(S)-\kp\cdot c(S):\ v(S)\leq\kp B\bigr\}$ in
step~\ref{demandset}, 
where $\kp=\frac{\ld V_1+\vemax}{B}$. 
With additive $v$ and additive costs, $\optlp[(\gset_1)]$ 
can be solved in polytime, and setting $u_e=v(e)-\kp\cdot c(e)$ for all
$e\in\gset_2$, 
computing $S^*$ amounts to solving the knapsack problem,
maximize $u(S)$ subject to $v(S)\leq \kp B$, $S\sse\gset_2$.

We scale the $v(e)$'s to obtain polynomially-bounded weights, modifying the
$u_e$s and the knapsack budget correspondingly, to obtain a polytime-solvable knapsack
problem. 
Set $\bnew=\frac{n\ld V_1}{\vemax}+n$, and  
$w_e=\ceil{\frac{nv(e)}{\vemax}}$ and $u'_e=w_e-\frac{\bnew+n}{B}\cdot c(e)$ for all
$e\in\gset_2$.  
We now obtain $S^*$ by solving the knapsack problem over $\gset_2$, with
item-values $\{u'_e\}$, item weights $\{w_e\}$ and knapsack budget $\bnew+n$;  
since the weights lie in $\{0\}\cup[n]$ this takes polynomial time.
We return $S^*$ with probability $p$ and $e^*$ with probability $1-p$.

The modified knapsack problem minimizes an affine function of the player
costs. Since we solve this knapsack problem optimally, VCG again applies. So as in 
Lemma~\ref{xosgen-budgetfeas}, one obtains universal truthfulness, and the payment to each  
player $i\in\pl_2$ is at most $\frac{B}{\bnew+n}\cdot\bigl(w(S^*)- w(S^*\cap\pset_i)\bigr)$. 
Budget-feasibility follows since $w(S^*)\leq\bnew+n$; so we only need to analyze the
approximation ratio. 

We argue as in the proof of Lemma~\ref{xossupaddpolythm}. 
We may assume that $\lpopt_2\geq\lpopt_1\geq\frac{\optbench(v,B,c)}{4}$, which happens
with probability at least $\frac{1}{4}$.
Let $\bx$ be an optimal solution to $\optlp[(\gset_2)]$. 
We transform $\bx$ into a fractional solution $x'$ such that $x'_S>0$ only if
$v(S)\leq\kp$, and then use $x'$ to obtain a lower bound on $v(S^*)$.  
Initialize $x'\gets 0$.
Consider a set $S$ with $\bx_S>0$. 
We partition $S$ into at most $\ceil{\frac{w(S)}{\bnew}}$ sets, each having $w$-weight
at most $\bnew+n$, and increase $x'$-value of each of these sets by $\bx_S$. 
We then have
\begin{equation*}
\sum_S x'_S\leq\frac{\sum_S w(S)\bx_S}{\bnew}+1
\leq\frac{\tfrac{n}{\vemax}\cdot\lpopt_2+n}{\bnew}+1
=\frac{\tfrac{\lpopt_2}{\vemax}+1}{\tfrac{\ld V_1}{\vemax}+1}+1
\leq\frac{\lpopt_2}{\ld V_1}+1
\end{equation*}
where the last inequality follows since $\lpopt_2\geq\ld V_1$. 
Therefore, 
\begin{alignat*}{2}
& & w(S^*) & \geq u'(S^*)\geq\frac{\sum_S u'(S)x'_S}{\sum_S x'_S} 
\geq\frac{\sum_S w(S)x'_S-(\bnew+n)}{{\lpopt_2}/{\ld V_1}+1} \\
& && \geq\frac{\tfrac{n}{\vemax}\cdot\lpopt_2-(\bnew+n)}{\tfrac{\lpopt_2}{\ld V_1}+1}
=\frac{n}{\vemax}\cdot\frac{\lpopt_2-\ld V_1}{\tfrac{\lpopt_2}{\ld V_1}+1}
-\frac{2n}{\tfrac{\lpopt_2}{\ld V_1}+1} \\
\implies\ && v(S^*) & \geq
\frac{\lpopt_2-\ld V_1}{\tfrac{\lpopt_2}{\ld V_1}+1}
-\vemax\Bigl(\tfrac{2}{\tfrac{\lpopt_2}{\ld V_1}+1}+1\Bigr) \\
& && \geq V_1\cdot\frac{\ld(1-\ld)}{1+\ld}-\vemax\cdot\frac{3\ld+1}{\ld+1}.
\end{alignat*}
We take $\ld=0.5$ and $p=\frac{12}{17}$ so that 
$\frac{p}{4}\cdot\frac{3\ld+1}{\ld+1}=1-p$. 
The expected value returned is at least 
$\frac{p}{4}\cdot v(S^*)+(1-p)\vemax\geq\optbench(v,B,c)/272$.
\end{proof}

\section{Dropping the no-overbidding assumption} \label{overbid}
We now describe how to obtain universally-budget-feasible mechanisms for XOS valuations
without assuming no-overbidding.
Recall that $\pl=[k]$ is the set of all players.
Now define $e^*$ to be $\argmax\,\{v(e): e\in\gset,\ c(e)\leq B\}$, 
and let $\vemax=v(e^*)$.
Let 
$\opt^*=\max_{i\in\pl}\max_{S\sse\pset_i}\,\bigl\{v(S): c_i(S)\leq B\bigr\}$. 

As discussed in Section~\ref{overview}, dropping this assumption entails
figuring out a way of 
offsetting the additive loss 
incurred in arguing the existence of a large-value set $T\sse\gset_2$ with 
$v(T)\leq\ld V_1$, $c(T)\leq B/2$.
While this additive loss is bounded by $\opt^*$, as noted earlier, the mechanism that
returns $S\sse\pset_i$, for some player $i$, with $v(S)=\opt^*$, $c(S)\leq B$, is not truthfully
implementable.
\footnote{Concretely, suppose $B=4$, there is one player, $\gset=\{e,f\}$, $v$ is additive
with $v(e)>v(f)$, and we have two additive cost functions $c^{(1)}$, $c^{(2)}$, given by 
$c^{(1)}_e=B,\ c^{(1)}_f=1$ and $c^{(2)}_e=B+1$, $c^{(2)}_f=B$. When the true cost
function is $c^{(1)}$, truthful reporting would yield utility $0$, whereas reporting
$c^{(2)}$ would yield utility $B-1>0$.}
So we need to devise an alternative to this ``return (set corresponding to) $\opt^*$
mechanism.''
We devise such a mechanism in Section~\ref{optmax-mech}, 
and show in Section~\ref{unibf-overbid} how the analyses of algorithms from
Section~\ref{xos-bfuni} can be modified so as to leverage the guarantee of this mechanism 
and obtain approximation guarantees relative to $\optbench$.

\subsection{Truthful mechanism to offset additive loss}
\label{optmax-mech}

Recall that given cost functions $\{c_i\in\C_i\}$, we define 
$c(S):=\sum_i c_i(S\cap\pset_i)$ for all $S\sse\gset$.

\SetAlgoProcName{Mechanism}{Mechanism}
\begin{procedure}[ht!] 
\caption{2ndOpt() \textnormal{\qquad // budget-feasible mechanism: substitute for
    returning $e^*$} \label{optmax-proxy}} 
\KwIn{Budget-feasible MD instance
$\bigl(v:2^\gset\mapsto\R_+,B,\{\pset_i,\C_i\sse\R_+^{2^{\pset_i}},c_i\in\C_i\}\bigr)$}
\KwOut{subset of $\pset_i$ for some player $i$; \quad payment = $B$}
\SetKwComment{simpc}{// }{}
\SetCommentSty{textnormal}
\DontPrintSemicolon

For every player $i$, define 
$\opt_i:=\max\,\bigl\{v(S): S\sse\pset_i,\ c(S)\leq B\bigr\}$.
Also, for each player $i$, define
$\Sc_i=\Sc_i(c_{-i})=\{S\sse\pset_i: v(S)\geq\max_{j<i}\opt_j,\ v(S)>\max_{j>i}\opt_j\}$.
\label{optmax-init}

Choose $\hS=\argmin\,\{c(S): S\in\bigcup_{i\in\pl}\Sc_i\}$, 
and let $\hS\in\Sc_{\hi}$. 
(We show in the analysis that $\bigcup_i\Sc_i\neq\es$.) \label{optmax-set}

\Return $\hS$, and pay $B$ to player $\hi$. \label{optmax-output}
\end{procedure}
\SetAlgoProcName{Algorithm}{Algorithm}

\begin{lemma} \label{new-welldefined}
There exists some player $i'$ 
and some $S'\in\Sc_{i'}$ such that $c_{i'}(S')\leq B$.
\end{lemma}

\begin{proof}
Let $i'$ be the player with largest index such that 
$\opt_{i'}=\max_{i\in\pl}\opt_i$,  
and let 
$S'\sse\pset_{i'}$ be such that $v(S')=\opt_{i'}$ and $c_{i'}(S')\leq B$. 
By definition then, we have $v(S')\geq\opt_j$ for all $j<i$, and
$v(S')>\opt_j$ for all $j>i'$. So $S'\in\Sc_{i'}$. 
\end{proof}

\begin{theorem} \label{optmax-thm}
Mechanism~\ref{optmax-proxy} is budget-feasible, and its output $\hS\sse\pset_{\hi}$
satisfies 
$v(\hS)\geq\max_{i\neq\hi}\opt_i=\max\,\{v(S):S\sse\pset_i\text{ for some }i\neq\hi,\ c(S)\leq B\}$.
\end{theorem}

\begin{proof}
Lemma~\ref{new-welldefined} shows that the mechanism is well defined.
The performance-guarantee statement follows from construction, since we have 
$\hS\in\Sc_{\hi}$. 
The payment made is $B$ by construction.
Individual rationality follows from Lemma~\ref{new-welldefined}, since this implies that
there is some $S\in\bigcup_i\Sc_i$ with $c(S)\leq B$,
We focus on proving truthfulness.

The key observation is that a player $i$ cannot affect her collection of sets $\Sc_i$. 
We first claim that player $\hi$ cannot benefit by lying. 
We have that $\hS$ is a minimum-cost set from $\Sc_{\hi}$. 
So since $c_{\hi}$ does not affect $\Sc_{\hi}$, player $\hi$ cannot lie and cause a lower  
$c_{\hi}$-cost subset of $\pset_{\hi}$ to be chosen.

Next, consider a player $i\neq\hi$. We show that $c_i(S_i)>B$ for every $S\in\Sc_i$,  
which implies that $i$ cannot obtain positive utility by reporting some $c'\in\C_i$, 
$c'\neq c_i$. 
Again, $\Sc_i$ does not depend on player $i$'s reported cost (so we can
unambiguously say $\Sc_i$ given that $c_{-i}$ is fixed).
Suppose $i<\hi$. We have $v(\hS)\geq\opt_i$, and so $v(\hS)\geq v(S)$ for
every $S\sse\pset_i$ with $c_i(S)\leq B$. Therefore, if $S\sse\pset_i$ is such that
$v(S)>\opt_{\hi}\geq v(\hS)$, then we must have $c_i(S)>B$.
Similarly, suppose $i>\hi$. Then, we have $v(\hS)>\opt_i$, and so
$v(\hS)>v(S)$ for every $S\sse\pset_i$ with $c_i(S)\leq B$. So if $S\sse\pset_i$ satisfies 
$v(S)\geq\opt_{\hi}\geq v(\hS)$, then we must again have $c_i(S)>B$.
\end{proof}

We obtain the following immediate corollary. 
Let $\opt^{(2)}$ be the second-largest $\opt_i$ value. 

\begin{corollary} \label{optmax-cor}
The set $\hS$ returned by Mechanism~\ref{optmax-proxy} satisfies $v(\hS)\geq\opt^{(2)}$.
\end{corollary}

\begin{remark}
We remark that the asymmetry in the definition of $\Sc_i$ in Mechanism~\ref{optmax-proxy} is
crucial. The proof evidently exploits this, and we can show that if we change the
definition of $\Sc_i$ to 
$\bigl\{S\sse\pset_i: v(S)\geq\max_{j\neq i}\opt_j\bigr\}$, then we do
not obtain truthfulness.
\end{remark}

\begin{remark}[{\bf Computation using oracles}] \label{knapcover}
Computing the set $\hS$ requires two types of oracles: an oracle for computing the
$\opt_i$ quantities, and a 
{\em knapsack-cover oracle}~\cite{NeogiPS24} 
to find a minimum-cost set in $\Sc_i$ for every player $i$, in step~\ref{optmax-init}. 
In the multidimensional setting, a 
{\em knapsack-cover oracle for the class $\C=\Pi_i\C_i$} takes $q\in\C$ and a target 
value $\targ$ as input, and returns
$\argmin_{S\sse\gset}\,\bigl\{q(S): v(S)\geq\targ\bigr\}$, or \infeas, if no feasible set
exists.  
As with a demand- and constrained- demand oracle, we assume, somewhat more generally, that 
we can specify a subset $A\sse\gset$ of the form $A=\bigcup_{i\in I}\pset_i$,
for some $I\sse[k]$, and the oracle returns the optimum over $A$ (as opposed to $\gset$).
\footnote{This can be achieved by taking $q_\ell$ for $\ell\notin I$ to be the constant
function $q_\ell(S)=M$ for all $\es\neq S\sse\pset_\ell$, where $M$ is sufficiently large,
say $2|A|\cdot\max_{i\in A}q_i(\pset_i)$. 
We call the knapsack-cover oracle with this modified $q$, and return the oracle's output
if the output is \infeas, or a set $S\sse A$; otherwise, we return \infeas.} 

Given a knapsack-cover oracle, we can compute $\hS$ as follows. 
By scaling, we may assume that $v(S)$ is an integer for all $S\sse\gset$. 
For every player $i$, we compute $S^*_i=\argmin_{S\in\Sc_i} c_i(S)$
by calling the knapsack-cover oracle with the set
$A=\pset_i$, $q=c$, and target value
$\targ_i=\max\{\max_{j<i}\opt_j,1+\max_{j>i}\opt_j\}$.
We then return $\hS=\argmin\,\{c_i(S^*_i): i\in\pl\}$. 
\end{remark}

\subsection{Universally-budget-feasible mechanisms without assuming no-overbidding}
\label{new-sampling} \label{unibf-overbid}
We now utilize Mechanism~\ref{optmax-proxy} (and variants of it) in conjunction with 
mechanisms from Section~\ref{xos-bfuni} (with some changes) to obtain
{\em universally-budget feasible mechanisms that achieve $O(1)$-approximation with respect
to $\optbench$, without assuming no-overbidding}. 

Recall that $\opt_i=\max\,\bigl\{v(S): S\sse\pset_i,\ c(S)\leq B\bigr\}$, and $\opt^{(2)}$
is the second-largest $\opt_i$ value. 
Examining the analysis of Algorithm~\ref{xosalg-gencost} in Theorem~\ref{xosgenthm}, 
we see that one of the places where we incur a loss in value is when we argue the
existence of a large-value set of cost at most $B/2$, by applying Lemma~\ref{bredn-gencost}
to a suitable set $S\sse\gset_2$ (the set $T^*_2$ in the proof of
Theorem~\ref{xosgenthm}): this incurs a loss bounded by roughly 
$O\bigl(\max_{i\in\pl_2}v(S\cap\pset_i)\bigr)$ (due to the $\vbench[g](S)$ term in the
statement of Lemma~\ref{bredn-gencost}). 
The key insight is that we can now recover this loss using Mechanism~\ref{optmax-proxy}, 
provided that we can engineer things {\em in the analysis} so that
$\max_{i\in\pl_2}v(S\cap\gset_i)$ is at most $\opt^{(2)}$, and $v(S)$ is large. This  
requires a careful application of the random-partitioning lemma (Lemma~\ref{rpartition})
coupled with some additional observations.
The procedure 
is described in Algorithm~\ref{xosalggen-overbid} below, and follows the same template as
in prior algorithms, but we use Mechanism~\ref{optmax-proxy} in place of the
``return-$e^*$ mechanism.''  

In Section~\ref{poly-obidadd}, we consider additive valuations and additive costs, 
and obtain a {\em polytime} $O(1)$-factor approximation, universally budget-feasible
mechanism for such instances. 
Here, we exploit the fact that with additive valuations and costs, various computations in
Mechanism~\ref{optmax-proxy} and the mechanism from Section~\ref{xos-bfuni}
amount to solving knapsack or knapsack-cover problems, and we can make
suitable changes to these mechanisms 
to move to related problems that can be solved efficiently using dynamic programming. 
In Section~\ref{poly-obidsupadd}, we consider general XOS valuations and superadditive
cost functions, and devise polytime mechanisms given access to a constrained demand
oracle. The crucial (and only) change here compared to Algorithm~\ref{xosalggen-overbid}
lies in coming up with a different ``version'' of Mechanism~\ref{optmax-proxy} that can be
implemented in polytime using a constrained demand oracle.

\begin{procedure}[ht!] 
\caption{XOS-Gen() \textnormal{\quad // universally budget-feasible mechanism: general
costs without no-overbidding} 
\label{xosalggen-overbid}}
\KwIn{Budget-feasible MD instance
$\bigl(v:2^\gset\mapsto\R_+,B,\{\pset_i,\C_i\sse\R_+^{2^{\pset_i}}\},\{c_i\in\C_i\}\bigr)$; 
parameters $\ld\in[0,0.5]$, $p\in[0,1]$}
\KwOut{subset of $\gset$; \quad payments are VCG payments}
\SetKwComment{simpc}{// }{}
\SetCommentSty{textnormal}
\DontPrintSemicolon

Independently, for each player $i\in\pl$, place $i$ in $\pl_1$ or $\pl_2$, each with
probability $\frac{1}{2}$.
For $j=1,2$, let $\gset_j:=\bigcup_{i\in\pl_j}\pset_i$. \label{partition-overbid} 

Compute $V_1=V^*_1:=\max_{S\sse\gset_1}\,\{v(S):\ c(S)\leq B\}$. 
\; \label{optestim-overbid}

Use a constrained demand oracle to obtain 
$S^* \gets \argmax_{S\sse\gset_2}\,
\bigl\{v(S)-\frac{\ld V_1}{B}\cdot c(S):\ v(S)\leq \ld V_1\bigr\}$. 
\label{demandset-overbid}

\Return $S^*$ with probability $p$ and the output of Mechanism~\ref{optmax-proxy} with
probability $1-p$. \label{setoutput-overbid} 
\end{procedure}

\begin{theorem} \label{xosgenthm-overbid}
Taking $\ld=0.5$ and $p=\frac{128}{145}$ in 
Algorithm~\ref{xosalggen-overbid}, together with suitable  
payments, we obtain a
universally budget-feasible mechanism that obtains expected value at least
$\frac{1}{\xosgenapxobid}\cdot\optbench(v,B,c)$.
\end{theorem}

\begin{proof} 
Universal budget-feasibility follows from the same arguments as in the proof of
Lemma~\ref{xosgen-budgetfeas}, so we focus on proving the approximation guarantee.
We may assume that $n\geq 2$, otherwise $\optbench=0$, and there is nothing to be shown.

We follow a similar approach as in the proof of Theorem~\ref{xosgenthm}. 
As alluded to earlier, we carefully identify a set $S\sse\gset_2$ with $c(S)\leq B$
such that when we apply Lemma~\ref{bredn-gencost} to $S$ to extract a subset of cost at most
$B/2$ and value roughly $\ld V_1$, the additive loss incurred (relative to $\ld V_1$) can
be charged to $\opt^{(2)}$.  
More precisely, we identify such a set $S$ with 
$v(S)=\frac{1}{O(1)}\cdot\bigl(\optbench-\opt^{(2)}\bigr)$ assuming a
certain good event happens, and argue that this good event happens with constant probability. 
Coupled with the $\opt^{(2)}$ value returned by Mechanism~\ref{optmax-proxy}, this will
yield the stated approximation bound.

Let $i^*\in\pl$ be such that $\opt_{i^*}=\max_{i\in\pl}\opt_i$, and let
$\pset^*=\pset_{i^*}$. 
Essentially, the idea is that since $\opt_{i^*}$ is the only $\opt_i$ quantity that we
cannot recover using Mechanism~\ref{optmax-proxy}, we simply consider the effect of random
partitioning after excluding $i^*$, noting that excluding $\pset^*$ only incurs a bounded
loss in value. 
Let $\optset=\argmax_{S\sse\gset}\,\{\vbench(S): c(S)\leq B\}$.
Define $\boptset=\optset-\pset^*$, and
$\boptset_j=\boptset\cap\gset_j$ for $j=1,2$. 
Note that $v(\boptset)\geq\vbench(\optset)=\optbench$. 
Also, for any $S\sse\gset-\pset^*$ with $c(S)\leq B$, we have $\vbench(S)\geq v(S)-\opt^{(2)}$, 
since by definition, for every $i\in\pl$ with $S\cap\pset_i\neq\es$, we have $i\neq i^*$
and so $v(S\cap\pset_i)\leq\opt_i\leq\opt^{(2)}$. 
In particular, we have $\vbench(\boptset)\geq\optbench-\opt^{(2)}$.

Let $\Gm$ be the event that $v(\boptset_1),v(\boptset_2)\geq\vbench(\boptset)/4$.
Applying Lemma~\ref{rpartition} (a) to $\boptset$, we obtain that $\Pr[\Gm]\geq 0.5$.
Let $V'_j=\max\,\{v(S): S\sse\gset_j-\pset^*,\ c(S)\leq B\}$ for $j=1,2$. 
Let $T'_2\sse\gset_2-\pset^*$ be such that $v(T'_2)=V'_2$ and $c(T'_2)\leq B$.
Note that $V'_1$ and $V'_2$
are identically distributed, and this remains true even when we condition on $\Gm$. So we
have $\Pr[\{V'_2\geq V'_1\}\cap\Gm]\geq\frac{1}{4}$. Finally, let $\Om$ be the event
$\pset^*\sse\gset_2$. Clearly, $\Pr[\Om]=0.5$.
Event $\Om$ depends only on the random choice made for player
$i^*$, whereas event $\Gm$ and the random variables $V'_1, V'_2$ depend only on the random
choices made for the other players. So $\Om$ and the event $\{V'_2\geq V_1\}\cap\Gm$ are
independent, and we have $\Pr[\{V'_2\geq V'_1\}\cap\Gm\cap\Om]\geq\frac{1}{8}$.

Let us condition on the good event $\{V'_2\geq V'_1\}\cap\Gm\cap\Om$. Note then that
since $\pset^*\cap\gset_1=\es$, we have $V'_1=V_1=V^*_1$. 
Applying Lemma~\ref{bredn-gencost} on $T'_2$ with $\targ=\ld V_1$, we can obtain 
$T\sse T'_2$ such that $c(T)\leq B/2$ and 
$\min\{\vbench(T'_2)-\ld V_1,\ld V_1-\max_{e\in T'_2}v(e)\}<v(T)\leq\ld V_1$.
We also have $\vbench(T'_2)\geq V'_2-\opt^{(2)}$, 
and $\max_{e\in T'_2}v(e)\leq\max_{i\in\pl_2}v(T'_2\cap\pset_i)\leq\opt^{(2)}$.
Since $V'_2\geq 2\ld V_1$, it follows that $v(T)>\ld V_1-\opt^{(2)}$.
So we have
\begin{equation*}
v(S^*) \geq v(S^*)-\frac{\ld V_1}{B}\cdot c(S^*)
\geq v(T)-\frac{\ld V_1}{B}\cdot c(T)
\geq \ld V_1-\opt^{(2)}-\frac{\ld V_1}{2}
\geq\frac{\ld V_1}{2}-\opt^{(2)}.
\end{equation*}

We also have 
\[
V_1\geq v(\boptset_1)\geq\frac{\vbench(\boptset)}{4}\geq\frac{\optbench-\opt^{(2)}}{4}.
\]
By Corollary~\ref{optmax-cor}, the value obtained by Mechanism~\ref{optmax-proxy} is at
least $\opt^{(2)}$.
Putting everything together,  
we obtain that the expected value returned is at least 
\begin{alignat}{1}
p\cdot\frac{1}{8}\cdot\biggl(\frac{1}{16}\cdot\bigl(\optbench-\opt^{(2)}\bigr)&-\opt^{(2)}\biggr)
+(1-p)\opt^{(2)} \label{finalineq-overbid} \\
&=\frac{p}{128}\cdot\optbench=\frac{1}{\xosgenapxobid}\cdot\optbench.
\notag
\qedhere
\end{alignat}
\end{proof}

\subsubsection{Polytime mechanism for additive valuations and additive costs}
\label{polytime-overbid} \label{poly-obidadd}
With additive valuations and additive cost functions, a constrained demand oracle involves
solving a knapsack problem. The computation of $\opt_i$'s in Mechanism~\ref{optmax-proxy}
also amounts to solving knapsack problems.
We argue that, using scaling and rounding, we can instead work with 
related knapsack problems that can be solved efficiently using dynamic programming (DP),
and thereby obtain a polytime $O(1)$-approximation mechanism.
But we need to exercise some care, because we cannot use $\vemax$ in the scaling. 

We discuss below the changes to Algorithm~\ref{xosalggen-overbid} and
Mechanism~\ref{optmax-proxy}; Algorithm~\ref{additivealg-overbid} contains the entire 
description. 
\begin{enumerate}[label=\arabic*., topsep=0.2ex, itemsep=0.1ex, leftmargin=*]
\item In step~\ref{optestim-overbid} of Algorithm~\ref{xosalggen-overbid}, we now use
any $\beta$-approximation algorithm for 
the knapsack problem, where $\beta\geq 1$, to obtain $V_1\geq V^*_1/\beta$.  

\item 
In step~\ref{demandset-overbid} of Algorithm~\ref{xosalggen-overbid}, 
we compute $S^*$ by solving the following roughly-equivalent knapsack problem.
Set $\bnew=4n$, and  
$w_e=\ceil{\frac{nv(e)}{\ld V_1}}$ and $u'_e=w_e-\frac{\bnew}{B}\cdot c(e)$ for all 
$e\in\gset_2$. 
Solve the knapsack problem over $\gset_2$, with
item-values $\{u'_e\}$, item weights $\{w_e\}$ and knapsack budget $\bnew$, to obtain
$S^*$. This takes polytime since $\bnew=4n$. 
(Without the $\ceil{.}$ in the $w_e$s, this would be the same as the problem 
$\max_{S\sse\gset_2}\,\bigl\{v(S)-\frac{4\ld V_1}{B}\cdot c(S): v(S)\leq 4\ld V_1\bigr\}$.)

\item Dealing with the (efficient computation of the) $\opt_i$ quantities in 
Mechanism~\ref{optmax-proxy} poses more difficulties, and the workaround is more  
involved. 
The issue is that 
the truthfulness guarantee of Mechanism~\ref{optmax-proxy} relies crucially on the fact
that we are working with the exact $\opt_i$ quantities. To make this more
approximation-friendly, we use random partitioning (again!). We obtain an estimate of
$\max_i\opt_i$ from one part, for which we can use any approximation algorithm for
computing the $\opt_i$s. For each player $i$ in the second part, we solve the 
{\em knapsack-cover problem} 
of finding a min-cost set
$T^*_i\sse\pset_i$ whose value is at least this estimate. 
Finding the smallest-index player $i$ (from the second part) for which $c_i(T^*_i)\leq B$,
if one exists, returning $T^*_i$ and paying $B$ to player $i$, 
yields a truthful mechanism with a guarantee similar to that stated in
Theorem~\ref{optmax-thm} for Mechanism~\ref{optmax-proxy}. Finally, since the valuation
and costs are additive, 
we use scaling and rounding to 
solve a related knapsack-cover problem in polynomial time, which suffices.

This modified version of Mechanism~\ref{optmax-proxy} is described
below. 

\SetAlgoProcName{Mechanism}{Mechanism}
\begin{procedure}[ht!] 
\caption{2ndOpt-Poly() 
\label{newoptmax-proxy}} 
\KwIn{Budget-feasible MD instance
$\bigl(v\in\R_+^{\gset},B,\{\pset_i,c_i\in\R_+^{\pset_i}\}\bigr)$ with additive valuation
and additive costs}
\KwOut{subset of $\pset_i$ for some player $i$; \quad payment = $B$}
\SetKwComment{simpc}{// }{}
\SetCommentSty{textnormal}
\DontPrintSemicolon

For every player $i$, define $\opt_i:=\max_{S\sse\pset_i}\,\bigl\{v(S): c_i(S)\leq B\bigr\}$.
Partition $\pl$ into two sets $\pl_1,\pl_2$ by placing each player independently with
probability $\frac{1}{2}$ in $\pl_1$ or $\pl_2$.

For every $i\in\pl_1$, compute $\opt'_i$, a $\gm$-approximation to
$\opt_i$. 
Set $\targ:=\max_{i\in\pl_1}\opt'_i$.

For every $i\in\pl_2$, do the following.
Set $\wt_e=\floor{\frac{2n v_e}{\targ}}$ for all $e\in\pset_i$. 
Let $T^*_i\sse\pset_i$ be an optimal solution to the following knapsack-cover problem: 
$\min_{S\sse\pset_i}\,\{c_i(S): \wt(S)\geq n\}$. 
\label{newproxy-knapcover}

If $c_i(T^*_i)>B$ for all $i\in\pl_2$, {\bf return} $\es$ and make no payment to any player. 
Otherwise, let $\hi\in\pl_2$ be the smallest index $i\in\pl_2$ such that $c_i(T^*_i)\leq B$;
\Return $\hS=T^*_{\hi}$ and pay $B$ to player $\hi$. 
\label{newproxy-output}
\end{procedure}
\SetAlgoProcName{Algorithm}{Algorithm}
\end{enumerate}

We run steps~\ref{partition-overbid}--\ref{demandset-overbid} of
Algorithm~\ref{xosalggen-overbid} with the above changes, and return $S^*$ or the output of
Mechanism~\ref{newoptmax-proxy}, each with some probability. The entire algorithm is
descibed below.

\begin{procedure}[ht!] 
\caption{Additive() \textnormal{\quad // polytime universally budget-feasible mechanism:
 additive valuation and costs, without no-overbidding} 
\label{additivealg-overbid}}
\KwIn{Budget-feasible MD instance
$\bigl(v\in\R_+^{\gset},B,\{\pset_i,c_i\in\R_+^{\pset_i}\}\bigr)$ with additive valuation
and additive costs; 
parameters $\ld\in[0,0.5]$, $p\in[0,1]$, $r\in\Z_+, r\geq 2$.}
\KwOut{subset of $\gset$; \quad payments are VCG payments}
\SetKwComment{simpc}{// }{}
\SetCommentSty{textnormal}
\DontPrintSemicolon

Independently, for each player $i\in\pl$, place $i$ in $\pl_1$ or $\pl_2$, each with
probability $\frac{1}{2}$.
For $j=1,2$, let $\gset_j:=\bigcup_{i\in\pl_j}\pset_i$. 

Compute $V_1$, a $\beta$-approximation to $V^*_1:=\max_{S\sse\gset_1}\,\{v(S):\ c(S)\leq B\}$. 

Set $w_e=\ceil{\frac{nv(e)}{\ld V_1}}$ and 
$u'_e=w_e-\frac{\bnew}{B}\cdot c(e)$ for all $e\in\gset_2$. 
Compute an optimal solution $S^*$ to the knapsack problem over $\gset_2$, with
item-values $\{u'_e\}$, item weights $\{w_e\}$ and knapsack budget $\bnew=4n$.
\label{demandset-additive}

\Return $S^*$ with probability $p$ and the output of Mechanism~\ref{newoptmax-proxy} with
probability $1-p$. 
\end{procedure}

Theorem~\ref{newoptmax-thm} states the performance guarantee of Mechanism~\ref{newoptmax-proxy}. 
We first prove that Mechanism~\ref{additivealg-overbid} leads to an $O(1)$-approximation,
universally budget-feasible mechanism assuming Theorem~\ref{newoptmax-thm}, 
and then prove Theorem~\ref{newoptmax-thm}. 

\begin{theorem} \label{newoptmax-thm}
Mechanism~\ref{newoptmax-proxy} is a polytime, universally budget-feasible mechanism that
with probability at least $\frac{1}{4}$, returns a set $\hS$ such that
$v(\hS)\geq\frac{\opt^{(2)}}{2\gm}$, where $\opt^{(2)}$ is the second-largest $\opt_i$
value. 
\end{theorem}

\begin{theorem} \label{additivethm-overbid}
Taking $\ld=0.125$, and for a suitable choice of $p$, 
Algorithm~\ref{additivealg-overbid}, together with suitable  
payments, yields a polytime, universally budget-feasible mechanism for additive valuations
and additive costs that achieves $O(1)$ approximation with respect to $\optbench$.
\end{theorem}

\begin{proof}
We proceed similarly to the proof of Theorem~\ref{xosgenthm-overbid}.
Let $p$ be such that
$\frac{p}{8}\cdot\bigl(\frac{1}{64\beta}+1\bigr)=\frac{1-p}{8\gm}$. Note that such a
$p\in[0,1]$ always exists.
Assume that $n\geq 2$.

Let $i^*\in\pl$ be such that $\opt_{i^*}=\max_{i\in\pl}\opt_i$, and let
$\pset^*=\pset_{i^*}$. Let $\optset=\argmax_{S\sse\gset}\,\{\vbench(S): c(S)\leq B\}$.
Define $\boptset=\optset-\pset^*$, and
$\boptset_j=\boptset\cap\gset_j$ for $j=1,2$. 
We have 
$\vbench(\boptset)\geq\optbench-\opt^{(2)}$ by construction. 
Let $V'_j=\max\,\{v(S): S\sse\gset_j-\pset^*,\ c(S)\leq B\}$ for $j=1,2$. 
Let $T'_2\sse\gset_2-\pset^*$ be such that $v(T'_2)=V'_2$ and $c(T'_2)\leq B$.

Let $\Gm$ be the event that $v(\boptset_1),v(\boptset_2)\geq\vbench(\boptset)/4$, and
$\Om$ be the event that $\pset^*\sse\gset_2$.
Applying Lemma~\ref{rpartition} (a) to $\boptset$, we obtain that $\Pr[\Gm]\geq 0.5$.
$V'_1$ and $V'_2$
are identically distributed, which also holds when we condition on $\Gm$. 
So $\Pr[\{V'_2\geq V'_1\}\cap\Gm]\geq\frac{1}{4}$. Finally, event $\Om$ is independent of
events $\Gm$ and $\{V'_2\geq V'_1\}$, since $\Om$ depends only on the random choice for player
$i^*$, and events $\Gm$ and the random variables $V'_1, V'_2$ depend only on the random
choices for the other players. 
The upshot is that $\Pr[\{V'_2\geq V'_1\}\cap\Gm\cap\Om]\geq\frac{1}{8}$.

We condition on this good event $\{V'_2\geq V'_1\}\cap\Gm\cap\Om$. 
Since $\pset^*\cap\gset_1=\es$, we have $V'_1=V^*_1$ and $V_1\geq V^*_1/\beta$. 
Since $v(T'_2)\geq V'_2\geq 8\ld V'_1$, we have $w(T'_2)\geq 8n$.
We now apply Lemma~\ref{bredn-supaddcost} with $\targ=4n$ to the additive valuation
given by the $\{w_e\}_{e\in\gset_2}$ weights, to obtain $T\sse T'_2$ such that $c(T)\leq B/2$
and $4n-\max_{e\in T'_2}w_e<w(T)\leq 4n$. 
We have $\max_{e\in T'_2}w_e\leq\frac{n}{\ld V_1}\cdot\max_{e\in T'_2}v(e)+1
\leq\frac{n}{\ld V_1}\cdot\opt^{(2)}+1$, 
where the last inequality is because we have $T'_2\sse\gset_2-\pset^*$.
Recall that the set $S^*$ computed in step~\ref{demandset-additive} of
Algorithm~\ref{additivealg-overbid}, is the optimal solution to the knapsack problem over
$\gset_2$ with item values $u'_e=w_e-\frac{\bnew}{B}\cdot c_e$ for all $e\in\gset_2$, item
weights $\{w_e\}$ and knapsack budget $\bnew=4n$.   
So the above implies that 
\begin{alignat*}{1}
w(S^*) & \geq u'(S^*)\geq u'(T)
\geq w(T)-\frac{\bnew}{2}\geq 2n-\frac{n}{\ld V_1}\cdot\opt^{(2)}-1. \\
\text{Therefore} \quad 
v(S^*) & \geq\frac{\ld V_1}{n}\cdot\bigl(w(S^*)-n\bigr)
\geq\ld V_1\Bigl(1-\tfrac{1}{n}\Bigr)-\opt^{(2)} 
\geq\frac{\ld V^*_1}{2\beta}-\opt^{(2)}, 
\end{alignat*}
where the last inequality follows since $n\geq 2$.
As in the proof of Theorem~\ref{xosgenthm-overbid}, 
we have 
$V^*_1\geq\frac{\optbench-opt^{(2)}}{4}$.
By Theorem~\ref{newoptmax-thm}, the expected value returned by
Mechanism~\ref{newoptmax-proxy} is at least $\frac{1}{4}\cdot\frac{\opt^{(2)}}{2\gm}$. 
So the expected value returned by Mechanism~\ref{additivealg-overbid} is at least 
\begin{equation*}
p\cdot\frac{1}{8}\cdot\biggl(\frac{1}{64\beta}\cdot\bigl(\optbench-\opt^{(2)}\bigr)-\opt^{(2)}\biggr)
+(1-p)\cdot\frac{\opt^{(2)}}{8\gm} 
=\frac{p}{512\beta}\cdot\optbench.
\qedhere
\end{equation*}
\end{proof}

\begin{proofof}{Theorem~\ref{newoptmax-thm}}
The mechanism runs in polytime because the knapsack-cover problem in
step~\ref{newproxy-knapcover} takes polytime as the target value $n$ is polynomially
bounded, and the $\wt_e$'s are integers. 
Payment of at most $B$, and individual rationality are baked into the mechanism.
We focus on proving truthfulness and the performance guarantee.

Clearly, players in $\pl_1$ always get $0$ utility, so again nothing by lying.
For each $i\in\pl_2$, the collection $\Sc_i=\{S\sse\pset_i: \wt(S)\geq n\}$ of feasible 
sets for player $i$ does not depend on the cost-vector $c$. 
If the mechanism outputs $\es$, then every set $S\in\bigcup_{i\in\pl_2}\Sc_i$ has
$c(S)>B$, so no player can achieve positive utility by lying.
Suppose the mechanism outputs $\hS\sse\pset_{\hi}$. Player $\hi$
cannot benefit by lying, since $\hS$ already has minimum cost in $\Sc_{\hi}$ among her
feasible sets. Consider a player $i\in\pl_2$, $i\neq\hi$. 
If $i<\hi$, then by the choice of $\hi$, we have $c_i(T^*_i)>B$, so player $i$ cannot
achieve positive utility by lying. If $i>\hi$, then $i$ will never be chosen in
step~\ref{newproxy-output}, so again $i$ cannot benefit by lying.

Let $\ti$ and $\tj$ be players in $\pl$ having the largest and second-largest $\opt_i$
values respectively among all players in $\pl$. With probability $\frac{1}{4}$, we have
$\tj\in\pl_1$, $\ti\in\pl_2$. Assume that this event happens. Then, we have
$\targ\geq\frac{\opt_{\tj}}{\gm}=\frac{\opt^{(2)}}{\gm}$, and there is some set 
$T_{\ti}\sse\pset_{\ti}$ such
that $v(T_{\ti})=\opt_{\ti}\geq\targ$, $c_{\ti}(T_{\ti})\leq B$. This implies that
$\wt(T_{\ti})\geq n$, so we must have $c(S)\leq B$ for some
$S\in\bigcup_{i\in\pl_2}\Sc_i$. So the mechanism will output a set $\hS$ with
$\wt(\hS)\geq n$, which implies that 
$v(\hS)\geq\frac{\targ}{2}\geq\frac{\opt^{(2)}}{2\gm}$.
\end{proofof}

\subsubsection{Superadditive costs: polytime mechanism using constrained demand oracle}
\label{poly-obidsupadd}
For general XOS valuations and superadditive cost functions, we devise a polytime
$O(1)$-approximation universally budget-feasible mechanism given access to a constrained
demand oracle. 
This generalizes the result in Section~\ref{xos-supaddcost}, where we
assumed no-overbidding. The only change to
Algorithm~\ref{xosalggen-overbid} for general costs is that we run a new
mechanism (Mechanism~\ref{neweroptmax-proxy}) in place of Mechanism~\ref{optmax-proxy} in 
step~\ref{setoutput-overbid} of the algorithm.
We first describe this new mechanism, 
and then show that the corresponding
modified-version of Algorithm~\ref{xosalggen-overbid} yields an $O(1)$-approximation
universally budget-feasible mechanism (Theorem~\ref{xossupaddpoly-overbid}).

\paragraph{Mechanism to replace Mechanism~\ref{optmax-proxy}.}
We extend the ideas underlying
Mechanism~\ref{newoptmax-proxy}, which was used for additive valuations and costs. 
The key  
is to view Mechanism~\ref{newoptmax-proxy} as a means of combining single-player
budget-feasible mechanisms---
the mechanism for player $i\in\pl_2$ outputs 
$T^*_i$ computed in step~\ref{newproxy-knapcover} 
and payment $B$, if $c(T^*_i)\leq B$, and $\emptyset$ and zero payment, otherwise---
while preserving truthfulness (and budget-feasibility). 

We make this framework for combining single-player 
budget-feasible mechanisms explicit, and 
devise suitable single-player mechanisms that utilize a constrained demand oracle. 
To elaborate, we devise two budget-feasible mechanisms $\mechone_i, \mechtwo_i$ for each
player $i$. 
Both mechanisms take a target value $\targ$ as a parameter, and their outputs have the
property that if $\opt_i\geq\targ$, then at least one of the mechanisms returns expected
value $\Omega(\targ)$ (Theorem~\ref{neweroptmax-thm}). Both mechanisms also output a
``success bit,'' which if set to $0$, indicates that player $i$ receives zero utility
under truthful reporting. 
 
We combine this collection of mechanisms to obtain expected value 
$\Omega\bigl(\opt^{(2)}\bigr)$ while maintaining budget-feasibility, 
by using random partitioning. We find, from one part $\pl_1$, the right target $\targ$ to
aim for. Then, we pick $j=1$ or $2$ with probability $0.5$. 
We select the first player $i\in\pl_2$ whose $\mech^{(j)}_i$ mechanism's success bit 
is set to $1$, and return the output of $\mech^{(j)}_i$; if no bit is set to $1$, we output 
$\es$ and $0$ payment. 
We now delve into the details.

\SetAlgorithmName{Mechanisms}{Mechanisms}{Mechanisms}
\begin{algorithm}[ht!]
\renewcommand{\thealgocf}{}
\caption{\boldmath\textnormal{$\mechone_i$ and $\mechtwo_i$ for player $i\in\pl$}}
\let\thealgocf\oldthealgocf

\KwIn{Valuation $v:2^\gset\mapsto\R_+$, budget $B$, 
item-set $\pset_i$, superadditive cost function
$c_i:2^{\pset_i}\mapsto\R_+$, $c_i\in\C_i$, target $\targ\in\R_+$}
\KwOut{subset of $\pset_i$, payment to player $i$, \quad success bit $\flag_i$}

\BlankLine

\nonl\Indm {\bf Mechanism \boldmath $\mechone_i$}

\nl\Indp
Compute 
$T^*_i\gets\argmax_{S\sse\pset_i}\,\bigl\{v(S)-\frac{\targ}{2B}\cdot c_i(S):\ v(S)\leq\frac{\targ}{2}\bigr\}$ 
using a constrained demand oracle.

If $v(T^*_i)-\frac{\targ}{2B}\cdot c_i(T^*_i)\geq\frac{\targ}{8}$, then {\bf return}
$T^*_i$, payment $\frac{2B}{\targ}\cdot v(T^*_i)-\frac{B}{4}$, $\flag_i=1$; otherwise
\Return $\es$, zero payment, $\flag_i=0$

\BlankLine

\nonl\Indm {\bf Mechanism \boldmath $\mechtwo_i$}
\setcounter{AlgoLine}{0}

\nl\Indp 
Let $e^*_i=\argmin_{e\in\pset_i}\,\bigl\{c_i(e): v(e)\geq\frac{\targ}{8}\bigr\}$

If $c_i(e)\leq B$, {\bf return} $T^*_i$, payment $B$, $\flag_i=1$; otherwise \Return
$\es$, zero payment, $\flag_i=0$
\end{algorithm}
\SetAlgorithmName{Algorithm}{Algorithm}{Algorithm}

\SetAlgoProcName{Mechanism}{Mechanism}
\begin{procedure}[ht!] 
\caption{2ndOpt-CDemd() 
\textnormal{\qquad // Combining the $\bigl\{\mechone_i,\,\mechtwo_i\bigr\}_{i\in\pl}$ mechanisms} 
\label{neweroptmax-proxy}} 
\KwIn{Budget-feasible MD instance
$\bigl(v:2^\gset\mapsto\R_+,B,\{\pset_i,\C_i\sse\R_+^{2^{\pset_i}}\},\{c_i\in\C_i\}\bigr)$
with superadditive costs} 
\KwOut{subset of $\pset_i$ for some player $i$, suitable payments}
\SetKwComment{simpc}{// }{}
\SetCommentSty{textnormal}
\DontPrintSemicolon

For every player $i$, define $\opt_i:=\max_{S\sse\pset_i}\,\bigl\{v(S): c_i(S)\leq B\bigr\}$.
Partition $\pl$ into two sets $\pl_1,\pl_2$ by placing each player independently with
probability $\frac{1}{2}$ in $\pl_1$ or $\pl_2$.

For every $i\in\pl_1$, compute $\opt'_i$, a $\gm$-approximation to
$\opt_i$. We can achieve $\gm=(2+\e)$, for any $\e>0$, in polytime using a constrained
demand oracle; see Section~\ref{algresults}.  
Set $\targ:=\max_{i\in\pl_1}\opt'_i$.

Set $j=1$ or $j=2$, each with probability $\frac{1}{2}$. 
If for every $i\in\pl_2$, mechanism $\mech^{(j)}_i$ sets $\flag_i=0$, {\bf return} $\es$
and make no payment to any player.
Otherwise, let $\hi\in\pl_2$ be the smallest index $i\in\pl_2$ for which $\mech^{(j)}_i$
sets $\flag_i=1$; \Return the output of $\mech^{(j)}_i$ (i.e., subset of $\pset_i$, and
payment).
\label{neweroutput}
\end{procedure}
\SetAlgoProcName{Algorithm}{Algorithm}

\begin{theorem} \label{neweroptmax-thm}
Mechanism~\ref{neweroptmax-proxy} is universally budget-feasible, runs in polytime given a
constrained demand oracle, and with superadditive cost functions, obtains expected value
at least $\frac{\opt^{(2)}}{32\gm}$, where $\opt^{(2)}$ is the second-largest $\opt_i$
value. 
\end{theorem}

Theorem~\ref{neweroptmax-thm} will follow fairly directly from the following result about
the $\mechone_i,\mechtwo_i$ mechanisms.  

\begin{theorem}\label{oneplmechs}
For any player $i$, and any parameter $\targ$, mechanisms $\mechone_i$ and $\mechtwo_i$
are budget-feasible, run in polytime given a constrained demand oracle, and satisfy the
following properties. 
\begin{enumerate}[label=(\alph*), topsep=0.1ex, noitemsep, leftmargin=*]
\item If $\opt_i\geq\targ$ and $c_i$ is superadditive, then at least one of $\mechone_i$
  or $\mechtwo_i$ obtains value at least $\frac{\targ}{8}$. 
\item If $\flag_i$ is set to $0$ by $\mech^{(j)}_i$ for any $j=1,2$, then player $i$
cannot obtain positive utility by lying.
\end{enumerate}
\end{theorem}

\begin{proof}
It is clear that $\mechone_i,\mechtwo_i$ run in polytime given a constrained demand oracle.
We first argue that both $\mechone_i$ and $\mechtwo_i$ are budget feasible. Consider
$\mechone_i$. The payment made is always at most $B$ and at least the cost incurred by $i$
under truthful reporting, since if the mechanism outputs $T^*_i$, we have
$v(T^*_i)\leq\frac{\targ}{2}$ and $v(T^*_i)-\frac{\targ}{2B}\cdot c_i(T^*_i)\geq\frac{\targ}{8}$.
So the budget constraint is met, and individual rationality (IR) holds. 
To see truthfulness, suppose that $i$ reports $c'_i\neq c_i$, and some other set
$S\sse\pset_i$ is output by the constrained demand oracle. However, we have
$v(T^*_i)-\frac{\targ}{2B}\cdot c_i(T^*_i)\geq v(S)-\frac{\targ}{2B}\cdot c_i(S)$, so the
utility of $i$ cannot increase under this mis-report. 
Moreover if $v(T^*_i)-\frac{\targ}{2B}\cdot c_i(T^*_i)<\frac{\targ}{8}$, then $i$ 
obtains negative utility if mis-report causes $\flag_i$ to be set to $1$.

For $\mechtwo_i$, things are even more straightforward. It is immediate that the payment
is at most $B$ and IR holds. Truthfulness is also immediate since we are outputting a
min-cost element from a feasible set that is not affected by $i$'s reported cost.

Part (a) follows, because by a now-routine analysis, we can show that if
$\opt_i\geq\targ$, then there is some $S\sse\pset_i$ with $c(S)\leq B/2$ and
$v(S)\in\bigl[\frac{\targ}{2}-\max_{e\in\pset_i}\{v(e): c_i(e)\leq B\},\,\frac{\targ}{2}\bigr]$.
This follows by applying Lemma~\ref{bredn-supaddcost} to the set
$S^*$ with $c_i(S^*)\leq B$ that achieves value $\opt_i$. 
So if $\max_{e\in\pset_i}\{v(e): c_i(e)\leq B\}\leq\frac{\targ}{8}$, then
\[
v(T^*_1)-\frac{\targ}{2B}\cdot c_i(T^*_i)\geq
v(S)-\frac{\targ}{2B}\cdot c_i(S)\geq
\frac{\targ}{2}-\frac{\targ}{8}-\frac{\targ}{4}=\frac{\targ}{8};
\]
otherwise $\mechtwo_i$ obtains value at least $\frac{\targ}{8}$.

Given the truthfulness of $\mechone_i,\mechtwo_i$, part (b) follows from construction.
\end{proof}

\begin{proofof}{Theorem~\ref{neweroptmax-thm}}
The polynomial running time follows from Theorem~\ref{oneplmechs}, and since the
$\opt'_i$ estimates can be efficiently computed using a constrained demand oracle.

As in the proof of Theorem~\ref{newoptmax-thm}, we may assume that the largest $\opt_i$
corresponds to player in $\pl_2$, and the second-largest $\opt_i$ value corresponds to a
player in $\pl_1$. This event happens with probability $\frac{1}{4}$, and assuming this,
we have $\targ\geq\frac{\opt^{(2)}}{\gm}$ in Mechanism~\ref{neweroptmax-proxy}.
So using Theorem~\ref{oneplmechs} (a), the expected value returned is at least
$\frac{1}{4}\cdot\frac{1}{8}\cdot\frac{\opt^{(2)}}{\gm}$. 

Fix $j\in\{1,2\}$ as chosen by the mechanism.
If Mechanism~\ref{neweroptmax-proxy} outputs $\es$, then $\flag_i=0$ for all $i\in\pl_2$,
which implies, by Theorem~\ref{oneplmechs}, that no player can obtain positive utility by
lying. So suppose otherwise. Consider any $i\in\pl_2$.
Player $\hi$ cannot benefit by lying, since
$\mech^{(j)}_{\hi}$ is truthful. If $i>\hi$, then $i$ will never be chosen in
step~\ref{neweroutput}, so $i$ cannot benefit by lying. If $i<\hi$, then since
$\flag_i=0$, again by Theorem~\ref{oneplmechs}, player $i$ cannot achieve positive utility
by lying and causing $\flag_i$ to be set to $1$. 
The budget constraint and IR hold with probability $1$, because this holds for all
the $\mechone_i, \mechtwo_i$ mechanisms. It follows that Mechanism~\ref{neweroptmax-proxy}
is universally budget-feasible.
\end{proofof}

\paragraph{Algorithm~\ref{xosalggen-overbid} modified for superadditive costs: proof of
Theorem~\ref{xossupaddpoly-overbid}.} 
As mentioned earlier, the modified algorithm is simply Algorithm~\ref{xosalggen-overbid},
replacing Mechanism~\ref{optmax-proxy} with Mechanism~\ref{neweroptmax-proxy}
in step~\ref{setoutput-overbid} of the algorithm.
Therefore taking $\ld=0.5$, the exact same analysis as in the proof of
Theorem~\ref{xosgenthm-overbid} leading up to \eqref{finalineq-overbid} holds, and we
obtain that the expected value returned is at least 
\[
\frac{p}{128}\cdot\optbench-\frac{17p}{128}\cdot\opt^{(2)}+(1-p)\cdot\frac{\opt^{(2)}}{32\gm}
\]
where we have utilized the guarantee in Theorem~\ref{neweroptmax-thm} about
Mechanism~\ref{neweroptmax-proxy}. 
So taking $p=\frac{1}{1+17\gm/4}$, we
obtain expected value $\frac{1}{128+544\gm}\cdot\optbench$. 
All steps run in polytime given a constrained demand oracle. So we obtain the following
result. 

\begin{theorem} \label{xossupaddpoly-overbid} \label{xossupaddpolythm-overbid}
Taking $\ld=0.5$ and suitable $p$ in the above modified version of
Algorithm~\ref{xosalggen-overbid}, together with suitable payments, we obtain a 
universally budget-feasible mechanism for superadditive costs 
that runs in polytime given a constrained demand oracle and
achieves $O(1)$-approximation with respect to $\optbench$.
\end{theorem}

\section{Submodular valuations} \label{submod-add} \label{submod}
We now devise universally budget-feasible mechanisms for 
monotone submodular functions, which form a subclass of XOS valuations. 
Our mechanism works for arbitrary costs, without
assuming no-overbidding, and runs in polytime given a demand
oracle; but it yields a weaker approximation guarantee. 

Let $v$ be a monotone, submodular valuation. After obtaining an estimate of $\optbench$
via random partitioning, we depart from the template used for the algorithms in
Section~\ref{xos}. We now use a greedy algorithm that iteratively considers the players in
$\pl_2$ in some fixed order, and picks a suitable set from $\gset_2\cap\pset_i$ to add to
the current set $T$ using a demand oracle, until $v(T)$ becomes large enough. 
The analysis proceeds along the lines of the analysis of
the standard greedy algorithm for monotone submodular-function maximization, to show that  
we obtain expected value $\optbench(v,B,c)/\submodapx-\max_iv(\pset_i)/4$. Thus, under inputs
$(v,B,c)$ satisfying the large-market assumption $v(\pset_i)\leq\ve\cdot\optalg(v,B,c)$,
we obtain an $O(1)$-approximation with respect to $\optalg(v,B,c)$.

Our algorithm below uses demand queries of the form 
$\argmax_{S\sse\pset_i}\bigl(v(S|T)-\kp\cdot c_i(S)\bigr)$, where $T\sse\gset$ 
and $v(S|T):=v(S\cup T)-v(T)$ is the incremental value of adding $S$ to $T$. 
Such a demand-oracle query can be translated to a demand-oracle query for $v$ over the set
$T\cup\pset_i$ by setting by taking $c_\ell$ to be the identically-$0$ function over
$2^{\pset_\ell}$ for all $\ell\neq i$, and $c_i(A)=c_i(A-T)$ for all $A\sse\pset_i$. 
This amounts to ``setting the prices for elements in $T$ to be $0$''. 
(We assume that $\C_i$ is closed under this operation.) 

\SetAlgoProcName{Mechanism}{Mechanism}
\begin{procedure}[ht!] 
\caption{Submod-UniBF() \textnormal{\qquad // universally-budget-feasible mechanism for
    monotone submodular valuations, general cost functions} \label{submodalg-gencost}}
\KwIn{Budget-feasible MD instance
$\bigl(v:2^\gset\mapsto\R_+,B,\{\pset_i,\C_i\sse\R_+^{2^{\pset_i}}\},\{c_i\in\C_i\}\bigr)$; 
parameter $\ld\in[0,1]$}
\KwOut{subset of $\gset$ and payments}
\SetKwComment{simpc}{// }{}
\SetCommentSty{textnormal}
\DontPrintSemicolon

Partition $\pl$ into two sets $\pl_1$, $\pl_2$ by placing each player independently with
probability $\frac{1}{2}$ in $\pl_1$ or $\pl_2$.
Let $\gset_j:=\bigcup_{i\in\pl_j}\pset_i$ be the induced partition of $\gset$. 
Compute $\lpopt_1$ using a demand oracle. 
Set $V_1=\lpopt_1$. \;
\label{submod-optestim}

Initialize $T\gets\es$. 
Considering players $i\in\pl_2$ in some fixed order: 
if $v(T\cup\pset_i)\leq\ld V_1$, compute
$\demd_i\gets\argmax_{S\sse\pset_i}\bigl\{v(S|T)-\frac{\ld V_1}{B}\cdot c_i(S)\bigr\}$ and
set $T_i=T$ and $T\gets T\cup\demd_i$; otherwise, exit the loop. \;
\label{greedy}

\Return $T$. Pay $\frac{B}{\ld V_1}\cdot v(\demd_i|T_i)$ to each player $i\in\pl_2$ for which
$T\cap\pset_i\neq\es$, and $0$ to the other players.
\label{submod-output}
\end{procedure}
\SetAlgoProcName{Algorithm}{Algorithm}

\begin{theorem} \label{submodthm}
Taking $\ld=0.5$ in Mechanism~\ref{submodalg-gencost}, 
we obtain a universally budget-feasible mechanism that obtains expected value
$\optbench(v,B,c)/\submodapx-\max_{i\in[k]}v(\pset_i)/4$.
\end{theorem}

\begin{proof}
Lemma~\ref{submod-budgetfeas} proves that the mechanism is universally budget feasible. We
prove the approximation guarantee here. We drop $(v,B,c)$ from the argument of
$\optbench$. 
Let $\kp=\frac{\ld V_1}{B}$.
Let $\optset_2:=\argmax_{S\sse\gset_2}\,\bigl\{v(S)-\kp\cdot c(S)\bigr\}$.
We first argue that $v(\optset_2)-\kp\cdot c(\optset_2)\geq\lpopt_2-\kp B$. Let $x^*$ be an
optimal solution to $\optlp[(\gset_2)]$. Then, 
$v(\optset_2)-\kp\cdot c(\optset_2)\geq\sum_S x^*_S\bigl(v(S)-\kp\cdot c(S)\bigr)
\geq\lpopt_2-\kp B$. 

If we exit the loop in step~\ref{greedy} prematurely without going over all the players in
$\pl_2$, then we clearly have $v(T)\geq\ld V_1-\max_{i\in\pl_2}v(\pset_i)$. Suppose otherwise.
Then we consider all players in $\pl_2$, and if player $i'$ is considered right after
player $i$, we have $T_{i'}=T_i\cup\demd_i$.
For any $i\in\pl_2$, we have 
\begin{equation}
v(\demd_i|T_i)-\kp\cdot c_i(\demd_i)
\geq v(\optset_2\cap\pset_i|T_i)-\kp\cdot c(\optset_2\cap\pset_i)
\geq v(\optset_2\cap\pset_i|T)-\kp\cdot c(\optset_2\cap\pset_i) 
\end{equation}
where the last inequality follows from submodularity of $v$.
Adding this for all $i\in\pl_2$, since $v(\cdot|T)$ is also monotone and submodular, and
hence subadditive, and $\{\optset_2\cap\pset_i\}_{i\in\pl_2}$ partitions $\optset_2$,
we obtain that 
$v(T)-\kp\cdot c(T)\geq v(\optset_2|T)-\kp\cdot c(\optset_2)
= v(\optset_2\cup T)-v(T)-\kp\cdot c(\optset_2)$.
It follows that $2v(T)\geq v(\optset_2)-\kp\cdot c(\optset_2)\geq\lpopt_2-\ld V_1$. 
Therefore, $v(T)\geq\frac{\lpopt_2-\ld V_1}{2}$. 

By Lemma~\ref{rsample-lp}, we have that $\lpopt_2\geq V_1=\lpopt_1\geq\frac{\optbench}{4}$ 
holds with probability at least $\frac{1}{4}$. Assuming that this event happens, plugging
in $\ld=\frac{1}{3}$, the above analysis shows that 
that $v(T)\geq\frac{V_1}{3}-\max_{i\in\pl_2}v(\pset_i)
\geq\frac{\optbench}{12}-\max_{i\in\pl_2}v(\pset_2)$. So the expected value returned is
at least $\frac{\optbench}{48}-\frac{\max_{i\in\pl_2}v(\pset_2)}{4}$.
\end{proof}

\begin{lemma} \label{submod-budgetfeas}
Mechanism~\ref{submodalg-gencost} is universally budget feasible.
\end{lemma}

\begin{proof}
We consider each possible random outcome of the mechanism, and show that the payments
given in step~\ref{submod-output} yield a budget-feasible mechanism.
Let $(\pl_1,\gset_1)$, $(\pl_2,\gset_2)$ be the partition obtained in
step~\ref{partition}. Let $\kp=\frac{\ld V_1}{B}$. Fix an input $(v,B,c)$.

Consider a player $i$. If $i\in\pl_1$, then her payment, cost incurred, and utility are
always $0$. So suppose $i\in\pl_2$. 
We say that player $i$ is {\em processed} in step~\ref{greedy}, if we
consider $i$ in the loop in that step and compute a demand-set for $i$.
The key observation is that whether player $i$ is processed, and the set $T=T_i$ at the
point when we consider $i$ {\em do not depend on $i$'s reported cost function}.
If $i$ is not processed, then $i$'s payment, cost incurred, and utility are
always $0$. Otherwise, 
$i$'s utility is $\frac{1}{\kp}\cdot v(S|T_i)-c_i(S)$, where $S$ is the demand-set
computed for $i$ when we process her. 
So by construction, this utility
is maximized by reporting her true cost function, which yields the set $S=\demd_i$, and we
obtain truthfulness. 

The total payment made by the mechanism is 
$\frac{1}{\kp}\cdot\sum_{i\text{ is processed}}v(\demd_i|T_i)=v(T)/\kp\leq B$, where the
last inequality follows because $v(T)\leq\ld V_1$ by construction.
\end{proof}

\section{Subadditive valuations} \label{subadditive}
Our mechanisms for XOS valuations can be utilized to obtain budget-feasible mechanisms for
monotone subadditive valuations that achieve $O(\log k)$-approximation with respect to 
$\optbench$, without assuming no-overbidding. 
Recall that $k$ is the number of players.
In particular, we obtain 
an $O(\log k)$-approximation universally budget-feasible mechanism 
for general costs without assuming no-overbidding.
  
The idea here is to use an ``XOS-like'' pointwise-approximation of the subadditive
valuation $v$. However, note that if $\tv$ is a pointwise approximation of $v$ 
satisfying $v(S)/\gm\leq\tv(S)\leq v(S)$ for all $S\sse\gset$, this
{\em does not} imply that $\optbench$ 
inherits this approximation
property: i.e., we do not necessarily have that
$\frac{\optbench(v,B,c)}{\optbench(\tv,B,c)}\in\bigl[1,O(\gm)\bigr]$. 
So we need to proceed more carefully, and utilize $\tv$ internally in our mechanisms, in
the appropriate demand-set computations of our mechanisms for XOS valuations.

It is well-known that $v$ can be pointwise-approximated within an $O(\log n)$-factor (in
the sense of $\tv$ above) by an XOS function, and that this is
tight~\cite{BhawalkarR11,BeiCGL12}. To do better, we observe that the analysis of our
mechanisms for XOS valuations relies on a weaker property than the function being XOS. An
XOS function $g$ satisfies the fractional-cover property (see Section~\ref{prelim}) that
for every set $S$, the value $\sum_{T\sse S}g(T)x_T$ of every fractional cover
$\{x_T\}_{T\sse S}$ of $S$ is at least $g(S)$. Our mechanisms for XOS valuations only
require this property to hold for fractional covers supported on 
{\em player-respecting subsets of $S$}, i.e., on sets 
$T\sse S$, where $T\cap\pset_i\in\{\es,S\cap\pset_i\}$ for all $i\in[k]$ (see
Definition~\ref{prcover}). This is because 
we only use the XOS property via Claim~\ref{xosnprop}, to argue that
$\sum_i\bigl(v(S)-v(S-\pset_i)\bigr)\leq v(S)$ (see, e.g., \eqref{bfexp-paymt}), which follows
from the fractional-cover property applied to a cover supported on
$\{S-\pset_i\}_{i\in[k]}$, which are player-respecting subsets of $S$.

\begin{definition} \label{prcover}
Given a set $S\sse\gset$, we say that $T\sse S$ is a {\em player-respecting subset} of $S$
if $T\cap\pset_i\in\{\es,S\cap\pset_i\}$ for all $i\in[k]$. Let $\Pc(S)\sse 2^S$ denote
the collection of player-respecting subsets of $S$.

We say that a function $g:2^\gset\mapsto\R_+$ is a {\em player-respecting XOS function},
if for every set $S\sse\gset$, the optimum value of the following LP, denoted
$\optfrcover$, is $g(S)$: 
\begin{equation*}
\min \quad \sum_{T\in\Pc(S)}g(T)x_T \qquad \text{s.t.} \qquad
\sum_{T\in\Pc(S): e\in T}x_T=1 \quad \forall e\in S, \qquad x\geq 0.
\tag*{$\frcoverlp$}
\end{equation*}
A feasible solution to the above LP is called a player-respecting fractional cover of $S$.
When $g$ is monotone, we can relax the equality constraints above to $\geq$-constraints.
\end{definition}

The approximation factor we obtain is actually $O(\gm)$, where $\gm\geq 1$ 
is the best factor possible for a pointwise approximation of $v$ by a player-respecting
XOS function. 
One can show that the best such approximation comes from the family of LPs \frcoverlp 
(just as the best pointwise XOS-approximation can be obtained by solving the
fractional-cover LP~\cite{BeiCGL12}); 
$\gm$ corresponds to the integrality gap of this family, and we always have $\gm=O(\log
k)$. 

\begin{lemma} \label{prxos}
Let $g:2^\gset\mapsto\R_+$ be a monotone subadditive function. 
Define 
$\gprxos(S):=\optfrcover$ for all $S\sse\gset$. Then
\begin{enumerate}[label=(\alph*), topsep=0ex, noitemsep, leftmargin=*]
\item $\gprxos$ is a player-respecting XOS function; hence, 
$\sum_{i\in[k]}\bigl(\gprxos(S)-\gprxos(S-\pset_i)\bigr)\leq\gprxos(S)$ for all $S\sse\gset$.
\item if $h$ is a player-respecting XOS function with $h(S)\leq g(S)$ for all $S\sse\gset$,
then we have $h(S)\leq\gprxos(S)$ for all $S\sse\gset$; 
\item we have $g(S)/\gm\leq\gprxos(S)\leq g(S)$ for all $S\sse\gset$, where $\gm$ is 
$\max_{S\sse\gset}(\text{integrality gap of $\frcoverlp$})$;
\item $\gm\leq O(\log k)$.
\end{enumerate}
\end{lemma}

We defer the proof of the above lemma to the end of the section.
In the rest of this section, we fix the monotone, subadditive valuation $v$, and suppose
that we have a player-respecting XOS function $\vxos$ such that 
$v(S)/\gm\leq\vxos(S)\leq v(S)$ for all $S\sse\gset$, where $\gm=O(\log k)$. 

\vspace*{-1ex}
\paragraph{Universally budget-feasible mechanism for general costs.}
We briefly describe the changes to the suitable mechanism from 
Section~\ref{overbid} below. 
Recall that $V^*_j:=\max\,\{v(S):\ S\sse\gset_j,\ c(S)\leq B\}$ for $j=1,2$, where
$\gset_1,\gset_2$ are identically-distributed player-respecting subsets of $\gset$ that
partition $\gset$.
We proceed as in Algorithm~\ref{xosalggen-overbid} (in Section~\ref{unibf-overbid}) for
XOS valuations without assuming no-overbidding, 
with the only change being that in step~\ref{demandset-overbid} of the algorithm, we now
compute  
$S^*\gets\argmax_{S\sse\gset_2}\,
\bigl\{\vxos(S)-\frac{\ld V_1}{\gm B}\cdot c(S):\ \vxos(S)\leq\frac{\ld V_1}{\gm}\bigr\}$.  
As before, we obtain universal budget-feasibility. For the approximation, we proceed as in
the proof for XOS valuations.

Let $\opt_i=\max_{S\sse\pset_i}\,\max\{v(S): c_i(S)\leq B\}$ for $i\in\pl$, and
$\opt^{(2)}$ be the second-largest $\opt_i$ value.
Let $i^*\in\pl$ be such that $\opt_{i^*}=\max_{i\in\pl}\opt_i$, and let
$\pset^*=\pset_{i^*}$. 
Let $\optset$ be such that $c(\optset)\leq B$ and $\vbench(\optset)=\optbench$.
Let $\boptset=\optset-\pset^*$, and $\boptset_j=\boptset\cap\gset_j$ for $j=1,2$.
Note that for any $S\sse\gset-\pset^*$ with $c(S)\leq B$, and any $i\in\pl$, we have
$\vxos(S\cap\pset_i)\leq v(S\cap\pset_i)\leq\opt^{(2)}$.
Let $V'_j=\max\,\{v(S): S\sse\gset_j-\pset^*,\ c(S)\leq B\}$ for $j=1,2$. 
let $T'_2\sse\gset_2-\pset^*$ be such that $v(T'_2)=V'_2$ and $c(T'_2)\leq B$.

We may assume that $v(\boptset_1),v(\boptset_2)\geq\vbench(\boptset)/4$, $V'_2\geq V'_1$, 
and $\pset^*\sse\gset_2$, an event that occurs with probability at least $\frac{1}{8}$. 
We take $\ld=0.5$, and $p=\frac{128\gm}{144\gm+1}$ so that 
$\frac{p}{8}\cdot\bigl(1+\frac{1}{16\gm}\bigr)=1-p$. 
We have $\vxos(T'_2)\geq\frac{V'_2}{\gm}\geq\frac{2\ld V'_1}{\gm}=\frac{2\ld V_1}{\gm}$. 
So using Lemma~\ref{bredn-supaddcost} with $\vxos$, we can find $T\sse T'_2$ with
$c(T)\leq B/2$ and  
$\min\bigl\{\vbench[{\vxos}](T'_2)-\frac{\ld V_1}{\gm},\frac{\ld V_1}{\gm}-\max_{e\in T'_2}v(e)\bigr\}<
\vxos(T)\leq\frac{\ld V_1}{\gm}$. 
This implies that $\vxos(T)>\frac{\ld V_1}{\gm}-\opt^{(2)}$, which 
can be used to show 
that $\vxos(S^*)\geq\frac{\ld V_1}{2\gm}-\opt^{(2)}$. So proceeding as in the rest of the
proof of Theorem~\ref{xosgenthm-overbid}, the expected value returned is at
least $\frac{p}{128\gm}\cdot\optbench(v,B,c)=\frac{1}{144\gm+1}\cdot\optbench(v,B,c)$.

\subsection*{Proof of Lemma~\ref{prxos}}

\noindent{\bf Part (a).}\ 
Consider any $S\sse\gset$, and any player-respecting fractional
cover $\{x_T\}_{T\in\Pc(S)}$ of $S$. For each $T\in\Pc(S)$, let $y^T$ be an optimal
solution to $\frcoverlp[T]$, so that $\gprxos(T)=\sum_{Z\in\Pc(T)}g(Z)y^T_Z$. Note that,
for any $T\in\Pc(S)$, we have $\Pc(T)=\{Z\in\Pc(S): Z\sse T\}$. 
Now consider the solution $\bx$, where we set 
$\bx_Z=\sum_{T\in\Pc(S):T\supseteq Z}x_Ty^T_Z$ for every $Z\in\Pc(S)$. We argue that this
is a feasible solution to \frcoverlp of objective value $\sum_{T\in\Pc(S)}\gprxos(T)x_T$,
which implies that $\gprxos(S)=\optfrcover\leq\sum_{T\in\Pc(S)}\gprxos(T)x_T$, proving
part (a). 

For any $e\in S$, we have 
\begin{equation*}
\begin{split}
\sum_{Z\in\Pc(S):e\in Z}\bx_Z
& =\sum_{Z\in\Pc(S):e\in Z}\sum_{T\in\Pc(S):T\supseteq Z}x_Ty^T_Z
=\sum_{T\in\Pc(S):e\in T}x_T\cdot\sum_{Z\in\Pc(S):Z\sse T,e\in Z}y^T_Z \\
& =\sum_{T\in\Pc(S):e\in T}x_T\cdot\sum_{Z\in\Pc(T):e\in Z}y^T_Z
=\sum_{T\in\Pc(S):e\in T}x_T\cdot 1 = 1.
\end{split}
\end{equation*}
The third equality is because $\Pc(T)=\Pc(S)\cap 2^T$, for $T\in\Pc(S)$, the fourth is
because $y^T$ is a feasible solution to $\frcoverlp[T]$, and the final equality is
because $x_T$ is a feasible solution to $\frcoverlp[S]$. This shows that $\bx$ is a
feasible solution to $\frcoverlp$. 
Its objective value is
\begin{equation*}
\begin{split}
\sum_{Z\in\Pc(S)}g(Z)\bx_Z
& =\sum_{Z\in\Pc(S)}\sum_{T\in\Pc(S):T\supseteq Z}g(Z)x_Ty^T_Z
=\sum_{T\in\Pc(S)}x_T\cdot\sum_{Z\in\Pc(T)}g(Z)y^T_Z \\
& =\sum_{T\in\Pc(S)}x_T\cdot\gprxos(T)
\end{split}
\end{equation*}
where the final equality is because $\gprxos(T)=\optfrcover[T]$ and $y^T$ is an optimal
solution to $\frcoverlp[T]$. 
The second statement follows from Claim~\ref{xosnprop}.

\smallskip
\noindent{\bf Part (b).}\
Consider any $S\sse\gset$. Let $x^*$ be an optimal solution to $\frcoverlp$.
Then,
\[
\gprxos(S)=\sum_{T\in\Pc(S)}g(T)x^*_T\geq\sum_{T\in\Pc(S)}h(T)x^*_T\geq h(S)
\]
where the last inequality is because $h$ is a player-respecting XOS function.

\smallskip
\noindent {\bf Part (c).}\ 
Consider any $S\sse\gset$. It is clear that $\gprxos(S)\leq g(S)$ because setting $x_S=1$
and all other $x_T=0$ is a feasible solution to $\frcoverlp$.
Moreover, the optimal value of an {\em integer} solution to $\frcoverlp$ is $g(S)$, since
any integer solution corresponds to a partition $T_1,\ldots,T_\ell$ of player-respecting
subsets of $S$ and has value $g(T_1)+g(T_2)+\ldots+g(T_\ell)\geq g(S)$, since $g$ is
subadditive. Therefore, we obtain that $\gprxos(S)\geq g(S)/\gm$. 

\smallskip
\noindent {\bf Part (d).}\
Consider any $S\sse\gset$. Relaxing the equality constraints of $\frcoverlp$ to
$\geq$-inequalities, we see that $\frcoverlp$ consists of at most $k$ distinct
covering constraints, one for each $i\in[k]$ for which $S\cap\pset_i\neq\es$. This is
simply because if $e,e'\in S$ are such that $\{e,e'\}\sse S\cap\pset_i$ for some $i$, then
$T\in\Pc(S)$ contains $e$ iff it contains $e'$, and so the constraints for $e$, $e'$ are
identical. Thus, by standard results on set cover, the integrality gap is $O(\log k)$.
\hfill \qed

\section{Guarantees relative to \boldmath $\OPTalg$} 
\label{optalg-bounds}

\subsection{Algorithmic problem of approximately computing \boldmath $\OPTalg$} 
\label{approx} \label{algresults}
We now consider the algorithmic problem of developing polytime approximation algorithms
for computing $\OPTalg(v,B,c)=\max_{S\sse\gset}\,\{v(S): c(S)\leq B\}$ given $(v,B,c)$ as
input, and a demand oracle. 
We exploit the ideas underlying our mechanisms to obtain 
a $(2+\ve)$-approximation for subadditive $v$ and superadditive $c$, i.e., $c$ satisfies
$c(S\cup T)\leq c(S)+c(T)$ for all $S,T\sse\gset$. 
Note that when the cost functions $c_1,\ldots,c_k$ comprising $c$ are superadditive, then
the function $c(S)=\sum_i c_i(S\cap\pset_i)$ is superadditive. 
Our 
guarantee is tight, since even for additive $c$, one cannot do
better~\cite{BadanidiyuruDO19}. Also, it is known that value oracles are insufficient to
obtain any approximation factor better than $\sqrt{n}$, even for XOS valuations and
additive $c$. 

\begin{theorem} \label{subaddalgthm}
We can obtain a $(2+\ve)$-approximation with subadditive $v$ and superadditive $c$, using
demand oracles.
\end{theorem}

\begin{proof}
The idea is similar to the algorithm in~\cite{BadanidiyuruDO19}, which we simplify
somewhat. 
Recall that $\lpopt$ is the optimal value of the following LP, which can be solved using a
demand oracle. Here $S$ ranges over subsets of $\gset$.
\begin{equation}
\max \ \ \sum_{S}v(S)x_S \qquad \text{s.t.} \qquad
\sum_{S}c(S)x_S\leq B, \qquad \sum_{S}x_S\leq 1, \qquad x\geq 0. 
\tag{\optlpname(\gset)} 
\end{equation}
We may assume that $c(e)\leq B$ for all $e\in\gset$. 
We first describe a simple $3$-approximation algorithm. Let $\bx$ be an optimal solution
to the above LP. If there exists $e\in\gset$ with $v(e)\geq\lpopt/3$, then we simply
return such an element.
Otherwise, find
$S^*=\argmax_{S\sse\gset}\,\bigl(v(S)-\kp\cdot c(S)\bigr)$, where
$\kp=\frac{2\lpopt}{3B}$. We have 
$v(S^*)-\kp\cdot c(S^*)\geq\sum_S\bx_S\bigl(v(S)-\kp\cdot c(S)\bigr)=\lpopt-\kp B=\lpopt/3$.
Now let $T$ be a maximal subset of $S^*$ with $c(T)\leq B$. If $T=S^*$, then we are done.
Otherwise, for any $e\in S^*-T$,
taking $T'=T\cup\{e\}$, we have 
$v(S^*)-\kp\cdot c(S^*)\leq\bigl[v(T')-\kp\cdot c(T')\bigr]+\bigl[v(S^*-T')-\kp\cdot c(S^*-T')\bigr]$
since $v$ is subadditive, and $c$ is superadditive, and so $v(T')-\kp\cdot c(T')\geq 0$. 
So $v(T')>\kp B$, and
$v(T)>\kp B-\vemax\geq\lpopt/3$. 
Note that we do not actually need $\lpopt$ above: we can consider the largest $V$ such
that $\max_{S\sse\gset}\bigl(v(S)-\frac{2V}{3B}\cdot c(S)\bigr)$ is at least $V/3$.

To refine this to a $(2+\ve)$-approximation, we use enumeration to enumerate all elements
$e$ that belong to some fixed optimal solution $\optset$ with $c(e)\geq\ve B$. Since $c$ is
superadditive, there are at most $\frac{1}{\ve}$ such elements. So letting 
$H=\{e\in\gset: c(e)\geq\ve B\}$ and $L=\{e\in\gset: c(e)<\ve B\}$, we may assume that we
know the set $A=\optset\cap H$. Note that $|A|\leq\frac{1}{\ve}$ and $c(A)\leq B$.
We now compute 
$S^*=\argmax_{S\sse A\cup L}\,\bigl(v(S)-\kp'\cdot c(S)\bigr)$ where
$\kp'=\frac{\lpopt}{2B}$. As before, we may assume that $\vemax<\lpopt/2$. Now when we
pick a maximal subset $T\sse S^*$ with $c(T)\leq B$, we first pick all the elements in
$A\cap S^*$. This ensures that $c(T)>B-\ve B$, and so $v(T)>(1-\ve)\lpopt/2$.
\end{proof}

\subsection{Tightness of the lower bounds in Theorem~\ref{intro-thm}: mechanism-design
  guarantees} \label{impos-tight}
We now devise budget-feasible mechanisms obtaining approximation factors relative
$\optalg(v,B,c)$, which will show that the lower bounds stated in
Theorems~\ref{rand-overbid} and~\ref{detlb} are tight for XOS valuations, and in some
cases, also for subadditive valuations.
Let $\optset=\argmax_{S\sse\gset}\,\{v(S): c(S)\leq B\}$ be
an optimal solution to the algorithmic problem. Let $\psmax=\max_i|\pset_i|$. 
For every player $i$, order the elements in $\pset_i$ arbitrarily; let
$e^i_1,\ldots,e^i_{n_i}$ be this ordering, where $n_i\leq\psmax$. 
For an index $j\in[\psmax]$, let $\sduniv^j:=\bigcup_{i\in[k]}\{e^i_j\}$, with the
understanding that $\{e^i_j\}$ is the empty-set if $j>n_i$.

Without no-overbidding, consider the following randomized mechanism: we pick an element
$e\in\gset$ uniformly at random, and return $e$ if $c(e)\leq B$ along with a payment of
$B$ (to the player owning $e$), and the empty-set with $0$ payments otherwise. This is
clearly universally budget-feasible, and the expected value returned is at least
$\sum_{e\in\optset}v(e)/n\geq\optalg/n$; this holds for any subadditive function. 
We can also obtain an $O(\psmax)$-approximation in this setting. We achieve this
by reducing to the single-dimensional single-item setting losing an $O(\psmax)$-factor. 
We pick a random index $j\in[\psmax]$, and run a budget-feasible mechanism for the
single-dimensional setting with universe $\sduniv^j$, valuation $v$ restricted to
$\sduniv^j$, and element-costs $\{c_i(e^i_j)\}_{i\in[k]}$. 
Letting $\optalg^j$ denote the optimum for this instance, by subadditivity of $v$, we have
that $v(\optset)\leq\sum_{j=1}^\psmax
v(\optset\cap\sduniv^j)\leq\sum_{j=1}^\psmax\optalg^j$. So if the budget-feasible
mechanism for the single-dimensional instance achieves an $O(1)$-approximation, we obtain
$O(\psmax)$-approximation relative to $\optalg$ for the original instance. 
The mechanism inherits the budget-feasibility guarantee, i.e., universally
budget-feasible or budget-feasible in expectation, from the budget-feasibility guarantee
of the single-dimensional mechanisms.

Assuming no-overbidding, it is easy to see that returning
$e^*=\argmax_{e\in\gset}v(e)$ and making a payment of $B$ (to the player owning $e^*$)
yields a deterministic budget-feasible mechanism that achieves an $n$-approximation
relative to $\optalg$ for any subadditive function, since $v(\optset)\leq nv(e^*)$ due to
subadditivity. 

The above mechanisms match the lower bounds in Theorems~\ref{rand-overbid}
and~\ref{detlb} (a). We next devise a universally budget-feasible mechanism that achieves
an $O(\log\psmax)$-approximation relative to $\optalg$, assuming no-overbidding. 

\begin{theorem} \label{xos-optalg}
Assuming no-overbidding, we can obtain polytime mechanisms $\mech_1$ and $\mech_2$  
for XOS valuations 
that achieve an $O(\log\max_i|\pset_i|)$-approximation relative to $\optalg$, such that:
\begin{enumerate}[label=(\alph*), topsep=0ex, noitemsep, leftmargin=*]
\item $\mech_1$ is budget-feasible in expectation, works with general costs, and uses a
demand oracle; 
\item $\mech_2$ is universally budget-feasible, works with superadditive costs, and uses a  
constrained demand oracle.
\end{enumerate}
\end{theorem}

\begin{proof}
Recall that $\psmax:=\max_{i\in\pl}|\pset_i|$.
This is a consequence of the fact that all our mechanisms work with a suitable estimate
$V_1$ of our benchmark. Their budget-feasibility guarantees hold regardless of the
quality of this estimate, but their approximation guarantee relies on $V_1$ being a
good-enough estimate. In our mechanisms, we obtain $V_1$ using random partitioning, and
argue (loosely speaking) that $V_1=\Omega(\optbench)$. 
But instead, if we have an estimate $\est\geq\optalg/M$, 
then 
we can take $V_1=\est\cdot 2^\ell$,
for a random exponent $\ell\in\{0\}\cup[\ceil{\log_2 M}]$. 
One of the choices for $\ell$,
will yield a good estimate of $\optalg$. This happens with probability 
$\Omega\bigl(\frac{1}{\log M}\bigr)$, 
and for this choice, we will obtain
$\Omega(\optalg)$ value; thus the expected value obtained is 
$\Omega(\optalg/\log M)$. 

One simple choice for $\est$ is $\vemax$, which is at least $\optalg/n$, and with this
estimate one can obtain an $O(\log n)$-approximation relative to $\optalg$, and the
mechanisms do not require random partitioning. For instance,
for the budget-feasible-in-expectation mechanism, we simply run steps
steps~\ref{bfexp-demandset}, \ref{bfexp-setoutput} of Algorithm~\ref{xosalg-expgencost},
with the above random value $V_1$ and taking $\ld=0.5$, where we now work with the entire
universe $\gset$ (so $\lpopt_2=\lpopt$ now). Since $\lpopt\leq n\vemax$, 
there is some $V_1$ such that $V_1\leq\lpopt\leq 2V_1$, and for this $V_1$, using
Lemma~\ref{bfexp-val}, one obtains value $\Omega(V_1)$.

To obtain $O(\log\psmax)$ approximation guarantees, we work with a better initial estimate
(than $\vemax$) that leverages random partitioning.
In essence, we take $\est$ to be the better of $\vemax$ and an estimate of the algorithmic
optimum with player-set $\pset_1$ and corresponding item-set $\gset_1$. 
Recall that $\lpopt_j$ is the optimal value of $\optlp[(\gset_j)]$ for $j=1,2$, 
$\lpopt$ is the optimal value of $\optlp[(\gset)]$, and
$V^*_j:=\max\,\{v(S): S\sse\gset_j,\ c(S)\leq B\}$ for $j=1,2$.

\medskip
For mechanism $\mech_1$, 
we take $\est=\max\bigl\{\lpopt_1,\frac{\vemax}{2}\bigr\}$, and set $V_1$ as indicated
above, to be $\est\cdot 2^{\ell}$, 
where $\ell$ is a random exponent chosen uniformly from $\{0\}\cup[\ceil{\log_2 6\psmax}]$.
Mechanism $\mech_1$ simply runs steps~\ref{bfexp-demandset}, \ref{bfexp-setoutput} of
Algorithm~\ref{xosalg-expgencost}, with the above random value $V_1$, taking $\ld=0.5$.

By Lemma~\ref{rsample-lp}, with probability at least $0.25$, 
we have $\lpopt_2\geq\lpopt_1\geq\frac{\optbench}{4}$ and $\lpopt_2\geq\frac{\lpopt}{2}$. 
Assuming that this event happens, we have 
$\lpopt_2\geq\max\bigl\{\frac{\optalg}{2},\est\bigr\}$, since $\lpopt\geq\optalg\geq\vemax$, 
and
\begin{equation*}
\begin{split}
\est & \geq\frac{0.5}{0.5+\psmax/4}\cdot\lpopt_1+\frac{\psmax/4}{0.5+\psmax/4}\cdot\frac{\vemax}{2}
\\ &
\geq\frac{0.5}{0.5+\psmax/4}\cdot\frac{\optalg-\psmax\cdot\vemax}{4}+\frac{\psmax/4}{0.5+\psmax/4}\cdot\frac{\vemax}{2}
\geq\optalg\cdot\frac{0.5}{2+\psmax}\geq\frac{\optalg}{6\psmax}
\end{split}
\end{equation*}
where the second inequality follows since $\vbench(S)\geq v(S)-\psmax\cdot\vemax$ for any
set $S\sse\gset$, and so $\optbench\geq\optalg-\psmax\cdot\vemax$. 
Lemma~\ref{bfexp-val} shows that for any given $V_1\leq\lpopt_2$, the expected value
obtained is at least $\ld V_1\bigl(1-\frac{\ld V_1}{\lpopt_2}\bigr)$. 
Since $\est\geq\frac{\optalg}{6\psmax}$, for some $V_1$, we have 
$V_1\leq\lpopt_2$ and $2V_1\geq\min\{\lpopt_2,\optalg\}\geq\frac{\optalg}{2}$. For this
$V_1$, we obtain value $\Omega(V_1)=\Omega(\optalg)$.

\medskip
For mechanism $\mech_2$, let $V'_1$ be a $\gm$-approximation to $V^*_1$; by
Theorem~\ref{subaddalgthm}, we can take $\gm=(2+\ve)$, for any $\ve>0$.
We take $\est=\max\bigl\{V'_1,\frac{\vemax}{2}\bigr\}$, and again let $V_1$ be random value
of the form $\est\cdot 2^{\ell}$, where $\ell$ is chosen uniformly from
$\{0\}\cup[\ceil{\log_2((4\gm+2)\psmax)}]$. 
Mechanism $\mech_2$ runs steps~\ref{demandset}, \ref{setoutput} of
Algorithm~\ref{xosalg-gencost} taking $\ld=0.5$ and $p=0.8$. 
By Corollary~\ref{rsample}, with probability at least $0.25$, 
we have $V^*_2\geq V^*_1\geq\frac{\optbench}{4}$ and
$V^*_2\geq\optalg/2$. Assuming that this event happens, similar to before, we can argue
that $\est\leq V^*_2$ and 
\[
\est\geq\frac{0.5}{0.5+\psmax/4\gm}\cdot\frac{\optalg-\psmax\cdot\vemax}{4\gm}
+\frac{\psmax/4\gm}{0.5+\psmax/4\gm}\cdot\frac{\vemax}{2}
\geq\optalg\cdot\frac{0.5}{2\gm+\psmax}\geq\frac{\optalg}{(4\gm+2)\psmax}
\]
where the first inequality follows since $V'_1\geq\frac{V^*_1}{\gm}\geq\frac{\optbench}{4\gm}$.
For any given $V_1$, applying Lemma~\ref{bredn-supaddcost} 
to $T^*_2\sse\gset_2$ such that $v(T^*_2)=V^*_2$, $c(T^*_2)\leq B$,
we can extract $T\sse T^*_2$ such that $c(T)\leq B/2$ and 
$\min\{V^*_2-\ld V_1,\ld V_1\}-\vemax<v(T)\leq\ld V_1$, so that 
$v(S^*) \geq\min\{V^*_2-1.5\ld V_1,0.5\ld V_1\}-\vemax$.
Amoung the random values we consider for $V_1$, 
there is some $V_1$ with $2\ld V_1=V_1\leq V^*_2$ and
$2V_1\geq V^*_2\geq\frac{\optalg}{2}$, and for this $V_1$, we obtain
expected value at least $\frac{pV_1}{16}=\Omega(\optalg)$.
So the overall expected value obtained is $\Omega(\optalg/\log\psmax)$.
\end{proof}

\begin{remark} \label{allin}
Amanatidis et al.~\cite{AmanatidisKMST23} make an ``all-in'' assumption in the
single-dimensional \los setting, 
which assumes that the budget $B$ is large enough that the buyer can buy {\em all} levels
of service from any player~\cite{AmanatidisKMST23}, which allows them to bypass the
impossibility results in Theorem~\ref{detlb} and obtain an $O(1)$-approximation with
respect to $\optalg$.
\footnote{The all-in assumption is closely related to the assumption that in the
divisible-item setting 
(where one can buy a fraction of an item), the buyer can afford to buy an entire unit
of item from any player~\cite{KlumperS22}; this is a weaker version of the large-market
assumption for divisible items~\cite{AnariGN18} that allows~\cite{KlumperS22} to again
obtain $O(1)$-approximation relative to $\optalg$.}

Under our notation, this translates to assuming that for every $i$, and $c\in\C_i$, we have
$c_i(\pset_i)\leq B$. This is a rather strong assumption, but we note that under this
assumption, it is quite easy to convert our results for XOS valuations to obtain
$O(1)$-approximation relative to $\optalg$. 
To see this, note that, under this assumption, choosing player $i^*=\argmax_i v(\pset_i)$,
and returning $\pset_{i^*}$ and paying $B$ to player $i^*$ yields a budget-feasible mechanism. 
Say we have a mechanism $\mech$ that obtains expected value at least 
$\optbench(\ell;v,B,c)/\al$,
for some $\al\geq 1$ and integer $\ell\geq 1$, where recall that 
$\optbench(\ell;v,B,c):=\max\,\{\vbengen{\ell}(S): S\sse\gset,\ c(S)\leq B\}$, which is at 
least $\optalg(v,B,c)-\ell\cdot v(\pset_{i^*})$; all our mechanisms have this type of
guarantee, for $\ell=O(1)$.
Then, we can obtain $(\al+\ell)$-approximation relative to $\optalg$ by simply running
$\mech$ with probability $\frac{\al}{\al+\ell}$, and the mechanism that returns
$\pset_{i^*}$ with probability $\frac{\ell}{\al+\ell}$. 
\end{remark}

\bibliographystyle{abbrv} 
\bibliography{refsmulti}

\appendix

\section{Results and proofs omitted from Section~\ref{impos}} \label{append-impos}

\begin{proofof}{Theorem~\ref{detlb}}
As mentioned earlier, this result was proved by~\cite{ChanC14}. We include a proof here
for completeness. 
For both results, we only require a single player.
For $\tht\in\R_+$, let $c^{(\tht)}$ denote the additive cost
function given by $c_e=\tht$ for all $e\in\gset$. 
We consider such cost-vectors with $\theta\leq B$, which satisfies the no-overbidding
assumption.

For part (a), consider the cost functions $c'=c^{(B)}$ and $c''=c^{(1)}$. 
On input $c'$, $\mech$ must return at least one element,
otherwise the statement holds for $c'$. Due to budget feasibility and IR, $\mech$ can
return at most one element, and so $\mech$ returns exactly one element. 
Now consider input $c''$. 
We now have $\optalg(c'')=n$. But we claim that $\mech$ must still output at most one 
element on input $c''$, which will prove part (a).  
If $\mech$ outputs more than one item under input $c''$, the player's utility when $c''$
is her true cost, is less than $B-1$, but her utility when she reports $c'$ would be $B-1$, 
contradicting truthfulness. 

\medskip
For part (b), let $n$ be a power of $2$. Let $I=\{0\}\cup[\log_2 n]$.
For $\ell\in I$, let $p_\ell$ be the expected payment made by $\mech$ 
and $a_\ell$ be the expected number of items procured by $\mech$, when the player reports 
the cost-function $c^{(2^\ell)}$. 
Since $\mech$ is truthful in expectation, the player's expected utility is maximized by
reporting truthfully.  
So for any $\ell,r \in [\log_2 n]$, we have 
$p_\ell-2^\ell a_\ell \ge p_r-2^\ell a_r$.  
Define $p_{\log_2 n+1}=a_{\log_2 n+1}=0$ for notational convenience.
Then, we have
\[ p_\ell-p_{\ell+1}\geq 2^\ell(a_\ell-a_{\ell+1}) \qquad \frall \ell\in I.\]
where the inequality for $\ell=\log_2 n$ follows from individual rationality.
Adding the above for all $\ell\in I$ 
we obtain that 
$p_0\geq a_0+\sum_{\ell=0}^{\log_2 n-1}(2^{\ell+1}-2^\ell)a_{\ell+1}\geq\sum_{\ell=0}^{\log_2 n}2^{\ell-1}a_\ell$. 
We have $p_0\leq B=n$ since $\mech$ is budget-feasible in expectation. So there is some
$\ell\in I$ for which $2^\ell a_\ell\leq 2n/(1+\log_2 n)$. But then 
on input $c^{(2^\ell)}$, since $\optalg(v,B,c^{(2^\ell)})=\frac{n}{2^\ell}$, $\mech$ obtains
value at most $\frac{2}{1+\log_2 n}\cdot\optalg$.
\end{proofof}

\subsection{Bayesian budget-feasible mechanism design.}
In the Bayesian setting, there is a publicly-known prior distribution $\D$ on $\C$ from
which players' private cost functions are drawn. (Note that players are not necessarily
independent.) A possibly randomized mechanism $\mech$ is {\em Bayesian budget-feasible}
if:  
\begin{enumerate}[label=(\alph*), topsep=0.1ex, noitemsep, leftmargin=*]
\item IR holds with probability $1$, where the
probability is over both $\D$, and any random choices made by $\mech$; and 
\item The expected payment made by $\mech$ on any input $c\in\C$ drawn from
$\D$ is at most the budget $B$; and 
\item It is {\em Bayesian incentive compatible} (BIC), which means that each player $i$
maximizes her expected utility by reporting truthfully, where the expectation is both over
the conditional distribution of other player's types given $i$'s true type, and any random
choices made by $\mech$: for every $i$, every $\bc_i,c'_i\in\C_i$, we 
have $\E[(c_i,c_{-i})\sim\D]{\util_i(\bc_i;\bc_i,c_{-i})\,|\,c_i=\bc_i}
\geq\E[(c_i,c_{-i})\sim\D]{\util_i(\bc_i;c'_i,c_{-i})\,|\,c_i=\bc_i}$. 

Note that if there is only one player, then BIC coincides with truthfulness.
\end{enumerate}
We say that $\mech$ achieves an $\al$-approximation relative to $\optalg$ for $(v,B)$,   
if $\E[c\sim\D]{v(f(c))}\geq\frac{1}{\al}\cdot\E[c\sim\D]{\optalg(v,B,c)}$. 
With these definitions in place, we now prove Corollary~\ref{bayesianlb}.

\begin{proofof}{Corollary~\ref{bayesianlb}}
We give two proofs. One is a black-box reduction showing that by Yao's minimax principle
(or equivalently, the min-max formula for two-person zero-sum games),
any approximation lower bound for budget-feasible-in-expectation mechanisms when there is
only one player also applies to Bayesian budget-feasible mechanisms. 
Second, one can easily adapt the lower-bound
constructions in Theorems~\ref{rand-overbid} and~\ref{truthexplb} (b) to apply to the
Bayesian setting.

\medskip\noindent
{\bf Using Yao's minimax principle.}\ 
The min-max formula for two-person zero-sum games states that for any $m\times n$ matrix
$A$, we have
$\max_{p\in\Dt_m}\min_{q\in\Dt_n}p^TAq=\min_{q\in\Dt_n}\max_{p\in\Dt_m}p^TAq$.

Suppose there is only one player, so that BIC coincides with truthfulness.
Fix some $(v,B)$.
Suppose that for every distribution $\D$ on $\C$, there is a Bayesian budget-feasible
mechanism that achieves approximation ratio $\al$ with respect to $\optalg$. 
Consider the two-person game, where the row player chooses an input $c\in\C$, and the
column player chooses a budget-feasible-in-expectation mechanism $\mech$, and the payoff
for $(c,\mech)$ is 
$\optalg(v,B,c)/\al-\E{\text{value returned by $\mech$ on $(v,B,c)$}}$. 
To make the strategy-space finite, we consider inputs $c$ and mechanisms
$\mech$ of some finite bit complexity. Since there is only one player, the strategy-space
of the column player consists of all Bayesian budget-feasible mechanisms.

In terms of
the above two-person game, this means that for every mixed-strategy of the row player,
there is a column-player strategy that yields payoff at most $0$. By the min-max formula 
for two-person zero-sum games, this implies that there is a mixed strategy of the column
player that guarantees that for every row-player (mixed) strategy, the payoff is at most
$0$. That is, there is some distribution over budget-feasible-in-expectation mechanisms,
which is another budget-feasible-in-expectation mechanism, that obtains value at least
$\optalg(v,B,c)/\al$ on every input, i.e., achieves approximation ratio at most $\al$
relative to $\optalg$ on every input. Hence, a lower bound on the approximation ratio
achievable by budget-feasible-in-expectation mechanisms yields the same lower bound for
Bayesian budget-feasible mechanisms.

\medskip\noindent
{\bf Adapting the lower bounds in Theorems~\ref{rand-lb} and~\ref{truthexplb} (b).}\
Recall the setup of the $n$-approximation lower bound in Theorem~\ref{rand-lb}. 
We have $\gm=1+\frac{n}{\epsilon}$, $\gset=[n]$,
$v$ is the additive valuation defined by $v(e)=\gm^{n-e}$ for all $e\in[n]$, the
budget $B$ is $1$, and
for each $\ell\in [n]$, $c^{(\ell)}$ is the additive cost function defined by 
$c^{(\ell)}_e=M\geq (1+n\gm^n)B$ for all $e\in[\ell-1]$, $c^{(\ell)}_\ell=1$, and
$c^{(\ell)}_e=0$ for all $e\in\{\ell+1,\ldots,n\}$. 
The distribution $\D$ chooses input $c^{(\ell)}$ with probability 
$K/\optalg(v,B,c^{(\ell)})$, where $K$ is a normalization constant; so
$\E[c\sim\D]{\optalg(v,B,c)}=nK$. 

The proof of Theorem~\ref{rand-lb} uses truthfulness, IR, and budget-feasibility to argue
that, letting \\
$q_\ell=\Pr\bigl[\mech\text{ returns a set containing item $\ell$ on input $c^{(\ell)}$}\bigr]$,
we have $\sum_{\ell\in[n]}q_\ell\leq 1$.
We then obtain that letting $\targ_\ell$ denote the expected value obtained by $\mech$ on
input $(v,B,c^{(\ell)})$, we have
\[
\targ_\ell\leq
q_\ell\cdot\optalg(v,B,c^{(\ell)})+(1-q_\ell)v(\{\ell+1,\ldots,n\})+\tfrac{\gm^{n-1}}{1+n\gm^n}\cdot
v(\ell).
\]
Multiplying the above by
$\frac{1}{\optalg(v,B,c^{(\ell)})}\leq\frac{1}{v(\ell)}$ 
and simplifying, we obtain that 
\begin{equation*}
\sum_{\ell\in[n]}\frac{\targ_\ell}{\optalg(v,B,c^{(\ell)})} \leq
\sum_{\ell\in[n]}\Bigl(q_\ell+\tfrac{(1-q_\ell)}{\gm-1}+\tfrac{\gm^{n-1}}{1+n\gm^n}\Bigr) 
\leq 1+\frac{n-1}{\gm-1}+\frac{1}{\gm}\leq 1+\frac{n}{\gm-1}\leq 1+\e.
\end{equation*}

\medskip\noindent
Next consider the setup of Theorem~\ref{truthexplb} (b), which considers $n$ to be a power
of $2$, the additive
valuation $v$ specified by $v(e)=1$ for all $e\in\gset$, budget $B=n$, and additive cost
functions $d^{(\ell)}$ specified by $d^{(\ell)}(e)=2^\ell$ for all $e\in\gset$, where
$\ell$ ranges in  $I=\{0\}\cup[\log_2 n]$. The arguments therein show that if $a_\ell$ is
the expected number of items procured by a Bayesian budget-feasible mechanism, then we
have $\sum_{\ell=0}^{\log_2 n} 2^{\ell-1}a_\ell\leq n$. Consider the distribution where
$d^{(\ell)}$ is chosen with probability $2^\ell/K$ for some normalization constant
$K$. The expected value of the optimum is then $\frac{n}{K}(1+\log_2 n)$, whereas the
expected value returned by the mechanism is at most 
$\sum_{\ell=0}^{\log_2 n}2^\ell a_\ell/K\leq\frac{2n}{K}$. 
\end{proofof}

\subsection{Proof of Theorem~\ref{optbenchlb}:
approximation-factor lower bounds relative to \boldmath \optbench}
As mentioned in Section \ref{impos}, we adapt the lower bound construction given by
\cite{AnariGN14} to our setting. 
They describe an instance with an additive valuation in the \emph{Bayesian} setting, that
is, the costs for the items are drawn from a distribution. They show that there is no
budget-feasible mechanism 
that achieves expected value strictly larger than
$(1-\frac{1}{e})\Ex[\optalg]$, where the expectation is over the distribution of the costs
and the randomness of the mechanism. (Recall that in our terminology, a budget-feasible
mechanism is also truthful.)

We will use the same instance as theirs, where each seller holds a single item.
In order to convert their lower bound into a lower bound for \optbench, we wish to utilize
the fact that $\vemax \ll \optalg$ in the large market setting.   
Since $\optalg \ge \optbench \ge \optalg-\vemax$, this implies that \optbench essentially
coincides with \optalg, and thus any lower bound with respect to \optalg is also a lower
bound with respect to \optbench. 

Such an argument will show that there is no budget-feasible mechanism that achieves 
expected value strictly larger than $(1-\frac{1}{e})\Ex[\optbench]$, where the expectation
is over the distribution of the costs and the randomness of the mechanism. 
However, one needs to be careful when arguing that this implies a lower bound in the worst-case setting.
Typically, one argues that, if a mechanism $\mathcal{M}$ achieves approximation ratio less
than $\alpha $ in the Bayesian setting, when given as input some distribution
$\mathcal{D}$ of the costs, then there exists some $c \in \supp(\mathcal{D})$ for which
$\mathcal{M}$ achieves approximation ratio less than $\alpha $ in the worst-case setting,
when given $c$ as input. 
However, for the distribution $\mathcal{D}$ used in the lower bound instance of Anari et
al., there exist cost vectors $c$ drawn from $\mathcal{D}$ for which $\optalg(c)$ is not
much larger than $\vemax$, and so $\optbench$ could be much smaller than $\optalg$.
Nonetheless, one can show that the distribution $\mathcal{D}$ satisfies $\vemax \ll
\optalg(c)$ with very high probability, and this suffices to show a lower bound in the
worst-case setting. 

Now, we describe the lower bound instance of \cite{AnariGN14}.
We are working in the single-dimensional setting, where each seller holds one item. 
So $n=|\gset|$ is the same as the number of players $k$.
The valuation function is additive, with $v(e)=1$ for every item $e$.
Consider the Bayesian instance $\mathcal{I}$, 
where the cost of each item is drawn independently from the distribution $\mathcal{D}$ whose 
CDF $F(x)$ is given by

\[F(x) = 
\begin{cases}
	\frac{1}{e(1-x)}, & \text{ for } 0 \le x \le 1-\frac{1}{e} \\
	1, & \text{ for } x > 1-\frac{1}{e} 
\end{cases}
\]

Note that the item costs i.i.d. 
Let $\mathcal{D}^k$ denote the joint distribution of costs over all the items.
Let $\overline{c} = \Ex_{x\sim \mathcal{D}}[x]$.
The budget $B$ of the instance is set to be $B = \overline{c}\cdot n = \Ex_{c\sim \mathcal{D}^k}[c(U)]$.
We use the following result proved in \cite{AnariGN14}.

\begin{lemma}[\cite{AnariGN14}] \label{anarilb} 
No budget-feasible mechanism can achieve value 
better than $(1-\frac{1}{e})\cdot\Ex_{c \sim \mathcal{D}^k}[\optalg]$, for the instance $\mathcal{I}$. 
\end{lemma}

We also need to make use of a Hoeffding bound, which is stated in the following lemma.

\begin{lemma} \label{hoeffding}
Let $x_1,\ldots,x_n$ be i.i.d. random variables whose values always lie in $[a,b]$.
Let $\mu  = \Ex[\sum_{i=1}^n x_i]$.
Then we have
\[\Pr\left[\sum_{i=1}^n x_i \ge (1+\epsilon )\cdot\mu \right] \le e^{-\frac{2 \epsilon ^2 \mu ^2}{n(b-a)^2}}\]
\end{lemma}

\begin{proof}[Finishing up the proof of Theorem~\ref{optbenchlb}]
Suppose that there exists a budget-feasible mechanism $\mathcal{M}$ for this instance that
achieves an $\alpha $-approximation with respect to $\optbench$. 
We run $\mathcal{M}$ on the instance $\mathcal{I}$.
	For $c \in \mathbb{R}^n_+$, let $\mathcal{M}(c)$ denote the (expected) value procured by $\mathcal{M}$ on instance $\mathcal{I}$ with costs $c$. 
	Since the value of each item is 1, and since each seller holds one item, we always have $\optbench = \optalg-1$ for the instance.
	So $\Ex_{c\sim \mathcal{D}^n} [\mathcal{M}(c)] \ge \alpha \Ex_{c\sim \mathcal{D}^n}[\optbench] = \alpha (\Ex_{c\sim \mathcal{D}^n}[\optalg]-1)$.
	Since the cost of each item is drawn i.i.d., by Lemma \ref{hoeffding}, we can show that the sum of the costs is concentrated around its mean, which is $B$. 
	In particular, for any constant $\epsilon >0$, we have,
	\[\Pr[c(U) \ge (1+\epsilon )B]\le e^{-\rho n}\]
	where $\rho := \rho (\epsilon )$ is some constant depending only on $\epsilon $.
	Thus, with probability $\ge 1-e^{-\rho n}$, we have that $c(U) < (1+\epsilon )B$.

	Suppose that the event $c(U) < (1+\epsilon )B$ occurs. Let $e_1,\ldots,e_k$ be the elements of $U$ in ascending order of their costs. Pick the smallest index $\ell$ such that $S=\{e_1,\ldots,e_\ell\}$ satisfies $c(S)\ge B$.
	We claim that $v(S) \ge \frac{n}{1+\epsilon }$.
	Indeed, since $S$ is the set of smallest total cost among all sets on $\ell$ elements, we have that the average cost of an element in $S$, which is $\frac{c(S)}{\ell}$, is at most the average cost, $\frac{c(U)}{n}$, over all the elements.
	Thus, $\frac{B}{\ell} \le \frac{c(S)}{\ell} \le \frac{c(U)}{n} < \frac{(1+\epsilon )B}{n}$ implying that $\ell > \frac{n}{1+\epsilon }$.
	Hence, the set $S' = S\setminus \{e_\ell\}$ satisfies $c(S')\le B$ and $v(S')\ge \frac{n}{1+\epsilon }-1$. 

	Thus, with probability $\ge 1-e^{-\rho n}$, we have that $\optalg \ge \frac{n}{(1+\epsilon )}-1$.
	Now the guarantee of $\mathcal{M}$ with respect to the algorithmic optimum is 
	\[\Ex_{c\sim \mathcal{D}^n} [\mathcal{M}(c)] \ge \alpha (\Ex_{c\sim \mathcal{D}^n}[\optalg]-1) \ge \alpha (1-e^{-\rho n})\left(\frac{n}{1+\epsilon }-2\right)\]
	For large enough $n$, we can pick $\epsilon '>0$ small enough so that we have $(1-e^{-\rho n})\left(\frac{n}{1+\epsilon }-2\right) \ge (1-\epsilon ')n$.
	Thus, $\Ex_{c\sim \mathcal{D}^n} [\mathcal{M}(c)] \ge (1-\epsilon ')n \ge (1-\epsilon ')\Ex_{c\sim \mathcal{D}^n}[\optalg]$.
	Hence $\mathcal{M}$ achieves a $(1-\epsilon ')\alpha $ approximation with respect to $\optalg$.
	By Lemma \ref{anarilb}, we must have that $(1-\epsilon ')\alpha \le 1-\frac{1}{e}$, showing that we cannot get a better than $\frac{1}{1-\epsilon '}(1-\frac{1}{e})$ approximation with respect to $\optbench$.
\end{proof}

\swamy{
\section{Proof of Lemma~\ref{gen-rpartition}} \label{append-rpartition}
We mimic the proof of Lemma~\ref{rpartition}.
Let $I$ be a minimal prefix of $[k]$ such that, letting 
$A_1=\bigcup_{i\in I}(S\cap\pset_i)$, we have
$\vbengen[g]{2\ell}(A_1)\geq\frac{\vbengen[g]{(4\ell+1)}(S)}{2}$. Let $A_2=S-A_1$. 
Let $r$ be the last index in $I$.  
Then, we have $\vbengen[g]{2\ell}(A_1-\pset_r)<\frac{\vbengen[g]{(4\ell+1)}(S)}{2}$. 
Let $J_1\sse[k]$, $J_2\sse[k]$ with $|J_1|,|J_2|\leq 2\ell$ be such that
$\vbengen[g]{2\ell}(A_1-\pset_r)=g\bigl(A_1-\pset_r-\bigcup_{i\in J_1}\pset_i\bigr)$ and
$\vbengen[g]{2\ell}(A_2)=g\bigl(A_2-\bigcup_{i\in J_2}\pset_i\bigr)$ and
Let $S'=S-\bigcup_{i\in J_1\cup J_2}\pset_i-\pset_r$. Since 
$|J_1\cup J_2\cup\{r\}|\leq 4\ell+1$, we have $g(S')\geq\vbengen[g]{(4\ell+1)}(S)$.
Also, $A_2-\bigcup_{i\in J_2}\pset_i=S'-(A_1-\pset_r-\bigcup_{i\in J_1}\pset_i)$, and so
by subadditivity of $g$, we have 
\[
\vbengen[g]{2\ell}(A_2)=
g\Bigl(A_2-\bigcup_{i\in J_2}\pset_i\Bigr)\geq 
g(S')-g\Bigl(A_1-\pset_r-\bigcup_{i\in J_1}\pset_i\Bigr)
>\frac{\vbengen[g]{(4\ell+1)}(S)}{2}.
\]

Now fix a partition $A'_1, A''_1$ of $A_1$, and a partition $A'_2, A''_2$ of $A_2$.
Observe that
$\vbengen[g]{\ell}(A'_1)+\vbengen[g]{\ell}(A''_1)\geq\vbengen[g]{2\ell}(A_1)$.
So some set $A_1^H\in\{A'_1, A''_1\}$ satisfies
$\vbengen[g]{\ell}(A_1^H)\geq\vbengen[g]{2\ell}(A_1)/2$; let $A_1^L$ be the other set in
$\{A'_1,A''_1\}$. 
Similarly, one of $A'_2,A''_2$, denoted $A_2^H$ satisfies 
and $\vbengen[g]{\ell}(A^H_2)\geq\vbengen[g]{2\ell}(A_2)/2$; let $A_2^L$ be the other set
in $\{A'_2,A''_2\}$.

Now we proceed exactly as in the proof of Lemma~\ref{rpartition}.
The random partition $\gset_1,\gset_2$ induces random partitions of $A_1$ and $A_2$. 
Consider the event $\Gm$ that for both $\ell=1,2$, the random partition of $A_\ell$ induced
by $\gset_1,\gset_2$ is the same as the partition $A_\ell^H, A_\ell^L$, up to permutations
of the parts. 
For any $j=1,2$, we have 
$\Pr\bigl[\gset_j\cap A_1=A_1^L,\ \gset_j\cap A_2=A_2^L\,|\,\Gm\bigr]=\frac{1}{4}$. So
conditioned on $\Gm$, with probability at least $\frac{1}{2}$, we have that both
$\gset_1,\gset_2$ contain some big set. Removing the conditioning completes the proof.
\qed}

\end{document}